\documentclass[11pt]{article} %

\usepackage{amsthm,amssymb} %
\usepackage{fullpage}
\usepackage{booktabs} %
\usepackage{dsfont} %
\usepackage{enumitem} %
\usepackage{float} %
\usepackage[margin=1in]{geometry} %
\usepackage{graphicx} %
\usepackage{hyperref} %
\usepackage[utf8]{inputenc} %
\usepackage{mdframed} %
\usepackage{microtype} %
\usepackage{charter} %
\usepackage{soul} %
\usepackage{suffix} %
\usepackage{tikz} %
\usepackage{xcolor} %
\usepackage{xspace} %
 \usepackage{multirow}
\usepackage[scr=rsfso]{mathalfa} %
\usepackage{amsmath}
\usepackage{cleveref}
\usepackage{thm-restate}

\usepackage{subcaption}

\usepackage{qcircuit}

\hypersetup{pagebackref, colorlinks=true, urlcolor=blue,
  linkcolor=blue, citecolor=magenta}

\usetikzlibrary{calc,positioning}

\mdfdefinestyle{figstyle}{ %
  linecolor=white, %
  backgroundcolor=white, %
  innertopmargin=10pt, %
  innerleftmargin=25pt, %
  innerrightmargin=25pt, %
  innerbottommargin=10pt %
}

\definecolor{White}{rgb}{1,1,1} %
\definecolor{Black}{rgb}{0,0,0} %
\definecolor{LightGray}{rgb}{.8,.8,.8} %
\colorlet{ChannelColor}{LightGray} %
\colorlet{ChannelTextColor}{Black} %
\colorlet{ReadoutColor}{White} %

\newcommand{\cE}{\mathcal E}
\newcommand{\cU}{\mathcal U}
\newcommand{\cC}{\mathcal C}
\newcommand{\cA}{\mathcal A}
\newcommand{\cP}{\mathcal P}

\newcommand{\cS}{\mathcal S}

\newcommand{\cH}{\mathcal H} %

\newcommand{\N}{\mathbb{N}}

\newcommand{\sA}{{\mathsf{A}}}
\newcommand{\sB}{{\mathsf{B}}}
\newcommand{\sC}{{\mathsf{C}}}

\newcommand{\sF}{{\mathsf{F}}}

\newcommand{\sM}{{\mathsf{M}}}
\newcommand{\sN}{{\mathsf{N}}}
\newcommand{\sO}{{\mathsf{O}}}
\newcommand{\sP}{{\mathsf{P}}}
\newcommand{\sR}{{\mathsf{R}}}

\newcommand{\sV}{{\mathsf{V}}}

\newcommand{\innerCode}{\cC_{inner}}
\newcommand{\outerCode}{\cC_{outer}}
\newcommand{\innerEnc}{\mathsf{Enc}_{inner}}
\newcommand{\outerEnc}{\mathsf{Enc}_{outer}}
\newcommand{\Enc}{\mathsf{Enc}}

\newcommand{\unary}{\mathsf{unary}}
\newcommand{\Toffoli}{\mathsf{Toffoli}}
\newcommand{\Magic}{\mathsf{Magic}}

\DeclareMathOperator{\Steane}{Steane}
\newcommand{\mcC}{\mathcal{C}}
\newcommand{\mcG}{\mathcal{G}}
\newcommand{\mcP}{\mathcal{P}}
\newcommand{\mcS}{\mathcal{S}}
\newcommand{\mcU}{\mathcal{U}}
\newcommand{\C}{\mathbb{C}}
\newcommand{\arr}{\rightarrow}
\newcommand{\tr}{\Tr}

\newcommand{\Sim}{\mathrm{Sim}}
\newcommand{\SimSnapshot}{\mathrm{SimSnapshot}}
\newcommand{\SimDensity}{\mathrm{SimDensity}}
\newcommand{\SimInterval}{\mathrm{SimInterval}}
\newcommand{\View}{\mathrm{View}}

\newcommand{\NEEXP}{\class{NEEXP}}

\newcommand{\SZKMIP}{\class{SZK\text{-}MIP}}

\newcommand{\PZKMIP}{\class{PZK\text{-}MIP}}

\renewcommand{\hat}[1]{\widehat{#1}} %
\renewcommand{\tilde}[1]{\widetilde{#1}} %

\newcommand{\Id}{\mathds{I}}
\newcommand{\comp}[1]{\overline{#1}}

\newcommand{\reg}[1]{{\textsf{#1}}}

\newcommand{\advR}{\adv{R}}
\newcommand{\advRna}{\adv{R}^{na}}

\newcommand{\SimR}{\Sim_{\advR}}
\newcommand{\SimRna}{\Sim_{\advR^{na}}}

\newcommand{\complex}{\mathbb{C}} %
\let\epsilon=\varepsilon %
\newcommand{\eps}{\epsilon} %

\newcommand{\microspace}{\mspace{.5mu}} %
\newcommand{\ket}[1]{\ensuremath{\lvert\microspace #1
    \microspace\rangle}} %
\newcommand{\bra}[1]{\ensuremath{\langle\microspace #1
    \microspace\rvert}} %
\newcommand{\ketbra}[2]{\lvert #1 \rangle \! \langle #2 \rvert} %
\newcommand{\kb}[1]{\ketbra{#1}{#1}} %

\newcommand{\Paren}[1]{\left(#1\right)}

\newcommand{\Bigbrac}[1]{\Big[#1\Big]}
\newcommand{\set}[1]{\{#1\}}

\newcommand{\comment}[1]{}

\newcommand{\class}[1]{\mathsf{#1}\xspace} %

\newcommand{\NEXP}{\class{NEXP}} %
\newcommand{\QMA}{\class{QMA}} %
\newcommand{\QMIP}{\class{QMIP}} %
\WithSuffix\newcommand\QMIP*{\ensuremath{\class{QMIP}^*}} %
\newcommand{\MIP}{\class{MIP}} %
\WithSuffix\newcommand\MIP*{\ensuremath{\class{MIP}^*}} %

\newcommand{\NTIME}{\class{NTIME}}

\newtheorem{theorem}{Theorem} %
\newtheorem{lemma}[theorem]{Lemma} %
\newtheorem{proposition}[theorem]{Proposition} %
\newtheorem{corollary}[theorem]{Corollary} %
\newtheorem{definition}[theorem]{Definition} %
\newtheorem{claim}[theorem]{Claim} %

\newcommand{\adv}[1]{\hat{#1}}

\DeclareMathOperator{\Tr}{Tr}

\DeclareMathOperator{\poly}{poly}

\title{Perfect zero knowledge for \\ quantum multiprover interactive proofs}
\author{Alex B. Grilo~\thanks{CWI and QuSoft, Amsterdam, The Netherlands.
alexg@cwi.nl} \and
William Slofstra~\thanks{IQC and Department of Pure Mathematics, University of Waterloo, Waterloo, Canada. weslofst@uwaterloo.ca} \and
Henry Yuen~\thanks{University of Toronto, Toronto, Canada. hyuen@cs.toronto.edu}}
\date{}

\begin{document}

\maketitle

\begin{abstract}
In this work we consider the interplay between multiprover interactive proofs,
  quantum entanglement, and zero knowledge proofs --- notions that are central
  pillars of complexity theory, quantum information and cryptography. In particular, we study the relationship between the complexity class $\MIP^*$, the set of languages decidable by multiprover interactive proofs with quantumly entangled provers, and the class $\PZKMIP^*$, which is the set of languages decidable by $\MIP^*$ protocols that furthermore possess the \emph{perfect zero knowledge} property. 

Our main result is that the two classes are equal, i.e., $\MIP^* = \PZKMIP^*$. This result provides a quantum analogue of the celebrated result of Ben-Or, Goldwasser, Kilian, and Wigderson (STOC 1988) who show that $\MIP = \PZKMIP$ (in other words, all classical multiprover interactive protocols can be made zero knowledge). We prove our result by showing that every $\MIP^*$ protocol can be efficiently transformed into an equivalent zero knowledge $\MIP^*$ protocol in a manner that preserves the completeness-soundness gap. Combining our transformation with previous results by Slofstra (Forum of Mathematics, Pi 2019) and Fitzsimons, Ji, Vidick and Yuen (STOC 2019), we obtain the corollary that all co-recursively enumerable languages (which include undecidable problems as well as all decidable problems) have zero knowledge $\MIP^*$ protocols with vanishing promise gap.

\end{abstract}

\newpage

\section{Introduction}

Multiprover interactive proofs (MIPs) are a model of
computation where a probabilistic polynomial time verifier interacts with several
all-powerful --- but non-communicating --- provers to check the validity of a statement (for example, whether a quantified boolean formula is satisfiable). If the statement is true,
then there is a strategy for the provers to convince the verifier of this fact.
Otherwise, for all prover strategies, the verifier
rejects with high probability. This gives rise to the complexity class $\MIP$,
which is the set of all languages that can be decided by MIPs. This model of computation was first introduced by Ben-Or,
Goldwasser, Kilian and Wigderson~\cite{BenOrGKW88}. A foundational result in complexity theory due to Babai, Fortnow, and Lund shows
that multiprover interactive proofs are surprisingly powerful: $\MIP$ is
actually equal to the class of problems solvable in non-deterministic
exponential time, i.e., $\MIP = \NEXP$~\cite{BabaiFL91}.

Research in quantum complexity theory has led to the study of \emph{quantum}
MIPs. In one of the most commonly considered models, the verifier interacts with
provers that are \emph{quantumly entangled}. Even though the provers still
cannot communicate with each other, they can utilize correlations arising from
local measurements on entangled quantum states. Such correlations cannot be
explained classically, and the study of the counter-intuitive nature of these
correlations dates back to the famous 1935 paper of Einstein, Podolsky and
Rosen~\cite{EPR35} and the seminal work of Bell in 1964~\cite{Bell64}. Over the
past twenty years, MIPs with entangled provers have provided a fruitful
computational lens through which the power of such correlations can be studied. The set of languages decidable by such interactive proofs is denoted by $\MIP^*$, where the asterisk denotes the use of entanglement.

Finally, another type of interactive proof system are zero knowledge proofs. These were introduced by Goldwasser, Micali and Rackoff~\cite{GoldwasserMR89} and have played a crucial role in the development of theoretical cryptography. In this model, if the claimed statement is indeed true, the interaction between the verifier and prover  must be conducted in such a way that the verifier learns \emph{nothing else aside from the validity of the statement}.
This is formalized by requiring the existence of an efficient {\em simulator}
whose output is indistinguishable from the distribution of the messages in a real execution of the protocol. It was shown by~\cite{BenOrGKW88} that any (classical) MIP protocol can be transformed into an equivalent perfect zero knowledge\footnote{The term {\em perfect}
refers to the property that the interaction in a real protocol can be simulated without any error.} MIP protocol. In other words, the complexity classes $\MIP$ (and thus $\NEXP$) and $\PZKMIP$ are equal, where the latter consists of all languages decidable by perfect zero knowledge MIPs.

Informally stated, our main result is a quantum analogue of the result of Ben-Or, Goldwasser, Kilian, and Wigderson~\cite{BenOrGKW88}: we show that
\begin{center}
	\emph{Every MIP* protocol can be efficiently transformed into an equivalent zero knowledge MIP* protocol.}
\end{center}
Phrased in complexity-theoretic terms, we show that $\MIP^* = \PZKMIP^*$. This is a strengthening of the recent results of Chiesa, Forbes, Gur and Spooner, who show that $\NEXP = \MIP \subseteq \PZKMIP^*$~\cite{ChiesaFGS18} (which is, in turn, a strengthening of the the result of Ito and Vidick that $\NEXP \subseteq \MIP^*$~\cite{ItoV12}). %

Surprisingly, there are no upper bounds known on the power of quantum MIPs. The recent spectacular result of Natarajan and Wright shows that $\MIP^*$ contains the complexity class $\NEEXP$, which is the enormously powerful class of problems that can be solved in non-deterministic \emph{doubly exponential} time~\cite{natarajan2019neexp}. Since $\NEXP \neq \NEEXP$ via the non-deterministic time hierarchy theorem~\cite{cook1973hierarchy}, this unconditionally shows that quantum MIPs are strictly more powerful than classical MIPs. Furthermore, it is conceivable that $\MIP^*$ even contains \emph{undecidable} languages. In~\cite{slofstra2016tsirelson,slofstra2019set}, Slofstra proved that determining
whether a given MIP* protocol admits a prover strategy that wins with certainty
is an undecidable problem. In~\cite{FitzsimonsJVY18}, Fitzsimons, Ji, Vidick and Yuen showed that the class $\MIP^*_{1,1-\eps(n)}$, the set of languages decidable by MIPs protocols with \emph{promise gap} $\eps(n)$ that can depend on the input size, contains $\NTIME[2^{\poly (1/\eps(n))}]$, the class of problems that are solvable in non-deterministic time $2^{\poly  (1/\eps(n))}$. In contrast, the complexity of $\MIP$ (even with a shrinking promise gap) is always equal to $\NEXP$.

Thus, our result implies that all languages in $\NEEXP$ -- and any larger complexity classes discovered to be contained within $\MIP^*$ -- have perfect zero knowledge interactive proofs with entangled provers.
In fact, we prove a stronger statement:
every $\MIP^*$ protocol with promise gap $\eps$ also has an equivalent zero knowledge $\MIP^*$ protocol
with promise gap that is polynomially related to $\eps$. This, combined with the
results of~\cite{FitzsimonsJVY18} and~\cite{slofstra2019set}, implies that languages of arbitrarily large time complexity -- including some undecidable problems -- have zero knowledge proofs (albeit with vanishing promise gap). 

\subsection{Our results}
\label{sec:results}
We state our results in more detail. Let $\MIP^*_{c,s}[k,r]$ denote the set of languages $L$ that admit $k$-prover, $r$-round MIP* protocols with completeness $c$, and soundness $s$. In other words, there exists a probabilistic polynomial-time verifier $V$ that interacts with $k$ entangled provers over $r$ rounds so that if $x \in L$, then there exists a prover strategy that causes $V(x)$ to accept with probability at least $c$; \footnote{Technically speaking, the completeness condition actually corresponds to a \emph{sequence} of prover strategies with success probability approaching $c$; we discuss this subtlety in Section~\ref{sec:prelim-zk}.} otherwise all prover strategies cause $V(x)$ to accept with probability strictly less than $s$. The class $\PZKMIP^*_{c,s}[k,r]$ are the languages that have $\MIP^*[k,r]$ protocols where the interaction between the verifier can be simulated exactly and efficiently, without the aid of any provers. We provide formal definitions of
these complexity classes in Section~\ref{sec:complexity}.

In what follows, let $n$ denote the input size. The parameters $k,r,s$ of a protocol are also allowed to depend on the input size. In this paper, unless stated otherwise, we assume that completeness parameter $c$ in a protocol is equal to $1$.

\begin{restatable}{theorem}{mainthm}
  \label{thm:main}
For all $0 \leq s \leq 1$, for all polynomially bounded functions $k, r$, 
\[
	\MIP^*_{1,s}[k,r] \subseteq \PZKMIP^*_{1,s'}[k + 4,1]
\]
where $s' = 1 - (1-s)^\alpha$ for some universal constant $\alpha > 0$.
\end{restatable}
The first corollary of Theorem~\ref{thm:main} concerns what we call \emph{fully quantum} MIPs, which are multiprover interactive proofs where the verifier can perform polynomial time quantum computations and exchange quantum messages with entangled quantum provers. The set of languages decidable by fully quantum MIPs is denoted by $\QMIP$, which clearly includes $\MIP^*$. Reichardt, Unger, and Vazirani~\cite{ReichardtUV13} showed that the reverse inclusion also holds by adding two additional provers; i.e., that $\QMIP[k] \subseteq \MIP^*[k+2]$. Combined with Theorem~\ref{thm:main} and the fact that we can assume that $\QMIP$ protocols have perfect completeness if we add an additional prover (see~\cite{vidick2016quantum}), this implies that
\begin{corollary}
  For all polynomially bounded functions $k, r$, we have 
  \[
  	\QMIP_{1,\frac{1}{2}}[k,r] \subseteq \PZKMIP_{1,\frac{1}{2}}^*[k+4,1].
\]
\end{corollary}

The combination of the results in~\cite{FitzsimonsJVY18}
and~\cite{natarajan2019neexp} implies that
for every
hyper-exponential function $f$,\footnote{A hyper-exponential function $f(n)$ is of the form $\exp(\cdots \exp(\poly(n)) \cdots)$, where the number of iterated exponentials is $R(n)$ for some time-constructible function $R(n)$.} we have that
\[\NTIME[2^{2^{f(n)}}] \subseteq \MIP^*_{1,s}[4,1],\]
where $\NTIME[g(n)]$ denotes the set of languages that can be decided by
nondeterministic Turing machines running in time $g(n)$ and $s =
1-Cf(n)^{-c}$ for some universal constants $C$ and $c$, independent of $n$.\footnote{The original result in~\cite{FitzsimonsJVY18} states that for all hyper-exponential functions $f(n)$, $\NTIME[2^{f(n)}] \subseteq \MIP^*_{1,s}[15,1]$ for $s = 1 = C f(n)^{-c}$. Using a more efficient error correcting code as described in Section~\ref{sec:differences}, the number of provers can be reduced to $4$. The improvement from $\NTIME[2^{f(n)}]$ to $\NTIME[2^{2^{f(n)}}]$ is obtained by plugging in the $\NEEXP \subseteq \MIP^*$ result of Natarajan and Wright~\cite{natarajan2019neexp} as the ``base case'' of the iterated compression scheme, instead of the $\NEXP \subseteq \MIP^*$ result of Natarajan and Vidick~\cite{NatarajanV18a}.}
Combining this with \Cref{thm:main}, we obtain the following.

\begin{corollary}
\label{cor:iterated}
	There exist universal constants $C, c > 0$ such that for all hyper-exponential functions $f: \N \to \N$,
	\[
    \NTIME[2^{2^{f(n)}}] \subseteq \PZKMIP^*_{1,s}[6,1]
	\]
  where $s =  1 - C f(n)^{-c}$.
\end{corollary}

Finally, it was also shown in~\cite{FitzsimonsJVY18,slofstra2019set} that the undecidable language $\mathrm{NONHALT}$, which consists of Turing machines that do not halt when run on the empty input tape, is contained in $\MIP^*_{1,1}[2,1]$. The ``$1,1$'' subscript
indicates that for negative instances (i.e., Turing machines that do halt), the
verifier rejects with positive probability. In more detail: there exists a polynomial time computable function that maps Turing machines $M$ to an $\MIP^*$ protocol $V_M$ such that if $M$ does not halt on the empty input tape, then there is a prover strategy for $V_M$ that is accepted with probability $1$; otherwise there exists a positive constant $\eps > 0$ (depending on $M$) such that for all prover strategies, the protocol $V_M$ rejects with probability $\eps$. 

Theorem~\ref{thm:main} implies there is a polynomial time computable mapping $V_M \mapsto V_M'$ such that $V_M'$ is a $\PZKMIP^*$ protocol that preserves completeness (if $V_M$ accepts with probability $1$, then so does $V_M'$) and soundness (if $V_M$ rejects with probability $\eps$ for all prover strategies, then $V_M'$ rejects with probability $\poly(\eps)$ for all prover strategies). Therefore, we can conclude the following:
\begin{corollary}
\label{cor:nonhalt}
	$\mathrm{NONHALT} \in \PZKMIP^*_{1,1}[4,1]$.
\end{corollary}
\noindent Corollary~\ref{cor:nonhalt} implies that all co-recursively enumerable languages (languages whose complement are recursively enumerable) have zero knowledge proofs (with vanishing gap).

\subsection{Proof overview}
\label{sec:overview}

The proof of Theorem~\ref{thm:main} draws upon a number of ideas and techniques that have been developed to study interactive protocols with entangled provers. At a high level, the proof proceeds as follows. Let $L$ be a language that is decided by some $k$-prover MIP* protocol with a verifier $V$. Assume for simplicity that on positive instances $x \in L$, there is a prover strategy that causes $V$ to accept with probability $1$, and otherwise rejects with high probability. Although $V$ is probabilistic polynomial time (PPT) Turing machine in a MIP* protocol, we can instead think of it as a quantum circuit involving a combination of verifier computations, and prover computations. 

First, we transform the verifier $V$ into an equivalent quantum circuit $V_{enc}$ where the computation is now performed on \emph{encoded data}. We do this using techniques from quantum fault-tolerance, where the data is protected using a quantum error correcting code, and physical operations are performed on the encoded data in order to effect logical operations on the underlying logical data. 

We then apply \emph{protocol compression} to $V_{enc}$ to obtain a new verifier $V_{ZK}$ for an equivalent protocol --- this will be our zero knowledge MIP* protocol. Protocol compression is a technique that was pioneered by Ji in~\cite{Ji17} (and further developed by Fitzsimons, Ji, Vidick and Yuen~\cite{FitzsimonsJVY18}) to show that $\NEXP$ has $1$-round $\MIP^*$ protocols where the communication is logarithmic length. Essentially, in the compressed protocol, the new verifier $V_{ZK}$ efficiently checks whether $V_{enc}$ would have accepted in the original protocol without actually having to run $V_{enc}$, by testing that the provers hold an entangled \emph{history state} of a successful interaction between $V_{enc}$ and some provers.

The reason this compressed protocol is zero knowledge is the following: the verifier $V_{ZK}$ asks the provers to report the outcomes of performing local measurements in order to verify that they hold an accepting history state. In the positive case (i.e., $x \in L$), there is an ``honest'' strategy where the provers share a history state $\ket{\Phi}$ of a successful interaction with $V_{enc}$. We argue that, because of the fault-tolerance properties of $V_{enc}$, individual local measurements on $\ket{\Phi}$ reveal no information about the details of the interaction. Put another way, the distribution of outcomes of honest provers' local measurements can be efficiently simulated, without the aid of any provers at all. Since we only require that this simulatability property holds with respect to honest provers, this establishes the zero knowledge property of the protocol run by $V_{ZK}$.

In the next few sections, we provide more details on the components of this transformation. We discuss things in reverse order: first, we give an overview of the protocol compression technique. Then, we discuss the fault tolerant encoding $V_{enc}$ of the original verifier $V$. Then we describe how applying protocol compression to $V_{enc}$ yields a zero knowledge protocol for $L$.

\subsubsection{Protocol compression}\label{sec:intro-compression}

The protocol compression technique of~\cite{Ji17,FitzsimonsJVY18} transforms any $k$-prover, $r$-round QMIP protocol where the verifier $V$ runs in time $N$ into a $k+O(1)$-prover, $1$-round MIP* protocol where the verifier $V'$ runs in time $\poly \log N$. In other words, the verifier has been compressed into an exponentially more efficient one; however, this comes with the price of having the promise gap shrink as well: if the promise gap of the original QMIP protocol is $\eps$, then the promise gap of the compressed protocol is $\poly(\eps/N)$. 

This compression is achieved as follows: in the protocol executed by the compressed verifier $V'$, the provers are tested to show that they possess an (encoding of) a \emph{history state} of the original protocol executed by $V$, describing an execution of the protocol in which the original verifier $V$ accepts. History states of some $T$-length computation generally look like the following:
\begin{equation*}
	\ket{\psi} = \frac{1}{\sqrt{T+1}} \sum_{t = 0}^T \ket{t} \otimes \ket{\psi_t}.
\end{equation*}
The first register holding the superposition over $\ket{t}$ is called the \emph{clock register}; the second register holding the superposition over $\ket{\psi_t}$ is called the \emph{snapshot register}. The $t$-th snapshot of the computation $\ket{\psi_t}$ is the global state of the protocol at time $t$:
\[
	\ket{\psi_t} = g_t g_{t-1} \cdots g_1 \ket{\psi_0}
\]
where the $g_i$'s are the gates used in the protocol, and $\ket{\psi_0}$ is the initial state of the protocol. Usually, each $g_i$ is a one- or two-qubit gate that is part of the verifier $V$'s computation. However, $g_i$ could also represent a \emph{prover gate}, which is the computation performed by one of the $k$ provers. Unlike gates in the verifier's computation, the prover gates are non-local, and there is no characterization of their structure. In general, they may have exponential circuit complexity, and may act on a Hilbert space that can be much larger than the space used by the verifier $V$.

This notion of history states for interactive protocols is a generalization of
the basic concept of history states for quantum circuits, which was introduced
by Kitaev to prove that the local Hamiltonians problem is
$\QMA$-complete~\cite{Kitaev02}. He showed that for every $\QMA$ verifier circuit $C$, there exists a local Hamiltonian $H(C)$ (called the Feynman-Kitaev Hamiltonian) such that all ground states of $H(C)$ are history states of the circuit $C$. To test whether a given state $\ket{\phi}$ is a history state of $C$, one can sample random terms from $H(C)$ and measure them to get an estimate of the energy of $\ket{\phi}$ with respect to $H(C)$. 

In slightly more detail, the local Hamiltonian $H(C)$ consists of terms that can be divided into four groups:
\begin{itemize}
	\item \emph{Input checking terms} $H_{in}$. These terms check that the initial snapshot $\ket{\psi_0}$, which represents the initial state of the $\QMA$ verifier, has all of its ancilla bits set to zero.
	
	\item \emph{Clock checking terms} $H_{clock}$. These terms check that the clock register is encoded in unary. The unary encoding is to ensure that the locality of $H(C)$ is a fixed constant independent of the computation.
	
	\item \emph{Propagation terms} $H_{prop}$. These terms check that the history state is a uniform superposition over snapshots $\ket{\psi_t}$, with $\ket{\psi_t} = g_t \ket{\psi_{t-1}}$.
	
	\item \emph{Output checking terms} $H_{out}$. These terms check that at time $t = T$, the decision bit of the $\QMA$ verifier is equal to $\ket{1}$ (i.e., the verifier accepted).
\end{itemize}

In~\cite{Ji17}, Ji showed that for every quantum interactive protocol $\Pi$, there is a \emph{generalized protocol Hamiltonian} $H(\Pi)$ whose ground states are all history states of $\Pi$. The Hamiltonian $H(\Pi)$ is essentially the Feynman-Kitaev Hamiltonian corresponding to the verifier $V$, except if at time $t$ in the protocol $\Pi$, prover $i$ is supposed to implement a unitary $g_t$ on their registers (which includes their private registers as well as some registers used to communicate with the verifier), then there will be a corresponding non-local propagation term
\begin{equation}
\label{eq:prop_term}
	\frac{1}{2} \Paren{ \ket{t-1} \otimes I - \ket{t} \otimes g_t}\Paren{ \bra{t-1} \otimes I - \bra{t} \otimes g_t^\dagger}.
\end{equation}
This term is non-local because of the prover gate $g_t$, which may act on a Hilbert space of unbounded size. Other than these prover propagation terms, the rest of $H(\Pi)$ corresponds to the local computations performed by the verifier $V$. 

Suppose that one had the ability to sample random terms of $H(\Pi)$ and efficiently measure a given state with the terms. Then, by performing an energy test on a state $\ket{\psi}$, one could efficiently determine whether the state was close to a history state that describes an accepting interaction in the protocol $\Pi$. This appears to be a difficult task for terms like~\eqref{eq:prop_term} when $g_t$ is a prover gate, since this requires performing a complex non-local measurement. Furthermore, the tester would not know what prover strategy to use.

Ji's insight in~\cite{Ji17} was that a tester could efficiently \emph{delegate} the energy measurements to entangled quantum provers. He constructs a protocol where the verifier $V'$ commands the provers to perform measurements corresponding to random terms of $H(\Pi)$ on their shared state. If the reported energy is low, then $V'$ is convinced that there must exist a history state of $\Pi$ that describes an accepting interaction (and in particular, the provers share this history state).  

In order to successfully command the provers, the verifier $V'$ relies on a
phenomena called \emph{non-local game rigidity} (also known as
\emph{self-testing}). Non-local games are one-round protocols between a
classical verifier and multiple entangled provers. This phenomena is best
explained using the famous CHSH game, which is a two-player game where the
optimal entangled strategy succeeds with probability $\omega^*(CHSH) =
\frac{1}{2} + \frac{1}{\sqrt{2}}$. The canonical, textbook strategy for CHSH is
simple: the two players share a maximally entangled pair of qubits, and measure
their respective qubits using the Pauli observables $\sigma_X$ and $\sigma_Z$, depending on their input. The rigidity property of the CHSH game implies that this canonical strategy is, in some sense, \emph{unique}: \emph{any} optimal entangled strategy for CHSH must be, up to a local basis change, identical to this canonical strategy. Thus we also say that the CHSH game is a \emph{self-test} for a maximally entangled pair of qubits and single-qubit Pauli measurements for the players.

There has been extensive research on rigidity of non-local
games~\cite{ReichardtUV13,CoudronN16,McKague2017,Coladangelo17,Chao2018,NatarajanV17,NatarajanV18,ColadangeloGJV19}, and many different self-tests have been developed. The non-local games used in the compression protocols of~\cite{Ji17,FitzsimonsJVY18} are variants of the CHSH game, where the canonical optimal strategy is roughly the following: the players share a maximally entangled state on $n$ qubits, and their measurements are tensor products of Pauli observables on a constant number of those $n$ qubits, such as
\[
	\sigma_X(i) \otimes \sigma_Z(j) \otimes \sigma_Z(k).
\]
which indicates $\sigma_X$ acting on the $i$'th qubit, and $\sigma_Z$ on the $j$'th
and $k$'th. This game also has the following robust self-testing guarantee: any
entangled strategy that succeeds with probability $1- \eps$ must be
$\poly(\eps,n)$-close to the canonical strategy. Here, $n$ is a growing parameter, whereas the weight of the Pauli observables (i.e. the number of factors that don't act as the identity) is at most some constant independent of $n$.

For the terms of $H(\Pi)$ that involve uncharacterized prover gates, the verifier $V'$ simply asks some provers to measure the observable corresponding to the prover gate. By carefully interleaving rigidity tests with the energy tests, the verifier $V'$ can ensure that the provers are performing the desired measurements for all other terms of $H(\Pi)$, and thus test if they have an accepting history state.

\subsubsection{Quantum error correction and fault tolerant verifiers}
\label{sec:qecc_intro}

In order to describe our fault tolerant encoding of verifiers, we first discuss quantum error correction and fault tolerant quantum computation. 

Quantum error correcting codes (QECCs) provide a way of encoding quantum information in a form that is resilient to noise. Specifically, a $[[n,k,d]]$ quantum code $\cC$ encodes all $k$-qubit states $\ket{\psi}$ into an $n$-qubit state $\Enc(\ket{\psi})$ such that for any quantum operation $\cE$ that acts on at most $(d-1)/2$ qubits, the original state $\ket{\psi}$ can be recovered from $\cE(\Enc(\ket{\psi}))$. The parameter $d$ is known as the \emph{distance} of the code $\cC$. %

QECCs are an important component of \emph{fault tolerant} quantum computation, which is a method for performing quantum computations in a way that is resilient to noise. In a fault tolerant quantum computation, the information $\ket{\psi}$ of a quantum computer is encoded into a state $\Enc(\ket{\psi})$ using some QECC $\cC$, and the computation operations are performed on the encoded data without ever fully decoding the state. 

For example, in many stabilizer QECCs, in order to compute $\Enc(g \ket{\psi})$ for some single-qubit Clifford gate $g$, it suffices to apply $g$ \emph{transversally}, i.e., apply $g$ on every physical qubit of $\Enc(\ket{\psi})$. Transversal operations are highly desirable in fault tolerant quantum computation because they spread errors in a controlled fashion.

Non-Clifford gates, however, do not admit a transversal encoding in most stabilizer QECCs. In order to implement logical non-Clifford gates, one can use {\em magic states}. These are states that encode the behaviour of some non-Clifford gate $g$ (such as a Toffoli gate, or a $\pi/8$ rotation), and are prepared and encoded before the computation begins. During the fault tolerant computation, the encoded magic states are used in \emph{gadgets} that effectively apply the non-Clifford $g$ to the encoded data. These gadgets only require measurements and transversal Clifford operations that are controlled on the classical measurement outcomes.

We now discuss the behaviour of the verifier $V_{enc}$. First, the encoded
verifier spends time manufacturing a collection of encoded ancilla states, as
well as encoded magic states of some non-Clifford gates (in our case, the
Toffoli gate), using some fixed quantum error correcting code $\cC$. We call
this the \textbf{Resource Generation Phase}. Then, the verifier $V_{enc}$
simulates the execution of $V$ on the encoded information from the Resource
Generation Phase. All Clifford operations of $V$ are performed transversally,
and non-Clifford operations of $V$ are performed with the help of the encoded
magic states. When interacting with the provers, the verifier $V_{enc}$ sends
its messages in encoded form as well -- the provers are capable of decoding and re-encoding messages using the code $\cC$. 

Finally, after the finishing the simulation of $V$, the verifier $V_{enc}$ executes an \textbf{Output Decoding Phase}: it performs a decoding procedure on the physical qubits corresponding to the output qubit of $V$.

It is clear that the protocol executed by $V_{enc}$ is equivalent to the protocol executed by $V$. The overhead introduced by this fault tolerant encoding is a constant factor increase in the length of the circuit (depending on the size of the code $\cC$). The fault tolerant properties of the computation of $V_{enc}$ will play a major role in our proof of zero knowledge.

\subsubsection{The zero knowledge protocol, and its analysis}
To distinguish between the parties of the ``inner'' protocol executed by $V_{enc}$ and the parties in the ``outer'' protocol executed by $V_{ZK}$, we say that $V_{enc}$ is a \emph{verifier} that interacts with a number of \emph{provers}. On the other hand, we say that $V_{ZK}$ is a \emph{referee} that interacts with a number of \emph{players}.

The zero knowledge protocol executed by $V_{ZK}$ consists of applying protocol compression to the fault tolerant verifier $V_{enc}$. The result is a MIP* protocol that checks whether the players possess a history state of an accepting interaction with $V_{enc}$ and some provers.

The formal definition of the zero knowledge property requires an efficient algorithm, called the simulator, that when given a yes instance (i.e., $x \in L$), produces an output that is identically distributed to the transcript produced by an interaction between the referee and players following a specified \emph{honest strategy}.
The interaction must be simulatable even when the referee doesn't follow the protocol.
A cheating referee could, for instance, sample questions differently than the honest referee, or interact with the
players in a different order. The only constraint we have is that the \emph{format} of the questions, from the perspective of an individual player, must look like something the honest referee \emph{could have} sent. In particular, if a cheating referee tries to interact with an individual player multiple times, the player would abort the protocol.

In the yes instance, the honest player strategy for $V_{ZK}$ consists of sharing a history state $\ket{\Phi}$ that describes the referee $V_{enc}$ interacting with some provers and accepting with probability $1$. When the players receive a question in $V_{ZK}$, they either measure some Pauli observable on a constant number of qubits of $\ket{\Phi}$, or measure the observable corresponding to a prover gate. The zero knowledge property of $V_{ZK}$ rests on the ability to efficiently sample the outcomes of measurements formed from \emph{any} combination of local Pauli observables and prover measurements that might be commanded by a cheating referee.

We first analyze \emph{non-adaptive} referees; that is, they sample the questions to all the players first. In the compressed protocol $V_{ZK}$, the honest referee asks the players to perform local measurements corresponding to a random term in the the Hamiltonian $H(\Pi)$. Thus, the support of the measurements commanded by a referee (even a cheating one) can only involve a constant number of qubits of $\ket{\Phi}$. Let $\adv{W}$ denote the tuple of questions sent to the players, and let $S_{\adv{W}}$ denote the registers of $\ket{\Phi}$ that are supposed to be measured. We argue that the reduced density matrix $\ket{\Phi}$ on the registers $S_{\adv{W}}$ can be computed explicitly in polynomial time.

This is where the fault tolerance properties of $V_{enc}$ come in. Since $V_{enc}$ is running a computation on encoded information, any local view of the state of $V_{enc}$ in the middle of its computation should not reveal any details about the actual information being processed. Intuitively, the purpose of a quantum error correcting code is to conceal information from an environment that is making local measurements. In the zero knowledge context, we can think of the cheating referee as the ``noisy environment'' to $V_{enc}$. Thus, the cheating referee should not be able to learn anything because it can only access local correlations, while all the ``juicy'' information about $V_{enc}$ is encoded in global correlations of $\ket{\Phi}$.

Although this is the high level idea behind our proof, there are several challenges that need to be overcome in order to make this argument work. First, the state of $V_{enc}$ is not always properly encoded in an error correcting code: it may be in the middle of some logical operations, so there is a risk that some information may be leaked. We argue that if the code used by $V_{enc}$ is \emph{simulatable} (see Section~\ref{sec:qecc}), then this cannot happen. We show that the concatenated Steane code is simulatable, by analyzing coherent implementations of logical operations that do not reveal any information.

The next challenge is that the referee is able to perform local measurements not only on intermediate states of $V_{enc}$ during its computation, but also \emph{superpositions} of them. This threatens to circumvent the concealing properties of the error correcting code, because of the following example: suppose that $\ket{\psi_0}$ and $\ket{\psi_1}$ are orthogonal $n$ qubit states such that the reduced density matrix of every small-sized subset of qubits of $\ket{\psi_0}$ or $\ket{\psi_1}$ looks maximally mixed. However, $\frac{1}{\sqrt{2}} \Paren{\ket{0} \ket{\psi_0} + \ket{1} \ket{\psi_1}}$ can be distinguished from $\frac{1}{\sqrt{2}} (\ket{0} + \ket{1}) \ket{\psi_0}$ via a local measurement (namely, an $\sigma_X$ measurement on the first qubit). One potential worry is that $\ket{\psi_0}$ and $\ket{\psi_1}$ might represent snapshots of the history state $\ket{\Phi}$ that are separated by many time steps, and therefore a simulator would have trouble simulating measurements on these superpositions, because it will not be able to determine what the inner product between $\ket{\psi_0}$ and $\ket{\psi_1}$ is in general. 

We argue that, because of the structure of the protocol and the honest strategy, the cheating referee can only measure a superpositions that involve only constantly many consecutive snapshots of $V_{enc}$. From this we deduce that reduced density matrices of the superpositions can be efficiently computed. 

Another challenge involves simulating the outcomes of measuring the prover gate, which may perform some arbitrarily complex computation. We carefully design the honest strategy for the compressed protocol so that measurement outcomes of the prover gate are always either constant, or an unbiased coin flip.

Finally, we argue that we can efficiently simulate the interaction of the protocol
even when the referee behaves adaptively. The simulator for the non-adaptive case actually computes the reduced density matrix of the honest players' state; we can perform post-selection on the density matrix at most a polynomial number of times in order to simulate the distribution of questions and answers between an adaptive referee and the provers.

\subsection{Related work}
In this section, we discuss some relevant work on quantum analogues of zero knowledge proofs.

In quantum information theory, zero knowledge proofs have been primarily studied in the context of \emph{single prover} quantum interactive proofs. This setting was first formalized by Watrous~\cite{watrous2002limits}, and has been an active area of research over the years. Various aspects of zero knowledge quantum interactive proofs have been studied, including honest verifier models~\cite{watrous2002limits,ChaillouxK08}, computational zero knowledge proof systems for $\QMA$~\cite{BroadbentJ0W16}, and more.

In the multiprover setting, Chiesa, Forbes, Gur and
Spooner~\cite{ChiesaFGS18} showed that all problems in $\NEXP$ (and thus $\MIP$) are in
$\PZKMIP^*[2,\poly(n)]$. Their approach is considerably different of ours. They
achieve their result by showing that model of interactive proofs called \emph{algebraic interactive PCPs} \footnote{An interactive PCP is a protocol where the verifier and a single prover first commit to an oracle, which the verifier can query a bounded number of times.  Then, the verifier and prover engage in an interactive proof. An
\emph{algebraic} interactive PCP is one where the committed oracle has a desired
algebraic structure.  We refer to~\cite{ChiesaFGS18} for an in-depth discussion of these models.} can be lifted to the entangled provers setting in a way that preserves zero knowledge, and then showing that languages in $\NEXP$ have zero knowledge algebraic interactive PCPs.

The results of~\cite{ChiesaFGS18} are, strictly speaking, incomparable to ours. We show that all languages in $\MIP^*$ have single-round $\PZKMIP^*$ protocols with four additional provers, whereas~\cite{ChiesaFGS18} show that $\MIP$ (which is a subset of $\MIP^*$) have $\PZKMIP^*$ protocols with \emph{two} provers and polynomially many rounds. Improving our result to only two provers seems to be quite a daunting challenge, as it is not even known how $\MIP^*[k]$ relates to $\MIP^*[k+1]$ -- it could potentially be the case that adding more entangled provers yields a strictly larger complexity class! 

Furthermore, the proof techniques of~\cite{ChiesaFGS18} are very different from ours: they heavily rely on algebraic PCP techniques, as well as the analysis of the low degree test against entangled provers~\cite{NatarajanV18a}. Our proof relies on techniques from fault tolerant quantum computing and the protocol compression procedure of~\cite{Ji17,FitzsimonsJVY18}, which in turn rely heavily on self-testing and history state Hamiltonians. 

Another qualitative difference between the zero knowledge protocol of~\cite{ChiesaFGS18} and ours is that the honest prover strategy for their protocol does not require any entanglement; the provers can behave classically. In our protocol, however, the provers are required to use entanglement; this is what enables the class $\MIP^*$ and $\PZKMIP^*$ to contain classes beyond $\NEXP$, such as $\NEEXP$ (and beyond).

Recently, Kinoshita~\cite{Kinoshita19} showed that a model of ``honest-verifier'' zero knowledge QMIP can be lifted to general zero knowledge QMIP protocols. He also shows that $\QMIP$ have interactive proofs with \emph{computational} zero knowledge proofs under a computational assumption.

Coudron and Slofstra prove a similar result to \cite{FitzsimonsJVY18}
for multiprover proofs with commuting operator strategies, showing that this
class also contains languages of arbitrarily large time complexity, if the
promise gap is allowed to be arbitrarily small \cite{CoudronS19}. Their
results (achieved via a completely different method from ours) also show that
there are two-prover zero knowledge proofs for languages of arbitrarily large time complexity,
albeit in the commuting operator model and with a quantitatively worse lower
bound than Corollary \ref{cor:iterated}.

Finally,  Cr{\'{e}}peau and Yang~\cite{CrepeauY18} refined the notion of zero knowledge,
requiring the simulator to be local, i.e., that there are non-communicating classical
simulators that simulate the (joint) output distribution of the provers.  We note that our result does not fulfill this modified definition, and we leave it as an open problem (dis)proving
that all $\MIP^*$ can be made zero knowledge in this setting.

\subsection*{Organization} 
The paper is organized as follows. We start with some
preliminaries in~\Cref{sec:prelim}. Then, in \Cref{sec:protocol}, we present our transformation on $\MIP^*$ protocols. In \Cref{sec:zk}, we prove the
zero knowledge property of the transformed protocol. In Section~\ref{sec:simulability}, we prove that the concatenated Steane code is simulatable.

\subsection*{Acknowledgments}
AG thanks Thomas Vidick for discussions on related topics. WS thanks Matt Coudron, David Gosset, and Jon Yard for helpful discussions.
AG is supported by ERC Consolidator Grant 615307-QPROGRESS. WS is supported by NSERC DG 2018-03968.

\section{Preliminaries}
\label{sec:prelim}

\subsection{Notation}
We denote $[n]$ as the set $\{1,...,n\}$. 
We assume that all Hilbert spaces are finite-dimensional.  An $n$-qubit binary observable (also called a reflection) $O$ is a Hermitian matrix with 
$\pm 1$ eigenvalues. 

We use the terminology ``quantum register'' to name specific quantum
systems. We use sans-serif font to denote registers, such as $\sA$, $\sB$.
For example, ``register $\reg{A}$'', to which is implicitly associated
the Hilbert space $\mathcal{H}_{\reg{A}}$.

For a density matrix $\rho$ defined on some registers $\sR_1 \cdots \sR_n$, and a subset $S$ of those registers, we write $\Tr_S (\rho)$ to denote the partial trace of $\rho$ over those registers in $S$. We write $\Tr_{\comp{S}}(\rho)$ to denote tracing out all registers of $\rho$ \emph{except} for the registers in $S$.

Let $\sigma_I,\sigma_X,\sigma_Y,\sigma_Z$ denote the four single-qubit
Pauli observables
\[
  \sigma_I = \begin{pmatrix}
    1 & 0 \\
    0 & 1 \\
  \end{pmatrix}\;, \qquad \sigma_X = \begin{pmatrix}
    0 & 1 \\
    1 & 0 \\
  \end{pmatrix}\;, \qquad \sigma_Y = \begin{pmatrix}
    0 & -i \\
    i & 0 \\
  \end{pmatrix}\;, \qquad \sigma_Z = \begin{pmatrix}
    1 & 0 \\
    0 & -1 \\
  \end{pmatrix}\;.
\]
We let $\cP_n$ denote the $n$-qubit Pauli group, so $\cP_n$ is the set of $n$-qubit unitaries $W_1 \otimes \cdots \otimes W_n$ where $W_i \in  \{ \pm \sigma_I, \pm i \sigma_I, \pm \sigma_X, \pm i \sigma_X, \pm \sigma_Y, \pm i \sigma_Y, \pm \sigma_Z, \pm i \sigma_Z \}$. %

We use two ways of specifying a Pauli observable acting on a specific
qubit.
\begin{enumerate}
\item Let $W \in \{I,X,Z\}$ be a label and let $\sR$ be a
  single-qubit register.
  We write $\sigma_W(\sR)$ to denote the observable $\sigma_W$ acting
  on $\sR$.
\item Let $\sR$ be an $n$-qubit register, and let $i \in
  \{1,\ldots,n\}$.
  Let $W = X_i$ (resp.
  $W = Z_i$).
  We write $\sigma_W$ to denote the $\sigma_X$ (resp.
  $\sigma_Z$) operator acting on the $i$-th qubit in $\sR$ (the
  register $\sR$ is implicit).
\end{enumerate}
We also use $W$ to label Pauli operators that have higher ``weight''.
For example, for $W = X_i Z_j$ the operator $\sigma_W$ denotes the
tensor product $\sigma_{X_i} \otimes \sigma_{Z_j}$.

\paragraph{Universal set of gates} A universal set of gates is $\set{ H,
\Lambda(X),\Lambda^2(X)}$, where $H$ is the Hadamard gate, $\Lambda(X)$ is the controlled-$X$ gate (also known as the CNOT gate), and $\Lambda^2(X)$ is the Toffoli
gate~\cite{Aharonov03}. 

\subsection{Error correcting codes} 
\label{sec:qecc}

Quantum error correcting codes (QECCs) provide a way of encoding quantum
information in a form that is resilient to noise. Specifically, a $[[n,k]]$
quantum code $\cC$ encodes all $k$-qubit states $\ket{\psi}$ into an $n$-qubit
state $\Enc(\ket{\psi})$. We say that a $[[n,k]]$ QECC has distance $d$
if for any quantum operation $\cE$ that acts on
at most $(d-1)/2$ qubits, the original state $\ket{\psi}$ can be recovered from
$\cE(\Enc(\ket{\psi}))$. In this case, we say that $\cC$ is a $[[n,k,d]]$ QECC.

Throughout this paper, we mostly use  codes that encode $1$ logical qubit into some number of physical qubits. If $\Enc$ is the encoding map of an $[[m,1]]$ QECC $\cC$ and $\ket{\phi}$ is an $n$-qubit state, then we overload notation and write $\Enc(\ket{\phi})$ to denote the $mn$ qubit state obtained from applying $\Enc$ to every qubit of $\ket{\phi}$. 
We refer to the qubits of $\ket{\phi}$ as \emph{logical qubits}, and the qubits of the encoded state $\Enc(\ket{\phi})$ as \emph{physical qubits}. We call any state $\ket{\psi}$ in the code $\cC$ a \emph{codeword}.

Given two QECCs $\mcC_1$ and $\mcC_2$, the concatenated code $\mcC_1\circ
\mcC_2$ is defined by setting $\Enc_{\mcC_1 \circ\mcC_2}(\rho) =
\Enc_{\mcC_2}(\Enc_{\mcC_1}(\rho))$, i.e. to encode $\rho$ in the concatenated
code, we first encode it using $\mcC_1$, and then encode every physical qubit
of $\Enc_{\mcC_1}(\rho)$ using $\mcC_2$.

\subsubsection{Inner and outer codes}

In our zero knowledge transformation, we use quantum error correcting codes in two different ways. One use, as described in the proof overview in Section~\ref{sec:overview}, is in the transformation from the original MIP* verifier $V$ to a fault-tolerant version $V_{enc}$. We call the error correcting code used in the fault tolerant construction the \emph{inner code}, denoted by $\innerCode$. 

The other use of quantum error correcting codes is in the protocol compression of $V_{enc}$ into the zero knowledge protocol $V_{ZK}$. In Section~\ref{sec:overview}, we described the protocol $V_{ZK}$ as testing whether the players share a history state $\ket{\Phi}$ of the protocol corresponding to $V_{enc}$. Actually, the protocol tests whether the players share an \emph{encoding} of the history state. The qubits of the history state $\ket{\Phi}$ corresponding to the state of the verifier $V_{enc}$ are supposed to be encoded using another error correcting code and distributed to multiple players (see Section~\ref{sec:honest} for more details). For this, we use what we call the \emph{outer code}, denoted by $\outerCode$.

\paragraph{The outer code} For the outer code $\outerCode$, we require a stabilizer code that satisfies the following properties~\cite{FitzsimonsJVY18}:
\begin{enumerate}
	\item For every qubit $i$, there exists a logical $X$ and $Z$ operator that acts trivially on that qubit. 
	
	\item The code can correct one erasure in a known location.
\end{enumerate}
The following four-qubit error detection code satisfies both properties~\cite{grassl1997codes}. 
\begin{align*}
	\ket{0} \mapsto \frac{1}{\sqrt{2}} \Paren{\ket{0000} + \ket{1111}} \qquad \ket{1} \mapsto \frac{1}{\sqrt{2}} \Paren{\ket{1001} + \ket{0110}}.
\end{align*}
The stabilizer generators for this code are $XXXX, ZIIZ, IZZI$. A set of logical operators for this code are $XIIX, IXXI, ZZII, IIZZ$. 
We use $\outerEnc$ to denote the encoding map for the outer code $\outerCode$.

\paragraph{The inner code} For the inner code $\innerCode$, we use the
concatenated Steane code $\Steane^K$ for some sufficiently large (but constant) $K$. 
We use $\innerEnc$ to denote the encoding
map for the outer code $\innerCode$. We describe the concatenated Steane code in
more detail in Section~\ref{sec:concatenated-steane}.

\subsubsection{Encodings of gates and simulatable codes}

An important concept in our work is that of \emph{simulatable codes}. The motivation for this concept is the observation that for a distance $d$ code $\cC$, the reduced density matrix of any codeword $\ket{\psi} \in \cC$ on fewer than $d-1$ qubits is a state that is independent of $\ket{\psi}$, and only depends on the code $\cC$. We generalize this indistinguishability notion to the context of fault tolerant encodings of gates with a QECC: informally, a QECC is simulatable if ``small width'' reduced density matrices of codewords $\ket{\psi}$ \emph{in the middle} of a logical operation are independent of $\ket{\psi}$. Intuitively, simulatability is a necessary condition for fault tolerant quantum computation; if local views of an in-progress quantum computation are dependent on the logical data, then environmental noise can corrupt the computation.

Let $U$ be a $k$-qubit gate.  If $\underline{a} = (a_1,\ldots,a_k)$ is a
$k$-tuple of distinct numbers between $1$ and $n$, we let $U(\underline{a})$ be
the gate $U$ applied to qubits $(a_1,\ldots,a_k)$.  If $\rho$ is an $n$-qubit
state, then $U(\underline{a}) \rho U(\underline{a})^\dagger$ is the result of
applying $U$ to $\rho$ in qubits $a_1,\ldots,a_k$.

 An encoding of a $k$-qubit gate $U$ in the code $\mcC$ is a way to transform $\Enc(\rho)$
to $\Enc(U(\underline{a}) \rho U(\underline{a})^\dagger)$ by applying operations on
the physical qubits, sometimes with an additional ancilla state used as a
resource. More formally, an \emph{encoding of a $k$-qubit $U$ in code $\mcC$} is a pair of states
$\sigma_U$ and $\sigma_U'$, and a number $\ell \geq 1$, along with a mapping
from $k$-tuples $\underline{a}$ of distinct physical qubits to sequences of
unitaries $O_1(\underline{a}),\ldots,O_{\ell}(\underline{a})$ such that
\begin{equation*}
    (O_\ell(\underline{a}) \cdots O_1(\underline{a})) \left(\Enc(\rho) \otimes \sigma_U\right) (O_{\ell}(\underline{a}) \cdots O_1(\underline{a}))^\dagger 
        = \Enc(U(\underline{a}) \rho U(\underline{a})^\dagger) \otimes \sigma'_U,
\end{equation*}
where (in a slight abuse of notation) the unitaries
$O_1(\underline{a}),\ldots,O_{\ell}(\underline{a})$ act only on the physical
qubits corresponding to logical qubits $a_1,\ldots,a_k$, as well as the ancilla
register holding $\sigma_U$. In this definition, the sequence
$O_1(\underline{a}),\ldots,O_{\ell}(\underline{a})$ depends on $\underline{a}$. However, in practice
$\underline{a}$ is only used to determine which physical qubits the gates
$O_1(\underline{a}),\ldots,O_{\ell}(\underline{a})$ act on, and otherwise the
sequence depends strictly on $U$. We say that an encoding \emph{uses physical
gates $\mcG$} if for every $\underline{a}$, the unitaries
$O_1(\underline{a}),\ldots,O_{\ell}(\underline{a})$ are gates in $\mcG$. 

If a QECC $\mcC$ can correct arbitrary errors on $s$ qubits,
then the partial trace $\tr_{\overline{S}}(\Enc(\rho))$ is independent of the
state $\rho$ for every set of physical qubits $S$ with $|S| \leq s$. If we start
with an encoded state $\Enc(\rho)$, and apply an encoded logical operation $U$
to some $k$-tuple of qubits $\underline{a}$, then we start in state $\Enc(\rho)
\otimes \sigma_U$ and end in state $\Enc(U(\underline{a}) \rho
U(\underline{a})^\dagger) \otimes \sigma'_U$. So as long as we can compute the
partial traces of $\sigma_U$ and $\sigma'_U$, then we can compute
$\tr_{\overline{S}}(\Enc(\rho))$ both before and after the operation. However,
the encoded operation is made of up a sequence of gates, and while we are
in the middle of applying these gates, the system might not be in an encoded
state. We say that an encoding is $s$-simulatable if we can still compute
the reduced density matrices on up to $s$ qubits of the state at any point during the encoding of
$U$. The following definition formalizes this notion:
\begin{definition}\label{D:simulatable}
    An encoding $(\sigma_U,\sigma'_U,\ell,O_1(\underline{a}),\ldots,
    O_\ell(\underline{a}))$ of a $k$-qubit gate in a QECC $\mcC$ is
    \emph{$s$-simulatable} if for all integers $0 \leq t \leq \ell$, $n$-qubit
    states $\rho$, and subsets $S$ of the physical qubits of $\Enc(\rho)
    \otimes \sigma_U$ with $|S| \leq s$, the partial trace 
    \begin{equation*}
        \tr_{\overline{S}}((O_t(\underline{a}) \cdots O_1(\underline{a})) \Enc(\rho) \otimes \sigma_U (O_{t}(\underline{a}) \cdots O_1(\underline{a}))^\dagger)
    \end{equation*}
    can be computed in polynomial time from $t$, $\underline{a}$, and
    $S$. In particular, the partial trace is independent of $\rho$.
\end{definition}
In our applications, $s$ will be constant. We also consider only a finite number
of gates $U$, and since $t$ is bounded in any given encoding, $t$ will also
be constant.  The partial trace in the above definition will be a $2^{|S|} \times
2^{|S|}$ matrix, where $|S| \leq s$. So when we say that the partial trace can be computed in polynomial time in Definition~\ref{D:simulatable}, we mean that the
entries of this matrix are rational, and can be computed explicitly in
polynomial time from $\underline{a}$, $S$, and $t$. %

A crucial component of our zero knowledge arguments is the notion of simulatable codes. We state now the theorem we will use to prove
zero knowledge. The proof is deferred to \Cref{sec:simulatable}.
\begin{theorem}\label{T:simulatable}
    Let $\mcU = \{H,\Lambda(X), \Lambda^2(X)\}$. For every constant $s$,
    there exists a $[[n,1]]$ QECC $\mathcal{C}$ where $n$ is constant, such that 
    $\mathcal{C}$ has $s$-simulatable encodings of $\mcU$ using
    only $\cU$ as physical gates.
\end{theorem}

If a code $\mcC$ admits a simulatable encoding of a gate $U$, then, applying
Definition \ref{D:simulatable} with $t=0$, we see that it must be possible to
compute the partial trace $\tr_{\overline{S}}(\Enc(\rho) \otimes \sigma_U)$ for
any set of physical qubits $S$ with $|S| \leq s$, with no knowledge of $\rho$.
In particular, it must be possible to compute partial traces of $\Enc(\rho)$ on
all but $s$ qubits. We must also be able to compute the partial traces of the
ancilla states $\sigma_U$ and (setting $t=\ell$) $\sigma_U'$, although this is
easier in principle, since we have full knowledge of these states.

\subsection{Quantum interactive protocols}
\label{sec:complexity}

We first define the notion of a \emph{protocol circuit}, which is a quantum circuit representation an interaction between a quantum verifier and one or more provers. A protocol circuit $C$ with $k$ provers and $r$ rounds is specified by a tuple $(n,m,\Gamma)$ where $n,m$ are positive integers and $\Gamma$ is a sequence of gates $(g_1,g_2,\ldots)$. This tuple is interpreted in the following manner. The circuit $C$ acts on these registers:
\begin{enumerate}
	\item A set of \emph{prover} registers $\sP_1,\ldots,\sP_k$.
	\item A set of \emph{message} registers $\sM_1,\ldots,\sM_k$; each register $\sM_i$ consists of $m$ qubits. The $j$'th qubit of register $\sM_i$ is denoted by $\sM_{ij}$.
	\item A \emph{verifier} register $\sV$ which consists of $n$ qubits. The $j$'th qubit of register $\sV$ is denoted by $\sV_j$.
\end{enumerate}
Each gate $g_i$ consists of a \emph{gate type}, and the label of the registers that the gate acts on. There are two gate types:
\begin{enumerate}
	\item A gate from a universal gate set (such as Hadamard, CNOT, and Toffoli), which can only act on registers $\sV, \sM_1,\ldots,\sM_k$.
	\item A \emph{prover gate} $P_{ij}$, which represents the $i$'th prover's unitary in round $j$. The prover gate $P_{ij}$ can only act on registers $\sP_i \sM_i$. 
\end{enumerate}
Furthermore, prover $i$'s gates $\{P_{ij}\}$ must appear in order; for example,
$P_{i2}$ can only appear in the circuit after $P_{i1}$ has appeared. A prover
gate $P_{ij}$ cannot appear twice in the circuit with the same label. 

Intuitively, a protocol circuit describes an interaction between a verifier and $k$ provers where the verifier performs a computation on the workspace register $\sV$, and communicates with the provers through the message registers $\{\sM_i\}$, and the provers carry out their computations on the registers $\{\sP_i \sM_i\}$. The verifier's workspace $\sV$ is initialized in the all zeroes state, and the $\{\sP_i \sM_i \}$ registers are initialized in some entangled state $\ket{\psi}$ chosen by the provers. At the end of the protocol circuit, the first qubit of the workspace register $\sV$ is measured in the standard basis to determine whether the verifier accepts or rejects. 

A \emph{prover strategy} $\cS$ for a protocol circuit $C$ is specified by a tuple $(d,\{P_{ij}\},\ket{\psi})$ where $d$ is a positive integer, a set of unitary operators $P_{ij}$ for $i=1,\ldots,k$ and $j = 1,\ldots,r$ that act on $\complex^d \otimes (\complex^2)^{\otimes m}$, and pure states $\ket{\psi}$ in $(\complex^d)^{\otimes k} \otimes (\complex^2)^{\otimes mk}$. Given a protocol circuit $C$, we write $\omega^*(C)$ to denote the supremum of acceptance probabilities of the verifier over all possible prover strategies $\cS$.

\vspace{20pt}

We now define the complexity class $\QMIP$, which stands for \emph{quantum multiprover interactive proofs}. This is the set of all languages $L$ that can be decided by a quantum interactive protocol with at most polynomially many provers, at most polynomially-many rounds, and polynomial-sized protocol circuits, whose gates are drawn from the gate set $\{H,\Lambda(X),\Lambda^2(X) \}$.

\begin{definition}
A promise problem $L = (L_{yes},L_{no})$ is in the complexity class $\QMIP_{c,s}[k,r]$ if and only if there exists a polynomial-time computable function $V$ with the following properties:
\begin{enumerate}
	\item For every $x \in L_{yes} \cup L_{no}$, the output of $V$ on input $x$ is a description of a $k$-prover, $r$-round prover circuit $V(x) = (n,m,\Gamma)$ where $n,m = \poly(|x|)$.
	
	\item \emph{Completeness}. For every $x \in L_{yes}$, it holds that $\omega^*(V(x)) \geq c$.
	
	\item \emph{Soundness}. For every $x \in L_{no}$, it holds that $\omega^*(V(x)) < s$.
\end{enumerate}
Furthermore, we say that $L$ has a $\QMIP_{c,s}[k,r]$ protocol $V$.
\end{definition}
Throughout this paper, we interchangeably refer to $V(x)$ as the protocol circuit, the protocol, or the verifier that is executing the protocol, depending on the context. 

We note that in the negative case (i.e. $x \in L_{no}$), we require that the entangled value of $V(x)$ is \emph{strictly} less than $s$. This allows us to meaningfully talk about ``zero promise gap'' classes such as $\QMIP_{1,1}[k,r]$, where in the Completeness case, the verifier has to accept with probability $1$, whereas in the Soundness case, the verifier has to reject with some positive probability. Finally, we follow the convention that $\QMIP[k,r]$ is defined as $\QMIP_{\frac{2}{3},\frac{1}{3}}[k,r]$.

We also define the class $\MIP*$, which is defined in the same way as $\QMIP$ except that the protocol is specified as a classical interaction between a randomized verifier (modelled as a probabilistic polynomial-time Turing machine) and quantum provers. Since the verifier is classical, the communication between the verifier and provers can be treated as classical. Thus, in a $k$-prover $\MIP^*$ protocol, we can equivalently talk about \emph{measurement prover strategies} $\cS$, where the $k$ provers share an entangled state $\ket{\psi} \in \cH^{\otimes k}$ for some Hilbert space $\cH$. In each round of the protocol, each prover receives a classical message from the verifier, and performs a measurement on their share of $\ket{\psi}$ that depends on the verifier's message as well as the previous messages exchanged between that prover and the verifier (but not the communication with the other provers). 

We call prover strategies for a general $\QMIP$ protocol as \emph{unitary strategies}, to distinguish them from measurement strategies for $\MIP^*$ protocols. Furthermore, when we speak of an $\MIP^*$ protocol $V$, we are referring to the verifier for the protocol (which is some probabilistic Turing machine).

\subsection{Zero knowledge $\MIP^*$}
\label{sec:prelim-zk}

First, we define the \emph{view} of an interaction between a classical, randomized verifier $\hat{V}$ and a set of $k$ provers that behave according to some strategy $\cS$, as might occur in an $\MIP^*$ protocol. The view is a random variable $\View(\hat{V}(x) \leftrightarrow \cS)$ which is the tuple $(x,r,m_1,m_2,\ldots,m_{2r})$ where $x$ is the input to $\hat{V}$, $r$ is the randomness used by $\hat{V}$, and the $m_i$'s are the messages between the provers and verifier. 

Next, we present the definition of zero knowledge $\MIP^*$ protocols, first
defined by~\cite{CleveHTW04}. We use the abbreviation ``PPT'' to denote ``probabilistic polynomial-time.''

\begin{definition}
	An $\MIP^*_{c,s}[k,r]$ protocol $V$ for a promise language $L = (L_{yes},L_{no})$ is \emph{statistically zero knowledge} if for all $x \in L_{yes}$, there exists a prover strategy $\cS$ (called the \emph{honest strategy}) satisfying the following properties:
	\begin{enumerate}
		\item The strategy $\cS$ is accepted by the protocol $V(x)$ with probability at least $c$, 
		\item For all PPT verifiers $\hat{V}$, there exists a PPT simulator $\Sim_{\adv{V}}$ such that the output distribution of $\Sim_{\adv{V}}(x)$ is $\eps(n)$-close in total variation distance to $\View(\hat{V}(x) \leftrightarrow \cS)$, for some negligible function $\eps(n)$.

	\end{enumerate}
	Furthermore, the complexity class $\SZKMIP^*_{c,s}[k,r]$ is the set of languages that have statistical zero knowledge proof systems.
\end{definition}

When a language can be decided by a zero knowledge proof system with closeness $\eps(n) = 0$, we say that it admits a \emph{perfect zero knowledge} proof system. In other words, the interaction can be simulated exactly. We let $\PZKMIP^*_{c,s}[k,r]$ denote languages that admit perfect zero knowledge $\MIP^*$ protocols. %

\paragraph{Some subtleties} We address two subtleties regarding the definitions of $\QMIP$ and $\PZKMIP^*$. 
\begin{enumerate}
	\item The definition of $\QMIP$ depends on our choice of gate set. If we allow the verifier circuits to use arbitrary single- and two-qubit gates, then our perfect zero knowledge results may not hold; however, we will still get the statistical zero knowledge property with exponentially small error.
	
	\item In a $\PZKMIP^*_{c,s}[k,r]$ protocol $V$, there may be no strategy $\cS$ for the provers that gets accepted with probability $c$ exactly. Instead, there may be a sequence of strategies whose success probability converges to $c$. Thus, in order for $\PZKMIP^*_{c,s}[k,r]$ to be correctly defined, we require that there exists a sequence of honest strategies $\cS_1,\cS_2,\ldots$  satisfying:
	\begin{itemize}
		\item The success probability of $\cS_i$ approaches $c$ as $i \to \infty$, and
		\item For all verifiers $\adv{V}$, there exists a simulator $\Sim_{\adv{V}}$ whose output distribution can be approximated arbitrarily well by the sequence of honest strategies. In other words, for all $\delta$ there exists an $i$ such that the total variation distance between $\View(\adv{V}(x) \leftrightarrow \cS_i)$ and $\Sim_{\adv{V}}$ is at most $\delta$.
	\end{itemize}
	
	This subtlety only arises when considering ``zero gap'' classes such as $\PZKMIP^*_{1,1}[k,r]$. 
\end{enumerate}

\subsection{Parallel repetition}
\label{sec:parallel-repetition}

Parallel repetition of interactive protocols is a commonly used technique for performing \emph{gap amplification}. We now define what this means for $1$-round $\MIP^*$
protocols. 

\begin{definition}[Parallel repetition of a one-round $\MIP^*$ protocol]
 Let $V$ denote a $1$-round, $k$-prover $\MIP^*$ protocol. 
 The $m$-fold parallel repetition of $V$ is another $1$-round, $k$-prover $\MIP^*$ protocol $V^m$ where $m$ independent instances of $V$ are executed simultaneously. Let $q_{ij}$ denote the questions from instance $i$ to prover $j$. Then prover $j$ receives $(q_{1j},q_{2j},\ldots,q_{mj})$ simultaneously, and responds with answers $(a_{1j},a_{2j},\ldots,a_{nj})$. The answers $(a_{i1},a_{i2},\ldots,a_{ik})$ is then given to the $i$'th verifier instance, and $V^m$ accepts if and only if all instances accept.
\end{definition}

The behaviour of $\omega^*(V^m)$ as a function of $n$ and $\omega^*(V) < 1$ is non-trivial; clearly, if $\omega^*(V) = 1$, then $\omega^*(V^m) = 1$ as well. Although one might expect that $\omega^*(V^m)$ decays exponentially with $m$ in the case that $\omega^*(V) < 1$, this is not known in general. Raz~\cite{Raz98} showed that such exponential decay \emph{does} hold for classical $1$-round, $2$-prover $\MIP$ proof systems, but extending this to the case of more provers or $\MIP^*$ proof systems has remained an active area of research. It is an open question for whether the analogue of Raz's result holds for $\MIP^*$ protocols (although a polynomial-decay bound is known~\cite{Yuen16}). 

Bavarian, Vidick, and Yuen~\cite{BavarianVY17} showed that an exponential-decay parallel repetition theorem also holds for $1$-round $\MIP^*$ protocols that have the property of being \emph{anchored}, and furthermore, \emph{every} $1$-round $\MIP^*$ protocol can be transformed into an equivalent anchored protocol. Their result has the additional benefit in that it holds for any number of provers.

We do not formally define the anchoring property here, but instead we describe a simple transformation to anchor any $1$-round $\MIP^*$ protocol.
\begin{definition}[Anchoring]
  \label{def:anchored}
 Let $\alpha > 0$ be some constant. Given a $1$-round, $k$-prover $\MIP^*$ protocol $V$, define its \emph{$\alpha$-anchored version} $V_\bot$ to be the protocol which:
 \begin{enumerate}
 	\item Runs the verifier in $V$ to obtain questions $(q_1,\ldots,q_k)$ for the $k$ provers. 
	\item Independently choose each coordinate $i$ with probability $\alpha$ and replace $q_i$ with an auxiliary question symbol $\bot$, and send the questions to each prover. 
	\item If any prover received the auxiliary question $\bot$, automatically accept. Otherwise, accept the provers' answers only if $V$ would have accepted.
\end{enumerate}
\end{definition}
This transformation preserves completeness and soundness: $\omega^*(V) = 1$ if and only if $\omega^*(V_\bot) = 1$. In general, we have the relationship
\[
	\omega^*(V_\bot) = (1 - \alpha)^k \omega^*(V) + (1 - (1-\alpha)^k).
\]
Bavarian, Vidick and Yuen~\cite{BavarianVY17} showed the parallel repetition of anchored games admits an exponential decay in success probability. 

\begin{theorem}
  \label{thm:parallel-repetition-anchored}
  Let $\alpha > 0$. Let $V$ be a $1$-round, $k$-prover $\MIP^*$ protocol. Let $V_\bot$ be the $\alpha$-anchored version of $V$ as defined in~\Cref{def:anchored}. Let $m > 0$ be an integer. If $\omega^*(V) = 1$, then $\omega^*(V_\bot^m) = 1$. Otherwise, 
  \[
  	\omega^*(V_\bot^m) \leq \exp(-\beta \eps^\gamma m)
  \]
  where $\beta$ is a universal constant depending on $\alpha$ and the protocol $V$, $\eps$ is defined as $1 - \omega^*(V)$, and $\gamma$ is a universal constant.

\end{theorem}

\section{Our zero knowledge protocol}
\label{sec:protocol}

In this section we present the zero knowledge transformation for general $\MIP^*$ protocols.
For convenience we reproduce the statement of Theorem~\ref{thm:main}. 

\mainthm*

Fix a promise language $L \in \MIP^*_{1,s}[k,r]$. 
There exists a polynomial-time computable function $V$ that on input $x$ outputs a $k$-prover, $1$-round protocol circuit $V(x)$ such that if $x \in L$, then $\omega^*(V(x)) = 1$, and otherwise $\omega^*(V(x)) < s$. Furthermore, since we are dealing with an $\MIP^*$ proof system, the communication between the verifier and the provers is classical. Thus, we can assume that the protocol circuit $V$ has the following structure. All qubits of the verifier register $\sV$ are initialized to $\ket{0}$. The protocol circuit proceeds in five phases:
\begin{itemize}
	\item \textbf{Verifier Operation Phase 1}: All computation in this phase of the protocol occurs on the verifier register $\sV$. At the end of the computation, the verifier's messages to the $i$'th prover are stored in a subregister $\sN_i$ of $\sV$. 
	\item \textbf{Copy Question Phase}: For each prover $i$, CNOT gates are applied bitwise from $\sN_i$ to bits in the register $\sM_i$.
	
	\item \textbf{Prover Operation Phase}: Each prover $i$ applies prover gate $P_i$ to registers $\sP_i \sM_i$, in sequence. 
	
	\item \textbf{Copy Answer Phase}: For each prover $i$, CNOT gates are applied bitwise from $\sM_i$ to bits in the register $\sN_i$.
	
	\item \textbf{Verifier Operation Phase 2}: The remaining computation in the protocol occurs on the verifier register $\sV$, and the accept/reject decision bit is stored in a designated output qubit of $\sV$.
\end{itemize}
As mentioned earlier, we assume that the non-prover gates of the  protocol circuit $V(x)$ are drawn from the universal gate set $\{ H, \Lambda(X), \Lambda^2(X) \}$. Figure~\ref{fig:protocol} gives a diagrammatic representation of this five-phase structure, depicting a protocol in which a verifier interacts with a single prover. 

\begin{figure}[H]
  \begin{mdframed}[style=figstyle]
    \begin{center}
      \begin{tikzpicture}[scale=1, control/.style={circle, fill,
          minimum size = 4pt, inner sep=0mm}, target/.style={circle,
          draw, minimum size = 7pt, inner sep=0mm},
        vgate/.style={draw, minimum height = 1.7cm, minimum width =
          1.3cm, fill=ChannelColor}, pgate/.style={draw, minimum
          height = 1.3cm, minimum width = 1.3cm, fill=ChannelColor}]

        \node (V) at (-2,.5) {$\phantom{\sV}$}; %
        \node (V2) at (-2,0) {$\phantom{\sV}$}; %
        \node (V3) at (-2,.25) {$\sV$}; %
        \node (M) at (-2,-.5) {$\sN$}; %
        \node (P) at (-2,-2.3) {$\sP$}; %
        \node (P2) at (-2,-1.8) {$\sM$}; %
        \node (P3) at (-2,-2.05) {$\phantom{sP}$}; %

        \node (OutV) at (6,.5) {}; %
        \node (OutV2) at (6,0) {}; %
        \node (OutM) at (6,-.5) {}; %
        \node (OutP) at (6,-2.3) {}; %
        \node (OutP2) at (6,-1.8) {}; %

        \draw (V.east)--(OutV); %
        \draw (V2.east)--(OutV2); %
        \draw (M.east)--(OutM); %
        \draw (P.east)--(OutP); %
        \draw (P2.east)--(OutP2); %

        \node[vgate] (CQ) at (-.8,.05) {$V_1$}; %
        \node[target] (QG) at (.5,-1.8) {};
        \node[control] (C1) at (.5,-.5) {}; %
        \draw (C1.center)--(QG.south); %
        \node[pgate] (PG) at (2.1,-2.05) {$P$}; %
		\node[control] (C2) at (3.7,-1.8) {};
        \node[vgate] (CA) at (5,.05) {$V_2$}; %
        \node[target] (T1) at (3.7,-.5) {}; %
        \draw (T1.north)--(C2.center);
      \end{tikzpicture}
    \end{center}
  \end{mdframed}
  \caption{A quantum circuit representation of an MIP* protocol}
  \label{fig:protocol}
\end{figure}
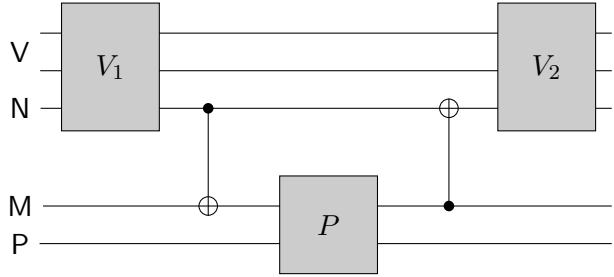

As described in the Introduction, we first transform the protocol circuit $V(x)$ into an equivalent protocol circuit $V_{enc}(x)$ that performs its computations fault-tolerantly. Then, we use the compression techniques of~\cite{Ji17,FitzsimonsJVY18} on the protocol defined by $V_{enc}(x)$ to obtain a protocol $V_{ZK}(x)$ which has the desired zero knowledge properties. 

\subsection{Robustifying protocol circuits}
\label{sec:transformation}

We now describe a polynomial-time transformation that takes as input the description of a $k$-prover, $1$-round $\MIP^*$ protocol circuit such as $V$ described above, and outputs another $k$-prover, $1$-round protocol circuit $V_{enc}$ that describes an \emph{equivalent} $\MIP^*$ protocol, but has additional fault-tolerance properties. %

The non-prover gates of $V_{enc}$ are drawn from the universal gate set $\{ H \otimes H, \Lambda(X), \Lambda^2(X) \}$.\footnote{The doubled Hadamard gate is used for technical reasons; the second Hadamard gate can always be applied to unused ancilla qubits if it is not needed.} The registers that are involved in the protocol $V_{enc}$ are $\{ \sP_1,\ldots,\sP_k,\sM_1,\ldots,\sM_k, \sV \}$. The verifier workspace register $\sV$ can be subdivided into registers $\sA$, $\sB$, $\sO$, and $\sN_1,\ldots,\sN_k$. Intuitively, the register $\sA$ holds encoded qubits, the register $\sB$ holds unencoded qubits, the register $\sO$ holds an encoding of the output bit at the end of the protocol, and the register $\sN_i$ is isomorphic to $\sM_i$ for all $i$.

Let the inner code $\innerCode$ be a $192$-simulatable code. We remark that from
\Cref{T:simulatable}, such codes exist and each logical qubit is encoded in $m$
physical bits, for some constant $m$. 

At the beginning of the protocol $V_{enc}$, the qubits in register $\sV$ are initialized to zero. In addition to the five phases of $V$, there are two additional phases in $V_{enc}$. First, the protocol $V_{enc}$ goes through a \textbf{Resource Generation Phase}, in which the verifier generates many $\innerCode$ encodings of the following states in its private workspace:
	\begin{enumerate}
		\item Toffoli magic states $\ket{\Toffoli} = \Lambda^2(X) (H \otimes H
      \otimes I) \ket{0,0,0}$.
		\item Ancilla $\ket{0}$ qubits.  
		\item Ancilla $\ket{1}$ qubits.  
	\end{enumerate}
	Thus the state of the register $\sV$ after the Resource Generation Phase will be a tensor product of encoded magic states, encoded $\ket{0}$ states, encoded $\ket{1}$ states, and unencoded $\ket{0}$ states. 
	
	Now the the verifier of $V_{enc}$ simulates the five computational phases of $V$, but as logical operations acting on data encoded using the inner code $\innerCode$. For the Verifier Operation Phases and the Copy Question/Answer Phases, each non-prover gate $g_i \in \{H, \Lambda(X), \Lambda^2(X)\}$ of $V$
    is replaced in $V_{enc}$ with the encoding of $g_i$ using the $\innerCode$, as given by Theorem~\ref{T:simulatable}. For
    example, if $g_i$ in $V$ is a Hadamard gate that acts on some qubit $\alpha$
    of $\sV$, then its equivalent will be a sequence of (double) Hadamard gates acting
    transversally on the physical qubits of the encoding of qubit $\alpha$. If
    $g_i$ in $V$ is a Toffoli gate, then in $V_{enc}$ the logical gate is applied
    using the Toffoli gadget (as described in Section~\ref{sec:simulatable}). Thus, all of the gates of the verifier in $V$ are performed in an encoded manner in $V_{enc}$. %
    
    The Prover Operation Phase proceeds as before; each prover applies their prover gate on the $\sM \sP$ registers in sequence. We assume that the Prover Operation Phase is padded with sufficiently many identity gates so that the number of time steps in between each prover gate application is at some sufficiently large constant times the block length of the inner code $\innerCode$.
    
    Note that the questions to the provers are encoded using the inner code $\innerCode$; this is not a problem for the provers, who can decode the questions before performing their original strategy, and encode their answers afterwards. 
    
    Finally, we assume that at the end of the (encoded) Verifier Operation Phase 2, the register $\sO$ stores the logical encoding of the accept/reject decision bit. After Verifier Operation Phase 2, the protocol $V_{enc}$ executes the \textbf{Output Decoding Phase}, where the logical state in register $\sO$ is decoded (using the decoder from $\innerCode$) into a single physical qubit $\sO_{out}$.

It is easy to see that this transformation from $V$ to $V_{enc}$ preserves the acceptance probability of the protocol.
\begin{proposition}
\label{prop:same_acc}
For all $1$-round $\MIP^*$ protocols $V$, for the $\MIP^*$ protocol $V_{enc}$ that is the result of the transformation just described, we have that
\[
  \omega^*(V) = \omega^*(V_{enc}).
\]
\end{proposition}

\subsubsection{Micro-phases of $V_{enc}$}
\label{sec:micro-phases}
We assume the following structural format to the protocol circuit $V_{enc}$: aside from the major phases of $V_{enc}$, we can partition the timesteps of the circuit into ``micro-phases'', where each micro-phase consists
of a constant number of consecutive timesteps, and each micro-phase can be
classified according to the operations performed within it:
\begin{itemize}
	\item \textbf{Idling}: the gates applied by the verifier during this micro-phase are all identity gates.
	\item \textbf{Resource encoding}: gates are applied to a collection of ancilla $\ket{0}$ qubits to form either an encoding of a $\ket{0}$ state, $\ket{1}$ state, or a Toffoli magic state.
	\item \textbf{Logical operation}: the encoding of a single logical gate is being applied to some encoded blocks of qubits, possibly along with some unencoded ancilla qubits.	
	\item \textbf{Output decoding}: the output register $\sO$ of the verifier circuit is decoded to obtain a single qubit answer. This is exactly the Output Decoding phase.
\end{itemize}
For example, the Resource Generation phase consists of a sequence of resource encoding micro-phases, applied to blocks of ancilla qubits. The Verifier Operation phases consist of sequences of both idling steps and logical operations, applied to blocks of encoded qubits as well as ancilla qubits. The timesteps during the Prover Operation phase are classified as idling steps, because the verifier is not applying any gates to its private space. 

\subsubsection{Prover reflection times}

Given the protocol circuit $V_{enc}$ of length $T$, we identify special timesteps during the protocol corresponding to the timesteps where the provers apply their prover gate. For every prover $i$, we define $t_\star(i) \in \{0,1,2,\ldots,T\}$ to be the time in the protocol circuit when prover $i$ applies their prover gate $P_i$.

\subsection{A zero knowledge $\MIP*$ protocol to decide $L$}

Given the transformation from an $\MIP^*$ protocol $V$ to an equivalent ``fault-tolerant'' protocol $V_{enc}$, we now introduce a second transformation that takes $V_{enc}$ and produces another equivalent $\MIP^*$ protocol $V_{ZK}$ that has the desired zero knowledge properties.

This protocol is obtained by applying the compression procedure of~\cite{Ji17,FitzsimonsJVY18} to $V_{enc}$. Since we are compressing interactive protocols (involving verifiers and provers) into other interactive protocols, to keep things clear we use the following naming convention:
\begin{itemize}
	\item \textbf{Verifiers} and \textbf{provers} refer to the parties in $V_{enc}$ (i.e. the protocol that is being \emph{compressed});
	\item \textbf{Referees} and \textbf{players} refer to the parties in $V_{ZK}$ (i.e. the protocol that is the result of the compression scheme).
\end{itemize}

At a high level, the protocol $V_{ZK}$ is designed to verify that the players possess (an encoding of) a \emph{history state} of the protocol $V_{enc}$:
\begin{equation}
\label{eq:history_state}
	\ket{\Phi}_{\sC \sV \sM \sP} = \frac{1}{\sqrt{T+ 1}} \sum_{t = 0}^T \ket{\unary(t)}_{\sC} \otimes \ket{\Phi_t}_{\sV \sM \sP}
\end{equation}
where $T$ is the number of gates of $V_{enc}$, $\unary(t)= t_1 t_2 \cdots t_T$ denotes
the unary encoding of time $t$, i.e. \[
		t_\ell =
     \begin{cases}
				1 & \qquad \mbox{ if } \ell \leq t \\
				0 & \qquad \mbox{ otherwise}
     \end{cases},
	\]
and $\ket{\Phi_t}$ is the state of the protocol $V_{enc}$ after $t$ time steps (called the $t$'th \emph{snapshot state}). %

We specify some details of the protocol $V_{ZK}$:

\begin{itemize}
	\item \textbf{Rounds}: $1$-round protocol
	\item \textbf{Number of players}: $k+4$ players, which are are divided into 
	$k$ \emph{prover players} (labelled $PP_1,\ldots,PP_k$) and $4$ \emph{verifier players} (labelled $PV_1,\ldots,PV_4$). %

\item \textbf{Question and answer format}: questions to the verifier players are $6$-tuples of the form $(W_1, \ldots,W_6)$, where each $W_i$ denotes a two-qubit Pauli observable on some specified pair of qubits, and the six observables commute. Furthermore, the Pauli observables are tensor products of operators from the set $\{ I, X, Z \}$. An example of a question would be: $(X_1 X_2, Z_1 Z_2, I_{7} Z_{5},X_3 Z_4, Z_3 X_4, X_7 I_5)$. Verifier players' answers are a $6$-tuple of bits $(a_1,\ldots,a_6)$.

Questions to prover player $PP_i$ can be one of three types:
\begin{enumerate}
	\item Prover reflection, denoted by $\star_i$. %
	\item Question gates, denoted by $Q_{ij}$ for $j = 1,\ldots,m'$, where $m'$ is the maximum number of qubits in the message registers $\{\sM_i\}$ in the protocol $V_{enc}$.
	\item Question flag flip, denoted by $QF_i$.
	\item Answer gates, denoted by $A_{ij}$ for $j = 1,\ldots,m'$.
	\item Answer flag flip, denoted by $AF_i$.
\end{enumerate}
	We notice that even if the Prover players' original answers consisted
    of a single bit, after robustifying the protocol circuits, the answers
    become an encoding of the logical bit.

\end{itemize}

The distribution of questions and the rules used by the referee in $V_{ZK}$ are essentially identical to the ones used in the compression protocol in~\cite{FitzsimonsJVY18}.\footnote{The main difference concerns the questions ``Question flag flip'' and ``Answer flag flip'' to the provers, which do not occur in~\cite{Ji17,FitzsimonsJVY18}. These will be helpful for the analysis of zero knowledge property. We explain in Appendix~\ref{sec:differences} the slight modifications to the protocol from~\cite{FitzsimonsJVY18} that are needed.} Given those, the results of~\cite{FitzsimonsJVY18} show that $V_{ZK}$ is a complete and sound $\MIP^*$ proof system for $L$:
\[
	L \in \MIP^*_{1,s'}[k + 4,1]
\]
where $s' = (1 - s)^\beta/p(n)$ for some universal constant $\beta$ and some polynomial $p(n)$ that depends on the original protocol $V$, and $s$ is the soundness of $V$.

The details of the the question distribution, the rules and the soundness
analysis are irrelevant for this paper, as we are only concerned with
establishing the zero knowledge property of $V_{ZK}$. For this, we only need to
consider the interaction between honest players and a potentially cheating referee $\hat{R}$. 

\subsubsection{An honest strategy $\cS_{ZK}$ for $V_{ZK}$} 
\label{sec:honest}

We now specify an honest strategy $\cS_{ZK}(x)$ for the players in $V_{ZK}(x)$ in the case that $x \in L_{yes}$. Since $x \in L_{yes}$, by definition we have that $\omega^*(V(x)) = 1$, and therefore by Proposition~\ref{prop:same_acc} we get $\omega^*(V_{enc}(x)) = 1$. Thus there exists a sequence of finite dimensional unitary strategies $\{ \cS_1(x),\cS_2(x),\ldots \}$ for $V_{enc}(x)$ such that the acceptance probability approaches $1$. For simplicity, we assume that there exists a finite dimensional unitary strategy $\cS(x)$ for $V_{enc}(x)$ that is accepted with probability $1$; in the general case, we can take a limit and our conclusions still hold. 

The strategy $\cS(x)$ consists of a dimension $d$, an entangled state
$\ket{\psi}$ on registers $\sP_1,\ldots,\sP_k$ and $\sM_1,\ldots,\sM_k$ (where the registers $\sP_i$ have dimension $d$), and a collection of unitaries $\{ P_i \}$ where $P_i$ acts on registers $\sP_i \sM_i$. We assume, without loss of generality, that in under the strategy $\cS$ in protocol $V_{enc}(x)$, the state of the message registers $\{\sM_i\}_i$ are in the code subspace of $\innerCode$ at each time step of the protocol (where we treat the prover operations as taking one time step). 

Given this, we define the measurement\footnote{Since the protocols $V$ and $V_{enc}$ are general $\QMIP$ protocols, the strategy $\cS$ is a unitary strategy. Since $V_{ZK}$ is a $\MIP^*$ protocol, we specify $\cS_{ZK}$ as a \emph{measurement} strategy.} strategy $\cS_{ZK}(x)$ in the following way. For notational simplicity, we omit mention of the input $x$ when it is clear from context.

\paragraph{The shared entanglement} Let $\ket{\Phi}_{\sC \sV \sM \sP}$ denote the history state of the protocol $V_{enc}(x)$ when the provers use strategy $\cS$ (as in~\eqref{eq:history_state}). The initial state $\ket{\Phi_0}$ is $\ket{0}_{\sV} \otimes \ket{\psi}_{\sM \sP}$. 

We now construct an \emph{distributed history state} $\ket{\Phi'}_{\sC' \sV' \sM \sP \sF}$ from $\ket{\Phi}$ in two steps. First, without loss of generality we augment a $k$-partite register $\sF = \sF_1,\ldots,\sF_k$ to $\ket{\Phi}$ so that serves as flags that indicate which operations the $i$'th prover has applied. Thus the augmented history state looks like
\[
	\ket{\Phi} = \frac{1}{\sqrt{T+ 1}} \sum_{t = 0}^T \ket{\unary(t)}_{\sC}
  \otimes \ket{\Delta_t}_{\sV \sM \sP} \otimes \ket{f(t)}_{\sF}
\]
where $\ket{f(t)}_{\sF} = \bigotimes_{i} \ket{f_i(t)}_{\sF_i}$ and 
$\ket{f_i(t)}_{\sF_i} = \ket{q_i(t)}_{\sF_{Q_i}} \otimes
\ket{p_i(t)}_{\sF_{P_i}} \otimes \ket{a_i(t)}_{\sF_{A_i}}$. For all $i \in
\{1,2,\ldots,k\}$, the functions $q_i(t), p_i(t), a_i(t)$ are boolean functions
of the time $t$, defined as follows:
\[
   q_i(t) =  \begin{cases}
				1 & \qquad \mbox{ if } t \geq t_\star(i) - 1 \\
				0 & \qquad \mbox{ otherwise}
     \end{cases},
\]
\[
   p_i(t) =  \begin{cases}
				1 & \qquad \mbox{ if } t \geq t_\star(i) \\
				0 & \qquad \mbox{ otherwise}
     \end{cases},
\]
and
\[
   a_i(t) =  \begin{cases}
				1 & \qquad \mbox{ if } t \geq t_\star(i) + 1 \\
				0 & \qquad \mbox{ otherwise}
     \end{cases}.
\]
The flags $q_i,p_i,a_i$ flip from $0$ to $1$ consecutively: at time $t = t_\star(i) - 2$, all flags for player $i$ are set to $0$. By the time $t = t_\star(i) + 1$, all flags for player $i$ are set to $1$.

Next, we perform a qubit-by-qubit encoding of the $\sC$ and $\sV$ registers of $\ket{\Phi}$ using the \emph{outer code} $\outerCode$, to obtain the \emph{encoded history state} $\ket{\Phi'}$ defined on registers $\sC', \sV', \sM, \sP$. Each qubit of $\sC$ and $\sV$ are encoded into $4$ physical qubits.

The allocation of the registers of $\ket{\Phi'}$ to the $k+4$ players are as follows:
\begin{enumerate}
\item The register $\sC$ consists of $T$ qubits. For $i = 1,\ldots,T$, let $\sC_i$ denote the $i$'th qubit register of $\sC$. For $j = 1,\ldots,4$, let $\sC'_{ij}$ denote the $j$'th share of the $\outerCode$ encoding of $\sC_i$. In the honest case, the $j$'th verifier player $PV_j$ has the qubits $\{\sC_{ij}' \}_i$. 

\item Similarly, the registers $\sV'_{ij}$ denote the $j$'th share of the encoding of the register $\sV_i$; the subregisters $\sA_i, \sB_i, \sO_i, \sN_i$ of $\sV$ are encoded into subregisters $\sA_{ij}',\sB_{ij}',\sO_{ij}',\sN_{ij}'$ of $\sV'$ respectively. In the honest case, the $j$'th verifier player $PV_j$ holds qubits $\{  \sV_{ij}' \}_i$. 

\item The prover players' $\{PP_1,\ldots,PP_k\}$ represent the original $k$ players of the protocol $V$ and $V_{enc}$. In the honest case, prover player $PP_i$ holds registers $\{ \sF_i  \sP_i \sM_i \}$. Note that these registers are not encoded and split up like with the clock and verifier registers.
\end{enumerate}

\paragraph{Player measurements} Since $V_{ZK}$ is a $1$-round $\MIP*$ protocol, we specify the strategy $\cS_{ZK}$ in terms of measurement operators. 
\begin{itemize}

\item When the verifier players receive a $6$-tuple of commuting Pauli observables $(W_1,\ldots,W_6)$, they measure each of the observables $\sigma_{W_1},\ldots,\sigma_{W_6}$ in sequence on the designated qubits of their share of $\ket{\Phi'}$, and report the measurement outcomes $(a_1,\ldots,a_6)$. For example, if $W_1 = X_1 Z_2$, then the corresponding observable would be $\sigma_X \otimes \sigma_Z$ acting on qubits labelled $1$ and $2$.

\item When prover player $PP_i$ receives a prover reflection question $\star_i$, they measure the following observable on the registers $\sF_{P_i} \sP_i \sM_i$:
\[
	P_i' = \ketbra{0}{1}_{\sF_{P_i}} \otimes P_i^\dagger + \ketbra{1}{0}_{\sF_{P_i}} \otimes P_i
\]
where $P_i$ acts on $\sP_i \sM_i$. It is easy to see that $P_i'$ is an observable with a $+1$ eigenspace and a $-1$ eigenspace.

\item When prover player $PP_i$ receives a ``Question gate'' question $Q_{ij}$, they measure the observable $\sigma_X$ on the register $\sM_{ij}$, and report the one-bit answer. When $PP_i$ receives an ``Answer gate'' question $A_{ij}$, they measure the observable $\sigma_Z$ on the register $\sM_{ij}$, and report the one-bit answer.

\item When prover player $PP_i$ receives the ``Question flag flip'' question $QF_i$, they measure the observable $\sigma_X$ on the register $\sF_{Q_i}$. When they receive ``Answer flag flip'' question $AF_i$, they measure the observable $\sigma_X$ on the register $\sF_{A_i}$. 

\end{itemize}

The analysis of the compression protocol in~\cite{FitzsimonsJVY18} implies that the strategy $\cS_{ZK}(x)$ is accepted in the protocol $V_{ZK}(x)$ with probability $1$. We now proceed to argue the zero knowledge property of the protocol $V_{ZK}$ with the honest player strategy $\cS_{ZK}$.

{
\renewcommand{\arraystretch}{1.2}
\begin{figure}[H]
\begin{center}
\begin{tabular}{l | p{12cm}}
 \textbf{Notation} & \textbf{Meaning} \\
\hline 
 $V_{enc}$			& 	The fault tolerant encoding of the original protocol $V$ \\
 $k$ 				&	 Number of provers in the protocol $V_{enc}$ \\
 $V_{ZK}$				&	The zero knowledge protocol \\
 $N_V$					 &   Number of verifier players  in $V_{ZK}$, which is $4$.    \\
 $t_\star(i)$ 	& The time that prover $i$ applies prover reflection $P_i'$. \\
 $q_i(t)$ 			& 	Indicator function that is $1$ iff $t \geq t_\star(i) - 1$. Used as a flag to indicate whether the questions for the $i$'th prover have been all copied. \\
  $p_i(t)$ 			& 	Indicator function that is $1$ iff $t \geq t_\star(i)$.  Used as a flag to indicate whether the $i$'th prover has applied its reflection $P_i'$. \\
 $a_i(t)$ 			& 	Indicator function that is $1$ iff $t \geq t_\star(i) + 1$. Used as a flag to indicate whether the $i$'th prover is ready to copy its answers to the verifier. \\
 $\cS_{ZK}$			&	The honest player strategy for $V_{ZK}$ \\  
 $\ket{\Phi}$		&	The unencoded history state of an interaction in the protocol $V_{enc}$ \\
 $\ket{\Phi_I}$		&	The restriction of the history state $\ket{\Phi}$ to a time interval $I$ \\ 
 $P_i'$				&	The prover reflection used by prover player $PP_i$ in $\cS_{ZK}$ \\
 $\adv{R}$			&  A (possibly cheating) referee in the protocol $V_{ZK}$ \\
 $\adv{W}$			& A tuple of questions in $V_{ZK}$, or the associated observable measured by the players in $\cS_{ZK}$. \\
 $\adv{W}^{(V,r)}, \adv{W}^{(P,r)}$	& Questions to the $r$'th verifier and prover players, respectively. \\
 $\tilde{W}^{(V,r)}, \tilde{W}^{(P,r)}$ & Players' measurement observables without the prover reflections \\
 $L$				 &   This is the number of verifier player qubits that can be addressed by a question $\adv{W}$. This is $12N_V = 48$, which is a constant.			 \\
\end{tabular}
\end{center}
\caption{Notation reference}
\label{fig:notation}
\end{figure}
}

\section{Zero knowledge property of $V_{ZK}$}
\label{sec:zk}
Let $\hat{R}(x)$ be an arbitrary referee (modelled as a probabilistic polynomial-time Turing machine) interacting with $k$ provers that use the measurement strategy $\cS_{ZK}(x)$ defined above. In general, this referee $\hat{R}(x)$ may try to gain forbidden knowledge by deviating from the behaviour of the referee specified by the protocol $V_{ZK}$. In this section, we show this cannot happen by describing an efficient simulator $\Sim(x)$ whose output
distribution is equal to $\View(\hat{R}(x) \leftrightarrow \cS_{ZK}(x))$.

A referee $\adv{R}$ could try to cheat by sampling questions from a different distribution than the one that is specified in the $V_{ZK}$ protocol. Furthermore, the referee could interact with the provers \emph{adaptively}: it could send some messages to a subset of the provers, get some answers, and depending on those responses choose questions for another set of provers. We can assume that a cheating referee  does not interact with the same prover in $V_{ZK}$ twice; since the protocol is supposed to be one round, an honest prover would abort the protocol if the referee interacted with it multiple times. Similarly, we assume that a cheating referee only asks questions that match the {\em format} of questions in $V_{ZK}$.

In Section~\Cref{sec:non-adaptive}, we show how to simulate the interaction between $\adv{R}$ and the players when $\adv{R}$ is \emph{non-adaptive}, meaning that the questions for all
players are picked simultaneously by the referee before interacting with them. 
In \Cref{sec:general-adversary} we show how to perform this simulation for general
adaptive referees $\adv{R}$. 

For the remainder of this section, we omit mention of the input $x$; we assume that the referee $\adv{R}$ and the strategy $\cS_{ZK}$ implicitly depend on $x$.

We introduce some additional notation. 
\begin{itemize}
	\item Let $N_V = 4, N_P = k$ denote the number of verifier players and prover players, respectively. 
	\item Let $\adv{W}^{(V,r)}$, and  $\adv{W}^{(P,r)}$ denote the question for the $r$'th verifier player and $r$'th prover player respectively. The question $\adv{W}^{(V,r)}$ is a $6$-tuple $(\adv{W}^{(V,r)}_1,\ldots,\adv{W}^{(V,r)}_6)$ of commuting two-qubit Pauli observables. 
	
	\item For $r \in [N_V]$, for $j \in \{1,\ldots,6\}$, we overload notation by also letting $\adv{W}^{(V,r)}_j$ denote the $j$'th Pauli observable used by $r$'th verifier player in the honest strategy $\cS_{ZK}$ when they receive question $\adv{W}^{(V,r)}$ , as specified in Section~\ref{sec:honest}. We also let $\adv{W}^{(V,r)}$ denote the observable that is the product $\adv{W}^{(V,r)}_1 \cdots \adv{W}^{(V,r)}_6$ (the order does not matter because the observables commute). Whether or not $\adv{W}^{(V,r)}_j$ and $\adv{W}^{(V,r)}$ are used to refer to the question or the observables will be clear from context.
	
	\item For $r \in [N_P]$, we let $\adv{W}^{(P,r)}$ also denote the observable used by prover player $PP_r$ in the honest strategy $\cS_{ZK}$ when they receive question $\adv{W}^{(P,r)}$. For example, if $\adv{W}^{(P,r)}$ is a ``Question gate'' $Q_{rj}$ or a ``Question flag flip'' $QF_{r}$, then as an observable we interpret $\adv{W}^{(P,r)}$ as the corresponding Pauli observable in the honest strategy $\cS_{ZK}$. If $\adv{W}^{(P,r)}$ is a prover reflection $\star_r$, then as an observable we interpret $\adv{W}^{(P,r)}$ as $P'_r$. 
	
	\item Let $\adv{W} = \Paren{ \adv{W}^{(D,r)} }_{D \in \{V,P\}, r \in [N_D]}$ denote the tuple of questions for all players in the protocol. We also use $\adv{W}$ to denote the tensor product of observables
	\[
		\adv{W} = \bigotimes_{D \in \{V,P\}, r \in [N_D]} \adv{W}^{(D,r)}.
	\]
	
	\item For $r \in [N_V]$, $j \in \{1,\ldots,6\}$, we define the observable $\tilde{W}^{(V,r)}_j$ to be $\adv{W}^{(V,r)}_j$. For $r \in [N_P]$, we define the observable 
		\[
			\tilde{W}^{(P,r)} = 
			\left \{ 
			\begin{array}{ll}
				\sigma_X(\sF_{P_r}) & \text{if } \adv{W}^{(P,r)} = \star_r \\
				\adv{W}^{(P,r)} &  \text{otherwise }
			\end{array}
			\right.
		\]
	
	Notice that the observables $\tilde{W}^{(D,r)}$ are simply Pauli observables (or products of Pauli observables). We explain the reasoning behind defining the observables $\tilde{W}^{(D,r)}$ in the next section. 
	
  \item We also define the projectors corresponding to the players' observables. The verifier players output a $6$-tuple of bits $(b_1,\ldots,b_6) \in \{0,1\}^6$. For $r \in [N_V]$, $j \in \{1,\ldots,6\}$, and bit $b \in \{0,1\}$, define
  \[
  	\adv{W}^{(V,r)}_j(b) = \frac{1}{2} \Paren{ I + (-1)^b \adv{W}^{(V,r)}_j}
  \]
  which is the projector onto the $b$ subspace of $\adv{W}^{(V,r)}_j$. Define
  \[
  	\adv{W}^{(V,r)}(b_1,\ldots,b_6) = \prod_{i = 1}^6 \adv{W}^{(V,r)}_i(b_i).
  \]
	The prover players only output a single bit, so for $r \in [N_P]$, define
  \[
  \adv{W}^{(P,r)}(b) = \frac{1}{2} \Paren{ I + (-1)^b \adv{W}^{(P,r)}}.
  \]
  
  Let $a = \Paren{ a^{(D,r)} }_{D \in \{V,P\}, r \in [N_D]}$ denote an answer vector for all players (where $a^{(V,r)}$ corresponds to a $6$-tuple of bits). Then for every tuple of questions $\adv{W}$, we define
	  \[
      \adv{W}(a) = \bigotimes_{D \in \{V,P\}, r \in [N_D]} \adv{W}^{(D,r)}(a^{(D,r)}).
	  \]

	We define the projectors $\tilde{W}^{(D,r)}(b)$ and $\tilde{W}^{(D,r)}$ analogously.
\end{itemize}
For convenience we also provide a notation reference table in Figure~\ref{fig:notation}.

\subsection{Non-adaptive cheating referees}
\label{sec:non-adaptive}
In this section, we show that for \emph{every} possible
combination of (correctly formatted) questions to the players, 
the joint distribution of answers of
players using the honest strategy $\cS_{ZK}$ is efficiently simulable.

The reason for defining the observables $\tilde{W}^{(D,r)}$ is as follows. Ultimately, the goal of the simulator is, for every question tuple $\adv{W}$, to sample answer vectors $a$ that is distributed according to the probability density
\[
	\Tr \Paren{ \outerEnc(\Phi) \, \adv{W}(a)}
\]
where $\Phi = \ketbra{\Phi}{\Phi}$ is the shared entangled state and $\adv{W}(a)$ also denotes the projectors corresponding to outcome $a$ in the honest strategy $\cS_{ZK}(x)$ (see Section~\ref{sec:honest}). The main difficulty is that the simulator does not have any control over the prover reflections, nor the parts of $\Phi$ that correspond to the provers' private registers (which may be unbounded in size). 

To get around this issue, the key observation we use is the following: the measurements of the verifier players are Pauli observables that act on at most a constant number of qubits. Furthermore, the measurements of the prover players when they get a question other than the special prover reflection $\star$ are also just Pauli observables on a constant number of qubits. 

We define two notions of support of a question tuple $\adv{W}$. Then, define the \emph{physical support} of $\adv{W}$ to be the set $S_{\adv{W}}'$ of qubit registers that are acted upon nontrivially by $\adv{W}$, omitting the prover players' private $\sP_r$ registers. The set $S_{\adv{W}}'$ contains subregisters of $\sC', \sV', \sM, \sF$. 

We can also define the \emph{logical support} of $\adv{W}$, denoted by the set $S_{\adv{W}}$, which contains subregisters of $\sC, \sV, \sM, \sF$ that correspond to the registers in $S_{\adv{W}}'$. The logical support set $S_{\adv{W}}$ contains all the $\sM_{ij}$ and $\sF_i$ registers that are in $S_{\adv{W}}'$. The set $S_{\adv{W}}$ contains $\sC_t$ if and only if $S_{\adv{W}}'$ contains $\sC_{tj}$ for some $j$, and similarly contains $\sV_i$ if and only if $S_{\adv{W}}'$ contains $\sV_{ij}$ for some $j$. The difference between the physical and logical support of $\adv{W}$ comes from the fact that the history state $\ket{\Phi}$ was encoded using $\outerCode$ and split between multiple provers.

Note that the number of qubit registers in $S_{\adv{W}}$ is at most $12N_V + k$. This is because the questions to each verifier player is a $6$-tuple of Pauli observables that act on up to $2$ qubits, and each prover player measures at most a single qubit flag register at a time. Define $L = 12N_V = 48$, which is the maximum number of verifier player qubits that can be addressed by $\adv{W}$.

For all $\adv{W}$, the simulator computes a \emph{succinct} description of density matrix $\rho$ defined only on the logical registers in $S_{\adv{W}}$ that mimics $\Phi$ in a certain sense that is captured by the following Lemma~\ref{lem:density-matrix}. Before stating the Lemma, however, we specify what we mean by succinct description of $\rho$. In general, $\rho$ will be a density matrix with dimension at least $2^{12 N_V + k}$, so the na\"{i}ve strategy of explicitly storing all the matrix entires of $\rho$ is not an efficient representation if the number of prover players $k$ is a growing function. Instead, we will specify our density matrices $\rho$ and measurement operators using the following type of efficient representation:

\begin{definition}[Efficient representations of operators]
\label{def:efficient_representation}
	Let $A$ denote a linear operator defined on $m$ qubits. The operator $A$ has an \emph{$(w,\ell)$-efficient representation} if there exist, for all $i \in \{1,2,\ldots,w\}$, a collection of operators $\{ A_{ij} \}$ where each $A_{ij}$ is defined on some subset $S_{ij} \subseteq \{1,2,\ldots,m\}$ of qubit registers, and 
\begin{enumerate}
	\item For all $i \in \{1,2,\ldots,w\}$, $\{ S_{ij} \}_j$ is a partition of $\{1,2,\ldots,m\}$.
	\item $|S_{ij}| \leq \ell$ for all $i,j$.
	\item The explicit matrix representation of $A_{ij}$ can be described using $2^{O(\ell)}$ bits.
	\item $A = \sum_{t = 1}^w \bigotimes_{j} A_{ij}$.
\end{enumerate}
\end{definition}

\noindent The following Claim justifies our definition of ``efficient representation'':
\begin{claim}
\label{clm:efficient_trace}
	Let $A,B$ be $m$-qubit operators with $(w,\ell)$-efficient representations $\{A_{ij}\}$ and $\{B_{ij}\}$, respectively. First, the efficient representations of both operators have bit complexity $w \cdot 2^{O(\ell)} \cdot \poly(m)$. Second, the trace $\Tr(AB)$ can be computed in time $w \cdot 2^{O(\ell)} \cdot \poly(m)$. 
\end{claim}

\noindent We can now state our main simulation Lemma:
\begin{lemma}\label{lem:density-matrix}
There is a PPT algorithm $\SimDensity$ that when given a tuple $\adv{W}$ of questions, outputs a $(3(T+1) L^2,4L)$-efficient representation of a density matrix $\rho$ such that
for all answer vectors $a = \Paren{ a^{(D,r)} }_{D \in \{V,P\}, r \in [N_D]}$, we have that
\begin{equation}
\label{eq:sim_density}
\Tr \Paren{ \outerEnc(\Phi) \, \adv{W}(a) } = \Tr \Paren{ \outerEnc(\rho) \,
  \tilde{W}(a)}.	
\end{equation}
Furthermore, the density matrix $\rho$ is defined on the logical support $S_{\adv{W}}$ of $\adv{W}$.
\end{lemma}

Before proving \Cref{lem:density-matrix}, we first prove a specialized version. Let $1 \leq t_1 \leq t_2 \leq T$ be such that $t_2 - t_1 \leq L$ and let $I(t_1,t_2) = \{ t : t_1 \leq t \leq t_2 \}$ denote the interval of time steps between $t_1$ and $t_2$. We 
show that we can simulate measurements on the state $\outerEnc(\Phi_I)$ where 
$\ket{\Phi_{I(t_1,t_2)}}$ is the post-measurement state
\[
	\ket{\Phi_{I(t_1,t_2)}} = \frac{1}{\sqrt{t_2 - t_1+1}} \sum_{t \in I(t_1,t_2)} \ket{\unary(t)}_{\sC} \otimes \ket{\Phi_t}_{\sV \sM \sP \sF}
\]
In other words, $\ket{\Phi_{I(t_1,t_2)}}$ denotes the part of the history state between times $t_1$ and $t_2$. Furthermore, when convenient we will omit mention of the $\unary$ encoding of the clock, and simply refer to the state of the clock register as $\ket{t}$. 

\begin{lemma}\label{lem:sim_density_interval}
There is a PPT algorithm $\SimInterval$ that when given a tuple $\adv{W}$ of questions and a pair of times $0 \leq t_1 \leq t_2 \leq T$ such that $t_2 - t_1 \leq L$, outputs a $(3L^2,4L)$-efficient representation of a density matrix $\rho$ such that
for all answer vectors $a = \Paren{ a^{(D,r)} }_{D \in \{V,P\}, r \in [N_D]}$, we have that
\begin{equation}
\label{eq:sim_density_interval}
\Tr \Paren{ \outerEnc(\Phi_{I(t_1,t_2)}) \, \adv{W}(a) } = \Tr \Paren{ \outerEnc(\rho) \,
  \tilde{W}(a)}.	
\end{equation}
Furthermore, the density matrix $\rho$ is defined on the logical support $S_{\adv{W}}$ of $\adv{W}$.
\end{lemma}
\begin{proof}
Let $I = I(t_1,t_2)$ and $S = S_{\adv{W}}$.
Because of padding, we can assume without loss of generality that the time interval $I$ belongs entirely to one of the six phases of the protocol circuit $V_{enc}$ defined in Section~\ref{sec:transformation}. 

We can write for all $t \in I$,
\begin{equation}
\label{eq:sim_density_interval0}
	\ket{\Phi_t} = \ket{\Delta_t}_{\sV \sM \sP} \otimes \ket{f(t)}_{\sF}.
\end{equation}
Thus, we have
\[
	\ketbra{\Phi_I}{\Phi_I} = \frac{1}{|I|} \sum_{t,t' \in I} \ketbra{t}{t'}_{\sC} \otimes \ketbra{\Delta_t}{\Delta_{t'}}_{\sV \sM \sP} \otimes \ketbra{f(t)}{f(t')}_{\sF}.
\]
The left hand side of~\eqref{eq:sim_density_interval} can be written as
\begin{align}
	&\Tr \Paren{ \outerEnc(\Phi_I) \, \adv{W}(a) } \notag \\
	&= \frac{1}{|I|} \sum_{t,t' \in I} \Tr \Paren{ \outerEnc\Paren{\ketbra{t}{t'} \otimes \ketbra{\Delta_t}{\Delta_{t'}} } \otimes \ketbra{f(t)}{f(t')}\, \adv{W}(a) }
	\label{eq:sim_density_interval1}
\end{align}

\noindent We consider two cases. 

\paragraph{Case 1.} First, suppose that the following holds for all $r \in [N_P]$: either the $r$'th prover flag $p_r(t)$ stays constant throughout the interval $I$, or if it changes, then $\adv{W}^{(P,r)} \neq \star$ (that is, prover player $PP_r$ was not asked a $\star$ question).

Fix a $t,t' \in I$. Let $\adv{W}^{\star}(a)$ denote the tensor factors of $\adv{W}(a)$ corresponding to the prover players who received a $\star$ question (if none received a $\star$ question, then this operator is the identity). Similarly, let $\tilde{W}^\star(a)$ denote the tensor factors of $\tilde{W}(a)$ corresponding to the prover players who received a $\star$ question. Thus, $\tilde{W}^\star(a)$ is tensor product of $\sigma_X$ operators and identity operators. Under our assumption, any operator $A$ defined on registers $\sC \sV \sM \sP$, we have that
\[
	\Tr \Paren{ \Paren{A \otimes \ketbra{f(t)}{f(t')}_{\sF} } \adv{W}^{\star}(a)} = \Tr \Paren{ \Paren{A \otimes \ketbra{f(t)}{f(t')}_{\sF} } \tilde{W}^{\star}(a)}.
\]
This is because of the following: consider the set $J \subseteq [N_P]$ who
  received a $\star$ question. If $J$ is empty, then $\adv{W}^\star =
  \tilde{W}^\star = \Id$,\footnote{We denote the identity matrix as $\Id$ here
  in order to avoid confusion with the interval $I$.} so the equation trivially holds. If $J$ is non-empty, then by assumption for any $r \in J$, the prover flags for $r$ stay constant on the interval $I$, so the traces are $0$.
This implies that~\eqref{eq:sim_density_interval1} is equal to
\[
\frac{1}{|I|} \sum_{t,t' \in I} \Tr \Paren{ \outerEnc\Paren{\ketbra{t}{t'} \otimes \ketbra{\Delta_t}{\Delta_{t'}} } \otimes \ketbra{f(t)}{f(t')}\, \tilde{W}(a) }
\]
We now argue that a $(1,4L)$-efficient representation of the following operator
\[
	\rho_{S}(t,t') = \Tr_{\comp{S}} \Paren{\ketbra{t}{t'} \otimes \ketbra{\Delta_t}{\Delta_{t'}}  \otimes \ketbra{f(t)}{f(t')} }
\]
can be efficiently computed in polynomial time, where $\Tr_{\comp{S}}(\cdot)$ denotes tracing out all registers except those in $S$. Notice that $S$ does not include the register $\sP$, and has at most $O(N_V + k)$ qubit registers.

Given this is true, and using the fact that $|I| \leq L$, then a $(L^2,4L)$-efficient representation of
\[
	\rho_{S} = \frac{1}{|I|} \sum_{t,t' \in I} \rho_{S}(t,t')
\]
can be computed in polynomial time, and satisfies~\eqref{eq:sim_density_interval}. 

Since $\Tr_{\comp{S}} \Paren{\ketbra{t}{t'}}$ and $\Tr_{\comp{S}} \Paren{\ketbra{f(t)}{f(t')}}$ have $(1,1)$-efficient representations (i.e. these are tensor products of single qubit operators), it suffices to show that we can efficiently compute $\Tr_{\comp{S}}(\ketbra{\Delta_t}{\Delta_{t'}})$. Assume without loss of generality that $t \leq t'$. 

First, we consider the sub-case that all prover flags $\{p_r(t)\}$ stay constant throughout the interval $I$. This means that there exists a sequence of elementary gates $g_t,g_{t+1},\ldots,g_{t'}$ (i.e., no prover gates) such that
\[
	\ket{\Delta_{t'}} = g_{t'} g_{t'-1} \cdots g_{t+1} g_t \ket{\Delta_t}.
\]
Let $G$ denote the union of the registers that are acted upon by the gates $g_t,\ldots,g_{t'}$. Since $t' - t \leq |I| \leq L$, and each gate acts on at most $3$ qubits, we get that $|G| \leq 3L$. Now, we can write
\begin{align*}
	\Tr_{\comp{S}}(\ketbra{\Delta_t}{\Delta_{t'}}) &= \Tr_{\comp{S}}(\ketbra{\Delta_t}{\Delta_{t}} g_{t'}^\dagger \cdots g_t^\dagger) \\
	&= \Tr_{G \cap \comp{S}} \Paren{ \Tr_{\comp{S \cup G}}(\ketbra{\Delta_t}{\Delta_{t}}) g_{t'}^\dagger \cdots g_t^\dagger}
\end{align*}
The density matrix $\Tr_{\comp{S \cup G}}(\ketbra{\Delta_t}{\Delta_{t}})$ is where all registers except for $S$ and $G$ are traced out. 
We notice that the number of qubits of this density matrix, $|S \cup G|$, is at most $4L$.
We can appeal to the following Lemma to get that the explicit matrix description of $\Tr_{\comp{S \cup G}}(\ketbra{\Delta_t}{\Delta_{t}})$ can be computed in polynomial time. %

\begin{restatable}{lemma}{simsnapshot}
\label{lem:sim_snapshot}
	There exists a PPT algorithm $\SimSnapshot$ that on input $(x,Y,t)$ such that
	\begin{enumerate}
		\item $x$ is a binary string
		\item $Y$ is a subset of registers used in the honest strategy $\cS_{ZK}(x)$ (that does not include the prover registers $\sP$ nor the prover flags $\sF$) that has size at most
    $4L$, and %
		\item $t$ is an integer between $0$ and the length of the protocol circuit $V_{enc}(x)$
	\end{enumerate}
	outputs matrix entries of the density matrix
	\[
		\Tr_{\comp{Y}} \Paren{ \ketbra{\Delta_t}{\Delta_t}}
	\]
	where $\ket{\Delta_t}$ is defined as in~\eqref{eq:sim_density_interval0}.
\end{restatable}
We defer the proof of Lemma~\ref{lem:sim_snapshot} to Section~\ref{sec:sim_snapshot}.

The simulator $\SimInterval$ can execute $\SimSnapshot$ on $(x,S \cup G,t)$ to obtain the description of $\Tr_{\comp{S \cup G}}(\ketbra{\Delta_t}{\Delta_{t}})$, and then perform some efficient post-processing to obtain the explicit matrix description of $\Tr_{\comp{S}}(\ketbra{\Delta_t}{\Delta_{t'}})$. 

Putting everything together, we get that
\[
	\rho_S(t,t') = \Tr_{\comp{S}} \Paren{\ketbra{t}{t'} } \otimes \Tr_{\comp{S}} \Paren{ \ketbra{\Delta_t}{\Delta_{t'}}}   \otimes \Tr_{\comp{S}} \Paren{\ketbra{f(t)}{f(t')} }
\]
has a $(1,4L)$-efficient representation.
\medskip \\
Next, we consider the next sub-case, where the prover flags $\{p_r(t)\}$ do not stay constant. The interval $I$ lies within the Prover Operation phase. Because of padding, the interval $I$ can at most cover a single prover's operation, so there exists a unique $r^* \in [N_P]$ such that $p_{r^*}(t)$ changes (all others stay constant). Thus, the $p_{r^*}$ flag changes from $0$ to $1$ at time $t_\star(r^*)$.  Let $t_\star = t_\star(r^*)$.

By our assumption at the beginning, $\adv{W}^{(P,r^*)} \neq \star$ (the prover player $PP_{r^*}$ was not asked the $\star$ question). Since $S$ does not include $\sF_{P_{r^*}}$, for $t < t_\star$ and $t' \geq t_\star$, we get that
\[
	\Tr_{\comp{S}} \Paren{\ketbra{f(t)}{f(t')}} = 0,
\]
Define $I^- = \{ t \in I : t < t_\star \}$ and $I^+ = \{ t \in I : t \geq t_\star \}$. We then have
\[
	\rho_{S} = \frac{1}{|I|} \Paren{\sum_{t,t' \in I^-} \rho_{S}(t,t') + \sum_{t,t' \in I^+} \rho_{S}(t,t') } .
\]
Notice that all prover flags $p_r(t)$ stay constant on $I^-$ and $I^+$. Therefore we can reduce to the previous sub-case to argue that $\rho_{S}(t,t')$ can be computed when both $t,t'$ either come from $I^-$ or $I^+$.

This completes the proof of Case 1.

\paragraph{Case 2.}
Next, we consider the case that there is an $r^*$ for which the prover flag $p_{r^*}(t)$ changes from $0$ to $1$ during the interval $I$, and furthermore $\adv{W}^{(P,r^*)} = \star$ (the prover player $PP_{r^*}$ was asked the $\star$ question). Again, this interval $I$ must lie in the Prover Operation phase and by padding all other prover flags must be constant throughout the interval $I$. Let $t_\star = t_\star(r^*)$.

Since prover player $PP_{r^*}$ received the $\star$ question, they could not have received questions $QF_{r^*}$ or $AF_{r^*}$, and therefore $\sF_{Q_{r^*}}$ and $\sF_{A_{r^*}}$ are not part of the logical support set $S$. Thus the reduced density matrix of $\Phi_I$ where we trace out all registers except for $S$ and $\sP$ is a convex combination 
\[
\Tr_{\comp{S \cup \sP}}(\Phi_I) =  \Tr_{\comp{S \cup \sP}} \Paren{ \frac{|I^-|}{|I|} \Phi_{I^-} + \frac{|I^\star|}{|I|} \Phi_{I^\star} + \frac{|I^+|}{|I|} \Phi_{I^+}}
\]
where we define the subintervals
\begin{itemize}
	\item $I^- = \{t \in I : t < t_\star - 1 \}$
	\item $I^\star = \{t \in I : t_\star - 1\leq t \leq t_\star \}$
	\item $I^+ = \{t \in I : t \geq t_\star + 1 \}$
\end{itemize}
This is because we are tracing out the Question Flip flag register $\sF_{Q_{r^*}}$ and Answer Flip flag register $\sF_{A_{r^*}}$; so cross-terms where $t$ and $t'$ belong to different subintervals above would disappear. 

Therefore~\eqref{eq:sim_density_interval1} is equal to
\begin{align*}
	\Tr \Paren{ \outerEnc(\Phi_I) \, \adv{W}(a)} &=  \Tr \Paren{ \outerEnc\Paren{ \frac{|I^-|}{|I|} \Phi_{I^-} + \frac{|I^\star|}{|I|} \Phi_{I^\star} + \frac{|I^+|}{|I|} \Phi_{I^+}}
 \, \adv{W}(a)}
\end{align*}
We now show how to compute $(L^2,4L)$-efficient representations of density matrices $\rho_S^-,\rho_S^+,\rho_S^\star$ such that
\begin{align}
\label{eq:sim_density_interval4}
	\Tr(\outerEnc(\rho_S^-) \, \tilde{W}(a)) = \Tr(\outerEnc(\Phi_{I^-}) \, \adv{W}(a)) \\
\label{eq:sim_density_interval5}	
	\Tr(\outerEnc(\rho_S^+) \, \tilde{W}(a)) = \Tr(\outerEnc(\Phi_{I^+}) \, \adv{W}(a)) \\
\label{eq:sim_density_interval6}	
	\Tr(\outerEnc(\rho_S^\star) \, \tilde{W}(a)) = \Tr(\outerEnc(\Phi_{I^\star}) \, \adv{W}(a)).
\end{align}
Once we have this, then a $(3L^2,4L)$-efficient representation of density matrix $\rho_S = \frac{|I^-|}{|I|} \rho_S^- + \frac{|I^\star|}{|I|} \rho_S^\star + \frac{|I^+|}{|I|} \rho_S^+$ is efficiently computable and satisfies~\eqref{eq:sim_density_interval}, and this completes the proof of Case 2.

We argue that $\rho_S^-$ and $\rho_S^+$ have efficient representations. Notice that the prover flags $\{ p_r(t) \}$ are constant on the intervals $I^-$ and $I^+$. Thus from the same arguments as in Case 1, $\SimInterval$ can, when given input $\adv{W}$ and a pair of times $(\min(I^-),\max(I^-))$, efficiently compute a $(L^2,4L)$-efficient representation of the density matrix $\rho_S^-$ defined on $S$ that satisfies~\eqref{eq:sim_density_interval4}. Similarly, $\SimInterval$ can also efficiently compute an efficient representation of $\rho_S^+$ that satisfies~\eqref{eq:sim_density_interval5}.
 
We now turn to $\rho_S^\star$. Since we are in Case 2, it must be that $I^\star = \{ t_\star - 1, t_\star \}$ (otherwise, the prover flag for $PP_{r^*}$ would stay constant on $I$). Thus, using that $\ket{\Phi_t} = \ket{\Delta_t} \otimes \ket{f(t)}$,
\begin{align}
	\ket{\Phi_{I^\star}} &= \frac{1}{\sqrt{2}} \Bigbrac{ \ket{t_\star - 1}  \ket{\Delta_{t_\star - 1}}  \ket{f(t_\star -1)} + \ket{t_\star}  \ket{\Delta_{t_\star}}  \ket{f(t_\star)}} \notag \\
	&= \frac{1}{\sqrt{2}} \Bigbrac{ \ket{t_\star - 1}  \otimes \ket{\Phi_{t_\star - 1}} + \ket{t_\star} \otimes  P_{r^*}' \ket{\Phi_{t_\star - 1}}}
	\label{eq:sim_density_interval2}
\end{align}
Furthermore, since $PP_{r^*}$ receives the $\star$ question in $\adv{W}$, the measurement operator $\adv{W}^{(P,r^*)}(a)$ is simply
\[
	\frac{1}{2} \Paren{ \Id + (-1)^{a^{(P,r^*)}} P_{r^*}'}.
\]

Let $\adv{W}^{-(P,r*)}(a)$ be the measurement operator obtained by taking
$\adv{W}(a)$ and deleting the factor $\adv{W}^{(P,r^*)}$, i.e.,
$\adv{W}^{-(P,r*)}(a)$ is the tensor product of questions of all provers except
$PP_{r^*}$. Therefore we can write
\begin{align}
	&\Tr \Paren{ \outerEnc(\Phi_{I^\star}) \, \adv{W}(a) } \notag \\
	&= \Tr \Paren{ \outerEnc(\Phi_{I^\star}) \, \adv{W}^{-(P,r^*)}(a) \otimes \adv{W}^{(P,r^*)}(a) } \notag \\
	&= \frac{1}{2} \Paren{ \Tr \Paren{ \outerEnc(\Phi_{I^\star}) \, \adv{W}^{-(P,r^*)}(a)} + (-1)^{a^{(P,r^*)}} \Tr \Paren{ \outerEnc(\Phi_{I^\star}) \, \adv{W}^{-(P,r^*)}(a) \otimes P_{r^*}' }}
	\label{eq:sim_density_interval3}
\end{align}
We analyze the first term above. By substituting in the expression~\eqref{eq:sim_density_interval2} for $\Phi_{I^\star}$, we get some cross terms of the form
\begin{align*}
	&\Tr \Paren{ \outerEnc(\ketbra{t_\star - 1}{t_\star} \otimes \ketbra{\Phi_{t_\star - 1}}{\Phi_{t_\star -1}} (P_{r^*}')^\dagger) \, \adv{W}^{-(P,r^*)}(a)}  \\
	&= \Tr \Paren{ \Bigbrac{\outerEnc(\ketbra{t_\star - 1}{t_\star} \otimes \ketbra{\Delta_{t_\star - 1}}{\Delta_{t_\star -1}} P_{r^*}^\dagger) \otimes \ketbra{f(t_\star-1)}{f(t_\star)} }\, \adv{W}^{-(P,r^*)}(a)}
\end{align*}
Notice that the operator $\adv{W}^{-(P,r^*)}(a)$ does not act on the prover flag register $\sF_{P_{r^*}}$, and the $\sF_{P_{r^*}}$ component of $\ketbra{f(t_\star-1)}{f(t_\star)}$ is $\ketbra{0}{1}$. Thus, the cross-term vanishes. The first term of~\eqref{eq:sim_density_interval3} can be written as
\begin{align*}
	&\frac{1}{2} \Tr \Paren{ \outerEnc( \ketbra{t_\star - 1}{t_\star - 1} \otimes \ketbra{\Phi_{t_\star - 1}}{\Phi_{t_\star -1}} + \ketbra{t_\star}{t_\star} \otimes \ketbra{\Phi_{t_\star}}{\Phi_{t_\star}} )\, \adv{W}^{-(P,r^*)}(a) } \\
	&=\frac{1}{2} \Tr \Paren{ \outerEnc \Paren{ (\ketbra{t_\star - 1}{t_\star - 1} + \ketbra{t_\star}{t_\star}) \otimes \ketbra{\Phi_{t_\star-1}}{\Phi_{t_\star-1}} }\, \adv{W}^{-(P,r^*)}(a) }
\end{align*}
where in the equality we used that $\ket{\Phi_{t_\star}} = P_{r^*}'
\ket{\Phi_{t_\star -1}}$ and the operator $P_{r^*}'$ commutes with
$\adv{W}^{-(P,r^*)}(a)$, and thus vanishes by the cyclity of the trace. Applying
similar reasoning to the second term of~\eqref{eq:sim_density_interval2}, we
remain only with the cross terms, and we get that it can be written as
\[
\frac{1}{2} \Tr \Paren{ \outerEnc \Paren{ (\ketbra{t_\star - 1}{t_\star} + \ketbra{t_\star}{t_\star-1}) \otimes \ketbra{\Phi_{t_\star-1}}{\Phi_{t_\star-1}} }\, \adv{W}^{-(P,r^*)}(a) }
\]
Putting everything together, we get that~\eqref{eq:sim_density_interval3} can be written as
\begin{align}
\label{eq:sim_density_interval7}
\frac{1}{2} \Tr \Paren{ \outerEnc \Paren{ \tau(a,t_\star) \otimes \ketbra{\Delta_{t_\star -1}}{\Delta_{t_\star -1}} \otimes \ketbra{f(t_\star-1)}{f(t_\star -1)} }  \, \adv{W}^{-(P,r^*)}(a)}
\end{align}
with
\[
	\tau(a,t_\star) = \frac{1}{2} \Bigbrac{ \ket{t_\star-1} + (-1)^{a^{(P,r^*)}} \ket{t_\star}}\Bigbrac{ \bra{t_\star-1} + (-1)^{a^{(P,r^*)}} \bra{t_\star}}.
\]
Define
\[
	\rho_S^\star = \Tr_{\comp{S}} \Paren{\tau(a,t_\star) \otimes \ketbra{\Delta_{t_\star -1}}{\Delta_{t_\star -1}} \otimes \ketbra{f(t_\star-1)}{f(t_\star -1)}} .
\]
Just like in Case 1, the density matrices $\Tr_{\comp{S}} \Paren{\tau(a,t_\star)}$ and $\Tr_{\comp{S}} \Paren{\ketbra{f(t_\star-1)}{f(t_\star -1)}}$ have $(1,1)$-efficient representations, and by Lemma~\ref{lem:sim_snapshot} we have that $\Tr_{\comp{S}} \Paren{\ketbra{\Delta_{t_\star -1}}{\Delta_{t_\star -1}}}$ can be efficiently computed as well. This shows that $\rho_S^\star$ has a $(1,4L)$-efficient representation. Finally, we have that the $\adv{W}^{-(P,r^*)}(a)$ operator in~\eqref{eq:sim_density_interval7} can be replaced with $\tilde{W}^{-(P,r^*)}(a)$. This shows that $\rho_S^\star$ satisfies~\eqref{eq:sim_density_interval6}, and this completes the proof of Case 2.

\end{proof}

\noindent We now prove \Cref{lem:density-matrix}.
\begin{proof}[Proof of Lemma~\ref{lem:density-matrix}]
Fix a tuple $\adv{W}$ of questions. We argue that computing an efficient description of a density matrix $\rho$ that satisfies~\eqref{eq:sim_density} can be efficiently reduced to computing efficient descriptions of density matrices $\rho_I$ for various intervals $I$, for which we can use the algorithm $\SimInterval$ from Lemma~\ref{lem:sim_density_interval}. 

Since there are only $N_V$ verifier players, and each verifier player receives a $6$-tuple of Pauli observables that have support on at most $12$ physical qubits each, the joint measurement of the verifier players acts on at most $12N_V$ physical qubits, and therefore at most $12N_V$ logical qubits of the underlying encoded clock register.

Let
\[
	C_{tr} = \{ i \in [T] : \text{the $i$'th logical clock qubit $\sC_i$ is not in $S_{\adv{W}}$} \}
\]
denote the set of (logical) clock qubit registers that, after the outer encoding, are not acted upon by the measurement corresponding to $\adv{W}$. Thus, for all answer vectors $a$, 
\begin{equation}
	\Tr \Paren{ \outerEnc(\Phi) \, \adv{W}(a) } = \Tr \Paren{ \outerEnc(\Tr_{C_{tr}}( \Phi)) \, \adv{W}(a) }.
\end{equation}
We argue that the density matrix $\Tr_{C_{tr}}( \Phi)$ is a convex combination
  of $\ket{\Phi_I}$ states for various intervals $I$:
\begin{align}
\label{eq:sim_density_traced_out}
\Tr_{C_{tr}}( \Phi) &= \frac{1}{T+1} \sum_{t,t'} \Tr_{C_{tr}}(\ketbra{\unary(t)}{\unary(t')}) \otimes \ketbra{\Phi_t}{\Phi_{t'}}.
\end{align}
The following Claim easily follows from the structure of unary encodings:
\begin{claim}
	For all $ 0 \leq t,t' \leq T$, the operator $\Tr_{C_{tr}}(\ketbra{\unary(t)}{\unary(t')})$ is non-zero only when $t = t'$, or for all $i \in C_{tr}$, either both $t, t' > i$, or both $t, t' < i$.
\end{claim}

Given this Claim, we notice that all cross-terms of~\eqref{eq:sim_density_traced_out} involving times $t,t'$ where $t \neq t'$ and at least one of $t,t'$ are in $C_{tr}$ vanish. Thus the only cross-terms that remain are times $t,t'$ that come from an interval $I \subseteq \{0,1,2,\ldots,T\}$ of consecutive time-steps where there is no $i \in C_{tr}$ such that $\min(I) \leq i \leq \max(I)$. Let $\{0,1,2,\ldots,T\} \setminus C_{tr}$ be the union of maximal intervals  $I_1,I_2,\ldots,I_\ell$ of consecutive time steps. Thus~\eqref{eq:sim_density_traced_out} can be written as
\[
	 \sum_{t \in C_{tr}} \frac{1}{T+1} \Tr_{C_{tr}}(\ketbra{\Phi_{\{t\}}}{\Phi_{\{t\}}}) + \sum_{j=1}^\ell \frac{|I_j|}{T+1} \Tr_{C_{tr}} (\ketbra{\Phi_{I_j}}{\Phi_{I_j}}) 
\]
where $\ket{\Phi_{\{t\}}} = \ket{\unary(t)} \otimes \ket{\Phi_t}$ denotes the history state restricted to the singleton interval $\{t\}$. As desired,~\eqref{eq:sim_density_traced_out} is a probabilistic mixture of interval states $\ket{\Phi_I}$ where each interval has size at most $6N_V \leq L$. The intervals $I_j$ occur with probability $|I_j|/(T+1)$ and the singleton intervals $\{t\}$ for $t \in C_{tr}$ occur with probability $1/(T+1)$.

The algorithm $\SimDensity$ works as follows: given a question tuple $\adv{W}$ it can compute the set $C_{tr}$, and then compute the intervals $I_1,\ldots,I_j$ in polynomial time. For each interval $I_j$, it invokes the algorithm $\SimInterval$ from Lemma~\ref{lem:sim_density_interval} to efficiently compute a $(3L^2,4L)$-efficient representation of the density matrix $\rho_{I_j}$ supported on $S_{\adv{W}}$ that satisfies
\[
	\Tr \Paren{ \outerEnc(\Phi_{I_j}) \, \adv{W}(a) } = \Tr \Paren{
    \outerEnc(\rho_{I_j}) \, \tilde{W}(a) }.
\]
Similarly, for every $t \in C_{tr}$ the algorithm $\SimDensity$ invokes $\SimInterval$ to compute a $(3L^2,4L)$-efficient representation of the density matrix $\rho_t$ that satisfies
\[
	\Tr \Paren{ \outerEnc(\Phi_{\{t\}}) \, \adv{W}(a) } = \Tr \Paren{
    \outerEnc(\rho_{t}) \, \tilde{W}(a) }.
\]
There are at most $T+1$ density matrices to compute. $\SimDensity$ then can then efficiently compute a $(3(T+1) L^2,4L)$-efficient representation of the convex combination
\[
	\rho = \sum_{t \in C_{tr}} \frac{1}{T+1} \rho_t + \sum_{j=1}^\ell \frac{|I_j|}{T+1} \rho_{I_j},
\]
which satisfies~\eqref{eq:sim_density}.
\end{proof}

With Lemma~\ref{lem:density-matrix}, we prove that $V_{ZK}$ has the zero knowledge property against cheating referees that are non-adaptive, meaning that the referee samples a question tuple $\adv{W}$ first, sends them to the players, and receives their answers $a$. 

\begin{lemma}
\label{lem:non-adaptive}
  For every non-adaptive polynomial-time referee $\advR^{na}$, there is a
  PPT simulator $\Sim_{\advR^{na}}$
  such that the output distribution of $\Sim_{\advR^{na}}(x)$
  is equal to $\View(\adv{R}^{na}(x) \leftrightarrow \cS_{ZK}(x))$.
\end{lemma}
\begin{proof}
$\SimRna$ starts by sampling the questions to the players
  $\adv{W} = \left(\adv{W}^{(D,r)}\right)_{D \in \{V,P\}, r \in [N_D]}$
  from the same joint distribution as
$\advRna$ on input $x$. This can be performed efficiently since $\advRna$ is a
  polynomial-time algorithm and the questions are sampled in a non-adaptive way. 

  Then, the simulator $\Sim_{\adv{R}^{na}}$ executes the algorithm $\SimDensity$ from \Cref{lem:density-matrix} on input $\adv{W}$, which outputs an efficient representation of a density matrix $\rho$
  such that
for all answer vectors $a = \Paren{ a^{(D,r)} }_{D \in \{V,P\}, r \in [N_D]}$ we have that
\[
\alpha(a) = \Tr \Paren{ \outerEnc(\Phi) \, \adv{W}(a) } = \Tr \Paren{ \outerEnc(\rho) \,
  \tilde{W}(a)}.	
\]
Note that $\alpha(a)$ is a probability distribution over answer vectors. We need to show that we can efficiently sample an answer vector $a$ from the probability distribution $\alpha(a)$. We can do that by sampling each bit of $a$ one at a time, and conditioning the density matrix $\rho$ on the partial outcomes. 

Index the players using $\{1,2,\ldots,N_V + k\}$  in some canonical way. Let $a = (a_1,\ldots,a_{N_V + k})$ where $a_i$ denotes the answer symbol of the $i$'th player, which might come from the alphabet $\{0,1,\}^6$ or $\{0,1\}$, depending on whether the $i$'th player is a prover player or a verifier player. 

We utilize the following important observation: for every answer vector $a$, $\tilde{W}(a)$ is equal to the tensor product of projectors where the projectors corresponding to the prover players are all single-qubit operators, and the projectors corresponding to the verifier players may act on up to $12N_V$ qubits. 

For every $i \in \{1,2,\ldots,N_V + k\}$, let $\tilde{W}(a_i)$ denote the projector of the $i$'th player corresponding to outcome $a_i$, when the players receive the question tuple $\adv{W}$. Note that $\tilde{W}(a) = \tilde{W}(a_1) \otimes \cdots \otimes \tilde{W}(a_{N_V + k})$.

To sample $a_1$, the simulator can explicitly compute the probabilities
\[
\alpha(a_1) = \Tr \Paren{ \outerEnc(\Phi) \, \adv{W}(a_1) } = \Tr \Paren{ \outerEnc(\rho) \,
  \tilde{W}(a_1)}.
\]
for all $a_1$, where we use $\alpha(a_1)$ to denote the marginal distribution of $a_1$ in $\alpha$. Since $a_1$ comes from a constant-sized alphabet, this distribution can be sampled from in polynomial time. Given a sample $a_1$, we can now sample $a_2$ \emph{conditioned} on $a_1$, so we can compute the conditional distribution
\[
	\alpha(a_2 | a_1) = \frac{\Tr \Paren{ \outerEnc(\rho) \, \tilde{W}(a_1) \otimes \tilde{W}(a_2)}}{\alpha(a_1)},
\]
and sample from it as well. We can continue in this manner, until we have sampled $a_1 \cdots a_{N_V + k}$. This can be done in polynomial time, because $\tilde{W}(a_1) \otimes \cdots \otimes \tilde{W}( a_i)$ for all $i \in \{1,\ldots,N_V + k\}$ has a $(1,12N_V)$-efficient representation.

 The simulator then outputs $(x,r,\adv{W},a)$ where $r$ is the randomness used by cheating referee $\adv{R}^{na}$. By construction, this output is distributed identically to $\View(\adv{R}^{na}(x) \leftrightarrow \cS_{ZK}(x))$.
\end{proof}

\subsection{General cheating referees}
\label{sec:general-adversary}

We now show that if that for an arbitrary cheating referee $\advR$, there exists simulator whose
output is distributed according to $\View(\hat{R}(x) \leftrightarrow \cS_{ZK}(x))$.

As mentioned earlier, the difficulty is that $\advR$ could send questions to a set of players,
and then depending on their answers, adaptively choose questions for another set of players, and so on. The arguments from \Cref{sec:non-adaptive}
strongly rely on the fact that the simulator can sample all of the questions before
sampling the answers. In this section, we show how to simulate the interaction between the
referee and the players in the adaptive scenario.

\begin{lemma}
  For every PPT $\advR$, there exists a PPT
  simulator $\Sim_{\advR}$
  such that the output distribution of $\SimR(x)$
  is equal to $View(\hat{R}(x) \leftrightarrow \cS_{ZK}(x))$.
\end{lemma}
\begin{proof}
A general cheating referee $\adv{R}$ behaves as follows: using randomness, it samples a set of players $B_1 \subseteq \cP = \{ (D,r) : D \in \{V,P\}, r \in [N_D] \}$, followed by some questions $\adv{W}^{B_1}$ for those players. It sends $\adv{W}^{B_1}$ to the $B_1$ players, and receives a partial answer vector $a^{B_1}$. Based on its randomness and the answers received, the referee samples another set of players $B_2 \subseteq \cP \setminus B_1$ and questions $\adv{W}^{B_2}$ for the $B_2$ players. We assume that $B_2$ is disjoint from $B_1$ because the players would abort the protocol if they are interacted more than once. The referee continues in this manner until it halts.  

The general simulator $\Sim$ runs the referee $\adv{R}$ on randomness $s$ to
  obtain the sample $(B_1, \adv{W}^{B_1})$. To simulate the $B_1$ players'
  responses to $\adv{W}^{B_1}$, the simulator will arbitrarily complete
  $\adv{W}^{B_1}$ to a question tuple $\adv{W}_1$ for all players, and then call
  $\SimDensity$ on $\adv{W}_1$ to obtain a density matrix $\rho_1$ defined on
  registers $S_{\adv{W}_1}$. With this density matrix, the simulator $\Sim$ can
  sample a partial answer vector $a^{B_1}$ with probability $\Tr(
  \outerEnc(\rho_1) \, \adv{W}^{B_1}(a^{B_1}))$. This partial answer vector can
  be sampled in the same way as described in the simulation for the non-adaptive
  referee in Lemma~\ref{lem:non-adaptive}. Note that this distribution does not
  depend how the question tuple $\adv{W}^{B_1}$ was completed, since the
  distribution is non-signalling.

Based on this sampled answer vector $a^{B_1}$ and the randomness $s$, the simulator can continue executing $\adv{R}$ to obtain a sample $(B_2, \adv{W}^{B_2})$. The simulator then constructs a question tuple $\adv{W}_2$ that contains both $\adv{W}^{B_1}$ and $\adv{W}^{B_2}$ (which are question tuples to disjoint sets of players), and invokes $\SimDensity$ to efficiently compute a density matrix $\rho_2$ defined on registers $S_{\adv{W}_2}$. The simulator can then sample a partial answer vector $a^{B_2}$ with probability
\[
	\frac{\Tr \Paren{ \adv{W}^{B_2}(a^{B_2}) \otimes \adv{W}^{B_1}(a^{B_1})\, \outerEnc(\rho_2) }}{\Tr( \outerEnc(\rho_2) \, \adv{W}^{B_1}(a^{B_1}))}.
\]
Once again, this partial answer vector can be sampled in the same way as described in the proof of Lemma~\ref{lem:non-adaptive}.
In the end, the simulator can repeat this process and obtain a sequence $(x,s,\adv{W}^{B_1},a^{B_1},\adv{W}^{B_2},a^{B_2},\ldots)$ that is distributed identically to $\View(\adv{R}(x) \leftrightarrow \cS_{ZK}(x))$.

The complete simulation algorithm is described in detail in Figure~\ref{fig:sim}. It is easy to see that the simulator runs in polynomial time. 
\end{proof}

\begin{figure}[H]
\rule[1ex]{\textwidth}{0.5pt}
   \textbf{Algorithm:} $\Sim_{\adv{R}}$(x) \\
\begin{enumerate}
	\item Set $i = 1$.
	\item Sample randomness $s$ for $\adv{R}$.
	\item Set $\pi = (x,s)$.
  \item While $\adv{R}$ has not halted:
  \begin{enumerate}
  	\item Continue the execution of the referee $\adv{R}$ on randomness $s$, the previous $i-1$ samples $(B_1,\adv{W}^{B_1},a^{(B_1)}),\ldots,(B_{i-1},\adv{W}^{B_{i-1}},a^{(B_{i-1})})$, to obtain a new sample $(B_i,\adv{W}^{B_i})$. If $B_i$ has non-zero intersection with any of the $B_1,\ldots,B_{i-1}$, add $\mathsf{abort}$ to the end of $\pi$ and output $\pi$.
	\item Let $\adv{W}_i$ denote the question tuple that is the concatenation of $\adv{W}^{B_1},\adv{W}^{B_2},\ldots,\adv{W}^{B_i}$ with arbitrary questions to the players in $\cP \setminus (B_1 \cup \cdots \cup B_i)$. 
	\item Execute $\SimDensity$ on input $\adv{W}_i$ to obtain a $(O(T),O(1))$-efficient representation of the density matrix $\rho_i$ supported on registers $S_{\adv{W}_i}$.
	\item Sample $a^{B_i}$ with probability
	\[
		\frac{\Tr \Paren{ \Pi_{i-1} \otimes \adv{W}^{B_i}(a^{B_i}) \, \outerEnc(\rho_i)}}{\Tr \Paren{ \Pi_{i-1} \, \outerEnc(\rho_i)}}
	\]
	where
	\[
		\Pi_{i-1} = \adv{W}^{B_1}(a^{B_1}) \otimes \cdots \otimes \adv{W}^{B_{i-1}}(a^{B_{i-1}}).
	\]
	\item Add $(\adv{W}^{B_i},a^{B_i})$ to the end of $\pi$.
	\item Set $i = i + 1$.
  \end{enumerate}
  \item Output $\pi$.
 \end{enumerate}
\rule[2ex]{\textwidth}{0.5pt}\vspace{-.5cm}
\caption{The simulator $\Sim_{\adv{R}}$}\label{fig:sim}
\end{figure}

\subsection{Simulating snapshots}
\label{sec:sim_snapshot}
We now prove Lemma~\ref{lem:sim_snapshot}. For convenience we recall the Lemma statement.

\simsnapshot*

\begin{proof}

Fix the protocol circuit $V_{enc} = V_{enc}(x)$. For convenience, we omit
  mention of the input $x$ for the remainder of this proof. Let $T$ denote the
  length of the circuit $V_{enc}$. We notice that our parameters imply that $4L
  = 192$,  and therefore $\innerCode$ is a
  $4L$-simulatable code and we denote $m$ as the blocklength of this code, as
  defined in \Cref{T:simulatable}.

The protocol circuit acts on registers $\sA, \sB, \sO, \sN, \sM, \sP$. Since the set $Y$ does not include any subregister of the prover register $\sP$, we only consider the subregisters of $\sR = \sA\sB\sO\sN\sM$. At each time $t$, we say that a group of $m$ qubit registers $\sR_{i_1},\ldots,\sR_{i_m}$ form an \emph{encoded block} if and only if $\Pi \ket{\Delta_t} = \ket{\Delta_t}$ where $\Pi$ is the projector onto the codespace for the qubits $\sR_{i_1},\ldots,\sR_{i_m}$. Since the protocol circuit $V_{enc}(x)$ can be computed in polynomial time, determining the encoded block of qubits than a physical qubit belongs to can be efficiently done.

  As explained in \Cref{sec:micro-phases}, we split the phases of the circuit into
  micro-phases.
For every time $t \in \{0,1,2,\ldots,T\}$, let $start(t) \leq t$ denote the start of the micro-phase containing time $t$, and let $end(t) \geq t$ denote the end of the micro-phase containing time $t$. For each time $t$, we can partition the qubit subregisters into three categories:
\begin{itemize}
	\item \textbf{Active}: These are qubits that have been acted upon by a gate $g_{t'}$ for some time $t' \in \{ start(t),\ldots,t \}$. Let $\cA(t)$ denote the set of active qubit registers at time $t$.
	\item \textbf{Encoded qubits}: These are qubits that belong to an encoded block, and are not active. Let $\cE(t)$ denote the set of encoded qubit registers at time $t$.
	\item \textbf{Unencoded qubits}: These are unencoded ancilla qubits in the state $\ket{0}$ or in the state $\ket{1}$, and are not active. Let $\cU_0(t)$ and $\cU_1(t)$ denote the sets of unencoded qubit registers in the state $\ket{0}$ an $\ket{1}$, respectively, at time $t$. %
\end{itemize}
Unencoded qubits are in a ``known'' state throughout the entire circuit $V_{enc}$ in the sense that their state is independent of the input $x$. In fact, for all $t$ the state $\ket{\Delta_t}$ can be written as
\[
	\ket{\Delta_t} = \ket{\Sigma_t}_{\cA(t) \cE(t)} \otimes \ket{0 \cdots 0}_{\cU_0(t)} \otimes \ket{1 \cdots 1}_{\cU_1(t)}
\]
where $\ket{\Sigma_t}$ corresponds to the registers that are either active or encoded, and the remaining qubits are unencoded ancillas.

By construction, the protocol circuit $V_{enc}$ satisfies the following invariant: at the beginning and end of every micro-phase of the circuit, all qubit subregisters are either encoded, or unencoded. Qubits can only be active within a micro-phase. 

We now argue that the description of $\Tr_{\comp{Y}}(\ketbra{\Delta_t}{\Delta_t})$ can be efficiently computed for all $t$. We argue this for each micro-phase separately. Let $t_0 = start(t)$.

\paragraph{Idling phase} During an idling phase, all qubits are either encoded
  or unencoded, and none are active. The reduced density matrix $\Tr_{\comp{Y}}
  \Paren{ \ketbra{\Delta_t}{\Delta_t}}$ thus consists of either at most $|Y|$
  unencoded $\ket{0}$ and $\ket{1}$ ancilla qubits, and the reduced density
  matrix of some encoded blocks on at most $|Y| \leq 4L$ qubits. By Theorem~\ref{T:simulatable}, the reduced density matrix of the encoded blocks is efficiently computable, and thus $\Tr_{\comp{Y}} \Paren{ \ketbra{\Delta_t}{\Delta_t}}$ is efficiently computable.

\paragraph{Resource encoding} In a resource encoding phase, a constant number of unencoded ancilla bits in $\ket{\Delta_{t_0}}$ will be transformed into an encoded resource state in $\ket{\Delta_{end(t)}}$, and the rest of the qubits are either in an encoded block or unencoded ancilla qubits. Thus the reduced density matrix $\Tr_{\comp{Y}} \Paren{ \ketbra{\Delta_t}{\Delta_t}}$ is a tensor product of the reduced density matrix of some encoded blocks (which is efficiently computable by Theorem~\ref{T:simulatable}), unencoded ancilla qubits, and the reduced density matrix of the intermediate state of a resource encoding circuit acting on a constant number of ancillas (which is efficiently computable). Thus $\Tr_{\comp{Y}} \Paren{ \ketbra{\Delta_t}{\Delta_t}}$ is efficiently computable.

  \paragraph{Logical operation} %
  In a logical operation micro-phase, either a logical Hadamard, logical CNOT, or logical Toffoli are being implemented on some encoded code blocks as well as some unencoded ancilla qubits. Let $U \in \{ H, \Lambda(X), \Lambda^2(X) \}$ be the logical gate,
  and $O_1, O_2, \ldots, O_{t - t_0}$ denote the first $t - t_0$ gates of the encoding of $U$.
  We have that 
\[
  \ket{\Delta_t} = O_{t - t_0} \cdots O_1 \ket{\Delta_{t_0}}.
\]

Since all qubits of $\ket{\Delta_{t_0}}$ are  correctly encoded, this
corresponds to the simulation in the middle of the application of a logical
gate, and again by \Cref{T:simulatable}, $\tr_{\overline{Y}}(\kb{\Delta_{t}})$ can
also be efficiently computable.

\paragraph{Output decoding} In the honest strategy $\cS_{ZK}(x)$, the state $\ket{\Delta_{t_0}}$ can be written as a tensor product
\[
	\ket{\Delta_{t_0}} = \innerEnc(\ket{1})_{\sO} \otimes \ket{\Sigma_{t_0}}_{\cE(t_0)} \otimes \ket{0 \cdots 0, 1 \cdots 1}_{\cU(t_0)}.
\]
This is because by assumption the strategy $\cS_{ZK}(x)$ causes the protocol circuit $V_{enc}$ to accept with probability $1$, and therefore the register $\sO$ at the beginning of the Output Decoding phase will store an encoding of $\ket{1}$.

Therefore, the reduced density matrix $\Tr_{\comp{Y}} \Paren{
  \ketbra{\Delta_t}{\Delta_t}}$ is a tensor product of the reduced density
matrix of a decoding circuit acting on $\innerEnc(\ket{1})$ (which is
efficiently computable), the reduced density matrix of $\ket{\Sigma_{t_0}}$ on
at most $|Y| \leq 4L$ qubits (which is efficiently computable by Theorem~\ref{T:simulatable}), and a constant number of unencoded ancilla qubits. Thus $\Tr_{\comp{Y}} \Paren{ \ketbra{\Delta_t}{\Delta_t}}$ is efficiently computable.

\end{proof}

\subsection{Completing the proof of Theorem~\ref{thm:main}}

If the completeness and soundness of the original $\MIP^*$ protocol for $L$ are $1$ and $s$ respectively, then the soundness of the resulting zero knowledge protocol $V_{ZK}$ for $L$ has completeness $1$ (i.e. perfect completeness) and has soundness $s'$ that is polynomially related to $1-s$:
\[
	s' \leq 1 - \frac{(1 - s)^\beta}{p(n)}
\]
for some universal constant $\beta$ and polynomial $p$.

Our zero knowledge transformation is not immediately gap preserving, in the sense that if $1-s$ is a  constant, the new soundness $s'$ is only separated from $1$ by an inverse polynomial. Since the standard definition of the complexity classes $\QMIP$, $\MIP^*$, and $\PZKMIP^*$ have constant completeness-soundness gaps, our result does not immediately show that $\MIP^* \subseteq \QMIP \subseteq \PZKMIP^*$.

To remedy this, we employ the gap amplification techniques described in~\Cref{sec:parallel-repetition}. Suppose that the soundness $s'$ of $V_{ZK}$ is at most $1 - 1/q$ for some polynomial $q$. First, we apply the anchoring transformation to $V_{ZK}$ to obtain a new protocol $V_{ZK,\bot}$ such that
\[
	\omega^*(V_{ZK,\bot}) = \alpha + (1- \alpha)\omega^*(V_{ZK})
\]
for some constant $\alpha$. Then, we use
Theorem~\ref{thm:parallel-repetition-anchored} of Bavarian, Vidick and
Yuen~\cite{BavarianVY17} to argue that the parallel repetition of $V_{ZK,\bot}$
has the desired soundness properties. In the case that $\omega^*(V_{ZK}) = 1$,
then $\omega^*(V_{ZK,\bot}^m) = 1$ for all $m$. Otherwise, for some polynomial
$m$ that depends on $q$, $k$, $\alpha$, and $V_{ZK}$, we have that
$\omega^*(V_{ZK,\bot}^m) \leq 1 - (1 - s)^{\gamma}$ for some universal constant $\gamma$. Thus, the soundness of $V_{ZK,\bot}^m$ is polynomially related to the original completeness-soundness gap $1-s$, and it also decides $L$.

It remains to argue that the amplified protocol $V_{ZK,\bot}^m$ still has the perfect zero knowledge property. In general, this is a delicate issue, since it is known that parallel repetition does not preserve zero knowledge in a black box manner~\cite{GoldreichK96,BellareIN97,Pass06}. 

In our case, however, since the referee is constrained to interacting with each prover only once, we can simulate the interaction in the amplified protocol $V_{ZK,\bot}^m$ by essentially running many copies of the simulator $\Sim_{\adv{R}}$ described in Figure~\ref{fig:sim} in parallel. Notice that the honest strategy for $V_{ZK,\bot}^m$ consists of sharing $m$ copies of the history state, and performing independent measurements on each of these copies. It is not hard to see that the interaction in $V_{ZK,\bot}^m$ can be simulated efficiently.

The number of provers involved in the protocol executed by $V_{ZK,\bot}^m$ is $k + 4$, and the protocol is $1$-round, which implies that
\[
	L \in \PZKMIP^*_{1,s''}[k+4,1].
\]
where $s'' \leq 1 - (1 -s)^\gamma$.
This concludes the proof of Theorem~\ref{thm:main}.
\section{Simulatable codes}
\label{sec:simulatable}
\label{sec:simulability}

In this section we show the existence of simulatable codes.  We start by
introducing stabilizer codes and some notation in \Cref{sec:stabilizer}. Then,
we analyse low-weight measurements on codewords of a stabilizer QECC in
\Cref{sec:partial-trace}. In
\Cref{sec:simulation-transversal,sec:simulation-toffoli} we show how to simulate
low-weight measurements on the encoding of transversal and non-transversal
gates, respectively. Finally, in~\Cref{sec:explicit-simulatable} we show that
the concatenated Steane code is a simulatable code.

\subsection{Stabilizer codes}
\label{sec:stabilizer}
We present some preliminary background on stabilizer codes, an important class of QECCs.
For an in-depth reference on stabilizer codes, we recommend consulting~\cite{Gottesman97}.

Let $\mcP_n$ be the $n$-qubit
Pauli group, so $P_n$ is the set of $n$-qubit unitaries $W_1 \otimes \cdots
\otimes W_n$, where $W_i \in \{\pm I, \pm i I, \pm X, \pm i X, \pm Y, \pm i
Y, \pm Z, \pm i Z\}$ for all $i=1,\ldots,n$. The \emph{weight} of an element
$W_1 \otimes \cdots \otimes W_n \in P_n$ is $|\{1 \leq i \leq n : W_i \not\in \{\pm
I, \pm i I \}|$.  

An $[[n,k]]$ stabilizer code is specified by an abelian subgroup
$\mcS \subseteq P_n$ such that $-I \not\in \mcS$, and any minimal generating
set of $\mcS$ has size $n-k$.  Usually we fix a minimal generating set
$g_1,\ldots,g_{n-k}$ of $\mcS$, and refer to these elements as the stabilizers
of the code. The codespace of an $[[n,k]]$ stabilizer code $\mcS$ is the subspace
of vectors in $\C^{2^n}$ fixed by $\mcS$. 
In other words, $\ket{\psi}$ is in the code if and only if $g \ket{\psi} = \ket{\psi}$
for all $g \in \cS$. 
This space always has dimension
$2^k$, and hence an $[[n,k]]$ stabilizer code can encodes $k$-qubit states in
$n$-qubit states. To  fix one of all the possible encodings of $k$-qubit states, we can find
$\overline{Z}_1,\ldots,\overline{Z}_k \in \mcP_n$ such that $\mcS' = \langle
g_1,\ldots,g_{n-k}, \overline{Z}_1,\ldots,\overline{Z}_k\rangle$ is abelian and
is minimally generated by
$g_1,\ldots,g_k,\overline{Z}_1,\ldots,\overline{Z}_k$. The encoding then sends
the computational basis state $\ket{x_1 \cdots x_k}$ to the unique state in the
codespace fixed by $g_1,\ldots,g_k,(-1)^{x_1} \overline{Z}_1, \ldots,
(-1)^{x_k} \overline{Z}_k$.

More generally, an $[[n,k]]$ stabilizer code can be used to encode $mk$-qubit
states in $mn$-qubit states. This can be expressed in the stabilizer formalism
by taking the product of the stabilizer code with itself $m$ times.
Specifically Let $\Delta^i_{n,mn} : \mcP_n \arr \mcP_{mn}$ be the
inclusion\footnote{An inclusion map $f : A
\mapsto B$ consists in treating an element $x \in A$, as an element of $B$.} induced by
having $\mcP_n$ act on qubits $(i-1)n+1,(i-1)n+2,\ldots,in$ of an $mn$-qubit
register. For example, we have $\Delta^2_{2,6}(X \otimes Z) = I \otimes I \otimes X
\otimes Z \otimes I \otimes I$. When $n$ and $m$ are clear, we write $\Delta^i$ for $\Delta^i_{n,mn}$. Given a stabilizer code $\mcS$ with stabilizers $g_1,\ldots,g_{n-k}$ in
$\mcP_n$, let
\begin{equation*}
    \mcS^{\otimes m} := \langle \Delta^i(g_j) \text{ where } 1 \leq i \leq m, 1 \leq j \leq n-k \rangle.
\end{equation*}
This defines an $[[mn,mk]]$ stabilizer code with minimal generating set
$\{\Delta^i(g_j)\}$.  To encode $mk$-qubit states in this code, we can take
elements $\overline{Z}_1,\ldots,\overline{Z}_k \in \mcP_n$ specifying an
encoding of $\mcS$ as above. Then elements $\Delta^i(\overline{Z}_j)$, $1 \leq i
\leq m$, $1 \leq j \leq k$ specify an encoding for the code $\mcS^{\otimes m}$.

Every pair of elements $g, h \in \mcP_n$ either commute or anticommute. Since
stabilizer codes do not contain $-I$ by definition, the normalizer $N(\mcS)$ of a
stabilizer code $\mcS$ in $\mcP_n$ is the set of all elements of $\mcP_n$ which
commute with every element in $\mcS$. The \emph{distance} of a stabilizer code is
the smallest integer $d$ such that $N(\mcS) \setminus \mcS$ contains an element of
weight $d$. An $[[n,k,d]]$ stabilizer code is an $[[n,k]]$ stabilizer code of
distance $\geq d$.

\subsection{Computing partial trace of codewords}
\label{sec:partial-trace}

Suppose we want to the compute the partial trace $\tr_{\overline{Q}}(\rho)$ for some
$n$-qubit state $\rho$ and subset of qubits $Q$. Because $\mcP_{|Q|}$ contains an
orthogonal basis (in the Hilbert-Schmidt inner product) for $2^{|Q|} \times
2^{|Q|}$ matrices and is closed under the adjoint operation, it is sufficient
to compute the inner products $\tr(\tr_{\overline{Q}}(\rho) w)$ for all elements
$w \in \mcP_{|Q|}$. Extending the notation from \Cref{sec:stabilizer}, let $\Delta^Q_n : \mcP_{|Q|} \arr \mcP_n$ be the inclusion
induced by having elements of $\mcP_{|Q|}$ act on qubits $Q$ (so for instance,
$\Delta_{n,mn}^i = \Delta^Q_{mn}$ where $Q = \{(i-1)n+1,\ldots,in\} \subseteq
\{1,\ldots,mn\}$). Then 
\begin{equation*}
    \tr\left(\tr_{\overline{Q}}(\rho) w\right) = \tr ( \rho \Delta_n^Q(w)),
\end{equation*}
so to compute $\tr_{\overline{Q}}(\rho)$, it is sufficient to be able to compute
$\tr( \rho \Delta_n^Q(w))$ for all $w \in \mcP_{|Q|}$. We record this fact
in the following lemma:
\begin{lemma}\label{L:paulitrace}
    The partial traces $\tr_{\overline{Q}}(\rho)$ of an $n$-qubit state $\rho$ can
    be computed from the traces $\tr( \rho \Delta_n^Q(w))$, $w \in \mcP_{|Q|}$,
    in time $\exp(O(|Q|))$. 
\end{lemma}

For a stabilizer code $\mcS$, we can easily compute $\tr( \Enc(\rho) w)$,
without knowing $\rho$, as long as $w$ is not in $N(\mcS) \setminus \mcS$.
\begin{lemma}\label{L:stabilizerinnerproduct}
    Let $\Enc(\rho)$ be an encoding of a $k$-qubit state $\rho$  in an $[[n,k]]$
    stabilizer code $\mcS$, and suppose $w \not\in N(\mcS) \setminus \mcS$. Then 
    \begin{equation*}
        \tr(\Enc(\rho) w) = \begin{cases} 1 & w \in \mcS \\ 0 & w \not\in \mcS \end{cases}.
    \end{equation*}
\end{lemma}
\begin{proof}
    If $w \in \mcS$, then $\Enc(\rho) w = \Enc(\rho)$ by definition, since $w$
    fixes the codespace of $\mcS$. So $\tr(\Enc(\rho) w) = \tr(\Enc(\rho)) = 1$.

    Suppose $w \not\in N(\mcS)$. Let $g_1,\ldots,g_{n-k}$ be a minimal generating
    set for $\mcS$. By a standard argument, we can assume that $g_2,\ldots,g_{n-k}$
    commute with $w$, and $g_1$ anticommutes. Let
    \begin{equation*}
        P = \left(\frac{I + g_1}{2}\right) \left(\frac{I + g_2}{2}\right)
                \cdots \left(\frac{I + g_{n-k}}{2}\right),
    \end{equation*}
    the projection onto the codespace of $\Enc(\rho)$. Then
    $P w = w P'$, where
    \begin{equation*}
        P' = \left(\frac{I - g_1}{2}\right) \left(\frac{I + g_2}{2}\right)
                \cdots \left(\frac{I + g_{n-k}}{2}\right),
    \end{equation*}
    an orthogonal projection to $P$. So
        \[
        \tr(\Enc(\rho) w) = \tr(P \Enc(\rho) P w) = \tr(P \Enc(\rho) w P')
         = \tr(P' P \Enc(\rho) w) = 0.
      \qedhere \]
\end{proof}
In particular, if $\mcS$ is an $[[n,k,d]]$ stabilizer code, and $|Q| < d$, then
$\Delta_n^Q(w)$ will have weight $< d$ for all $w \in \mcP_{|Q|}$, and hence
$\tr(\Enc(\rho) \Delta_n^Q(w))$ will be equal to $1$ or $0$ for all $w \in
\mcP_{Q}$, depending on whether $\Delta_n^Q(w) \in \mcS$. If $\mcS$ is
non-degenerate, meaning that every element of $\mcS$ has weight $\geq d$, then
$\tr(\Enc(\rho) \Delta_n^Q(w)) = 0$ unless $\Delta_n^Q(w)$ is the identity
matrix, so $\tr_{\overline{Q}}(\Enc(\rho))$ will be maximally mixed.
If the code is degenerate, 
$\tr_{\overline{Q}}(\Enc(\rho))$ will not always be maximally mixed; instead,
it is maximally mixed over the invariant subspace of the degenerate stabilizers. 

If $\mcS$ is an $[[n,k,d]]$ stabilizer code, then the product code $\mcS^{\otimes m}$
only has distance $d$ (and hence is an $[[mn,mk,d]]$ code). However, we can say
more about when an element of $\mcP_{mn}$ is in $N(\mcS^{\otimes m})$. 
\begin{lemma}\label{L:productnormalizer}
    Let $\mcS$ be an $[[n,k,d]]$ code, and suppose $w_1,\ldots,w_m$ are elements
    of $\mcP_n$. Then $w = w_1 \otimes \cdots \otimes w_m \in \mcP_{mn}$ belongs
    to $N(\mcS^{\otimes m})$ if and only if $w_i \in N(\mcS)$ for all $1 \leq i \leq m$.
    In particular, if $w_i$ has weight $< d$ for every $1 \leq i \leq m$, then
    $w \in N(\mcS^{\otimes m})$ if and only if $w \in \mcS^{\otimes m}$.
\end{lemma}
\begin{proof}
    Suppose $w \in N(\mcS^{\otimes m})$. Then $w \Delta^i(g) w^{-1} = \Delta^i(g)$
    for all $g \in \mcS$. But $w = \Delta^1(w_1) \cdots \Delta^m(w_m)$, so 
    $w \Delta^i(g) w^{-1} = \Delta^i(w_i g w_i^{-1})$. Since $\Delta^i$ is an
    inclusion, $w_i \in N(\mcS)$. 

    If $w_i$ has weight $< d$ for all $1 \leq i \leq m$, and $w \in
    N(\mcS^{\otimes m})$, then we must have $w_i \in \mcS$ for all $1 \leq i
    \leq m$, and hence $w \in \mcS^{\otimes m}$. 
\end{proof}

\subsection{Simulatable encoding of transversal Clifford gates}
\label{sec:simulation-transversal}
We now consider what happens if we add operations on encoded states into the
picture. For simplicity of description, we restrict to $[[n,1,d]]$ stabilizer
codes. Recall that the $n$-qubit Clifford group $\mcC_n$ is the normalizer of
$\mcP_n$ in the group of unitaries. 
\begin{lemma}\label{L:bitwisecliffords}
    Let $\mcS$ be an $[[n,1,d]]$ stabilizer code, let $\rho$ be a $k$-qubit
    state, and let $O_1,\ldots,O_\ell \in \mcC_{nk}$ such that $O_i$ acts
    on a subset $Q_i$ of the physical qubits of $\mcS^{\otimes k}$, where
    $Q_i \cap Q_j = \emptyset$ for all $1 \leq i \neq j \leq \ell$, and
    $Q_i$ contains at most one physical qubit from each logical qubit of
    $\mcS^{\otimes k}$ for all $1 \leq i \leq \ell$. 

    If $Q$ is a subset of the physical qubits of $\mcS^{\otimes k}$ with
    $|Q| < d$, then we can compute 
    \begin{equation*}
        \tr\left( O_\ell \cdots O_1 \Enc(\rho) (O_\ell \cdots O_1)^\dagger \Delta_{nk}^Q(w) \right)
    \end{equation*}
    for all $w \in \mcP_{|Q|}$ without knowledge of $\rho$. Furthermore, if
    $n$ and $d$ are constant, then this computation can be done in polynomial
    time in $k$, $\ell$, and the maximum amount of time needed to compute
    $O_i^\dagger g O_i \in \mcP_{nk}$ for any $1 \leq i \leq \ell$ and $g \in \mcP_{nk}$.
\end{lemma}
\begin{proof}
    Since $O_\ell \cdots O_1 \in \mcC_{nk}$, 
    \begin{equation*}
        (O_{\ell} \cdots O_1)^\dagger \Delta_{nk}^Q(w) = w' (O_{\ell} \cdots O_1)^\dagger,
    \end{equation*}
    where $w' \in \mcP_{nk}$ can be computed in time polynomial in $\ell$ and
    the time needed to compute $O_i^\dagger g O_i$ for any $g \in \mcP_{nk}$ and $1
    \leq i \leq \ell$.

    Let $R = Q \setminus \bigcup_{i=1}^{\ell} Q_i$. 
    Write $w' = W_1 \otimes \cdots \otimes W_{nk}$
    where $W_j \in \mcP_1$ for all $1 \leq j \leq nk$, and let $w_a =
    W_{(a-1)n+1} \otimes \cdots \otimes W_{an}$, so $w_a$ contains the
    operators corresponding to the $a$th logical qubit, $1 \leq a \leq k$. 
    Since the operators $O_i$ act on disjoint sets of physical qubits, if
    $W_j \not\in \{\pm I, \pm i I\}$, then either $j \in R$, or $j \in Q_i$
    for some $i$ with $Q_i \cap Q \neq \emptyset$. Because
    \begin{equation*}
        |\{(a-1)n + 1,\ldots,an\} \cap Q_i| \leq 1 \text{ for all } 1 \leq i \leq \ell,
    \end{equation*}
    we see that the weight of $w_a$ is at most
    \begin{equation*}
        |\{1 \leq i \leq \ell : Q_i \cap Q \neq \emptyset\}| + |R| \leq |Q| < d.
    \end{equation*}
    By Lemma \ref{L:productnormalizer}, $w' \in \mcS^{\otimes k}$ if and only
    if $w' \in N(\mcS^{\otimes k})$. Also, $w' \in N(\mcS^{\otimes k})$ if and
    only if $w_a \in N(\mcS)$ for all $1 \leq a \leq \ell$. Since $\mcS$ is
    fixed, we can check whether $w_a \in N(\mcS)$ in constant time, and hence
    we can determine whether $w' \in \mcS^{\otimes k}$. 

    Finally, we have
    \begin{equation*}
        \tr\left( O_\ell \cdots O_1 \Enc(\rho) (O_\ell \cdots O_1)^\dagger \Delta_{nk}^Q(w) \right)
            = \tr\left( O_\ell \cdots O_1 \Enc(\rho)w' (O_\ell \cdots O_1)^\dagger\right)
            = \tr(\Enc(\rho) w').
    \end{equation*}
    Since $w' \not\in N(\mcS^{\otimes k}) \setminus \mcS^{\otimes k}$, 
    \begin{equation*}
        \tr(\Enc(\rho)w') = \begin{cases} 1 & w' \in \mcS^{\otimes m} \\
                                          0 & w' \not\in \mcS^{\otimes m}.
            \end{cases}
    \end{equation*}
    by Lemma \ref{L:stabilizerinnerproduct}. 
\end{proof}

With the previous lemma, we can show how to simulate the transversal encoding of
Clifford gates.
 
\begin{proposition}\label{P:transversalsim}
  If the $[[n,1,d]]$ stabilizer code accepts a transversal encoding of a
  $k$-qubit Clifford gate $G$, then such encoding is $s$ simulatable for all $s < d$.
\end{proposition}
\begin{proof}
    Let $\rho$ be an $n$-qubit state,
    $\underline{a} = (a_1,...,a_k)$ be a $k$-tuple
    of disjoint integers between $1$ and $n$ and
    $O_1(\underline{a},\ldots,O_{\ell}(\underline{a})$ be the encoding of
    $G(\underline{a})$ for $\ell = n$, and let $S$ be a subset of
    $\{1,\ldots,\ell n\}$ with $|S| \leq s$. 

    If $1 \leq t \leq \ell$, then the set of gates
    $O_1(\underline{a}),\ldots,O_t(\underline{a})$ satisfy the conditions of Lemma
    \ref{L:bitwisecliffords}, and furthermore, since $O_i(\underline{a})$ acts
    on at most $k$ physical qubits, $O_i(\underline{a})^\dagger w O_i(\underline{a})$
    can be computed in polynomial time in $n$ for all $w \in \mcP_{\ell n}$. 

    The proposition
    follows from Lemmas \ref{L:paulitrace} and  \ref{L:bitwisecliffords}.
\end{proof}

\subsection{Simulatable encoding of non-transversal gates}
\label{sec:simulation-toffoli}

It is well known that there is no QECC where all logical operations from a universal set of gates can be performed transversally. In order to circumvent this barrier, we can use other tools from 
fault-tolerant quantum computation, namely magic states.

The general procedure for applying a $k$-qubit gate $G$ using a magic state for it is
depicted in \Cref{fig:magic}. The input to this procedure is some $k$-qubit
state $\rho$,
on which we want to apply $G$, and the magic state $\ket{\Magic_G}$.
In the first phase, a unitary $V_0$ is
applied to both registers. In the second phase, some of the qubits are
measured. Finally,  classical controled unitaries $V_{i}$ are applied.
We assume for simplicity that each $V_i$ can be applied transversally
\footnote{More generally, we could assume that $V_i$ can be decomposed on gates
that can be applied transversally and then the encoding of $V_i$ consists of
the sequence of encoding for each of these gates.}.
The output of this 
procedure is then $G\rho G^\dagger$.

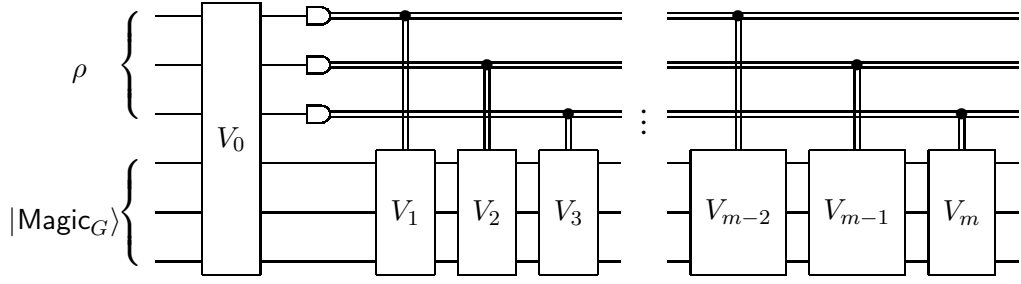
\begin{figure}[H]
\centering
$
\Qcircuit @C=0.8em @R=0.8em {
  &  & \qw & \multigate{5}{V_0} &  \qw & \measureD{}  & \cw  
    & \control \cw \cwx[3] & \cw & \cw 
  &\cw  & & 
    &   \control \cw \cwx[3] & \cw & \cw & \cw
    \\
  &  & \qw & \ghost{V_0} &  \qw & \measureD{}  & \cw  
    & \cw & \control \cw \cwx[2] & \cw
  & \cw & & 
    & \cw & \control \cw \cwx[2] & \cw & \cw \\
  &  & \qw & \ghost{V_0} &  \qw & \measureD{}  & \cw  
    & \cw & \cw & \control \cw \cwx[1]    
  & \cw &  \vdots  & 
    & \cw & \cw & \control \cw \cwx[1] & \cw   
\inputgroupv{1}{3}{.2em}{2.0em}{\rho \;} \\
 &  & \qw & \ghost{V_0}        &  \qw & \qw         & \qw
  & \multigate{2}{V_1} \qw 
  & \multigate{2}{V_2} \qw 
  & \multigate{2}{V_3} \qw 
  & \qw & & 
  & \multigate{2}{V_{m-2}} \qw 
  & \multigate{2}{V_{m-1}} \qw 
  & \multigate{2}{V_{m}} \qw & \qw
  \\
 &  & \qw & \ghost{V_0}        &  \qw & \qw         & \qw 
  & \ghost{V_1} \qw
  & \ghost{V_2} \qw
  & \ghost{V_3} \qw
  & \qw & &
  & \ghost{V_{m-2}} \qw
  & \ghost{V_{m-1}} \qw
  & \ghost{V_m} \qw & \qw
  \\
 &  & \qw & \ghost{V_0}        &  \qw & \qw         & \qw 
  & \ghost{V_1} \qw
  & \ghost{V_2} \qw
  & \ghost{V_3} \qw
  & \qw & & 
  & \ghost{V_{m-2}} \qw
  & \ghost{V_{m-1}} \qw
  & \ghost{V_m} \qw & \qw
\inputgroupv{4}{6}{.2em}{2.0em}{\ket{\Magic_G} \;\quad } \\
}
$
\caption{Teleportation gadget}
\label{fig:magic}
\end{figure}

There are two issues when one try to prove that the encoding of such 
gadget is simulatable. First,
it would be necessary to {\em decode} the measured qubit to perform the
controlled unitary, and this could make the encoding non-simulatable. 
In order to solve this problem, we present in \Cref{S:non-transversal} an extra
property that we require from our QECC and how it allows us to 
 define the non-transversal encoding of $G$.
Secondly, measurements are not allowed in 
\Cref{D:simulatable}. In this case, we show how to perform a coherent encoding of gate $G$ in \Cref{S:coherent}.

With these two pieces in hand, we are able to show the simulation of the
coherent non-transversal encoding of $G$ in \Cref{S:simulation-non-transversal}.

\subsubsection{Simulatable encoding of classically-controlled transversal gates}\label{S:non-transversal}
As mentioned before, one of the possible issues when simulating the encoding of \Cref{fig:magic}
is that the control qubit must be decoded during the gadget.  
Here, we show how to workaround this issue by
assuming that the QECC has the following property.

\begin{definition}
  For $b \in \{0,1\}$, let $E_b$ be the set of strings in $\Enc(\ket{b})$.
  A QECC $\mcS$ is order-consistent with a set of unitaries $\mathcal{V}$ 
  if for every $V \in \mathcal{V}$ and every $x \in E_b$, it follows that
  $V^{|x|} = V^b$.
\end{definition}

We require then our QECC to be order-consistent with the set of unitaries
$\mathcal{V} = \{V_0,...,V_m\}$ in \Cref{fig:magic}.
In order to explain how the non-transversal encoding works, let us 
first assume that starting from $\Enc(\ket{b}) \otimes \ket{\psi}$, we want to
apply $V^b$ on $\ket{\psi}$ without decoding $\ket{b}$.
This operation can be implemented by applying the gates
\begin{equation*}
  O_{1},O_{2},\ldots, O_{n}
\end{equation*}
where $O_{i}$ is a $\Lambda(V)$-gate with the $i$th qubit of 
$\Enc(\ket{b})$ as the control qubit and
$\ket{\psi}$ as target. 
We depict such encoding in 
\Cref{F:nontransversal}.

Notice that
this sequence of gates performs the correct logical operation, since
after applying 
$O_{1},O_{2},\ldots, O_{n}$, the resulting state is
\begin{align} \label{eq:order-consistent}
  \sum_{x \in E_b} \alpha_x \ket{x}V^{|x|}\ket{\psi},
\end{align}
where $E_b$ is the set of string in the support of $\Enc(\ket{b})$. Since
we assume that the QECC is order-consistent with $\mathcal{V}$, we have that
\Cref{eq:order-consistent} is equal to
\begin{align*}
  \Enc(\ket{b})V^{b}\ket{\psi}.
\end{align*}

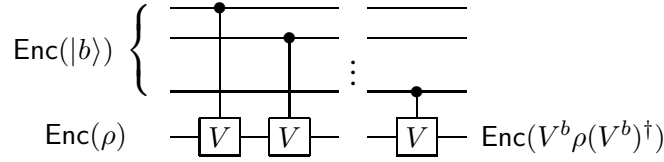
\begin{figure}[t]
  \centering
  $
        \Qcircuit @C=1em @R=.8em {
          &    & \ctrl{4}          & \qw               & \qw & & \qw      & \qw \\
          &  & \qw & \ctrl{3} & \qw &   & \qw      & \qw\\
          &  &     &          &  \quad\vdots & &          & \\
          &  & \qw & \qw      & \qw &   & \ctrl{1} & \qw & \\
        \lstick{\Enc(\rho)}   &  &\gate{V} &\gate{V} & \qw &  &\gate{V}   \qw    & 
  \rstick{\Enc(V^b \rho (V^b)^\dagger)} \qw\\
  {\inputgroupv{1}{4}{.2em}{1.5em}{\Enc(\ket{b}) \quad \quad }}
        }
        $
    \caption{Non-transversal encoding the controlled gates}
    \label{F:nontransversal}
\end{figure}

We can extend this idea and have an encoding of $\ket{\psi}$, and then we can
just decompose $O_i$
\begin{equation*}
  O_i = O_{i,1},O_{i,2},\ldots, O_{i,n}
\end{equation*}
where $O_{i,j}$  is a $\Lambda(V)$ gate with the $i$th physical qubit of the encoding of the
logical control bit, and the target are the $j$th physical qubits of the
encoding of the logical target state.

With this tool, we are able to prove the encoding of $\Lambda(V)$ is simulatable if the
control is a classical known bit.

\begin{lemma}\label{L:controlled-transversal}
  Let $\mcS$ be a $[[n,1,d]]$ stabilizer code that is order-consistent with
  some $k$-qubit unitary $V$ which has a transversal encoding in $\mcS$.
  Then the non-transversal encoding of 
  \[\Lambda(V)   \left(\kb{x} \otimes \rho \right)\Lambda(V)^\dagger\]
  for some fixed $x \in \{0,1\}$ is $s$-simulatable for any $s < d-k$.
\end{lemma}
\begin{proof}
    The non-transversal encoding of $\Lambda(V)$ consists of $O_{1,1}, O_{2,1},\ldots,O_{n,1},
    O_{1,2},\ldots,O_{n,n}$, where $O_{i,j}$ is the $\Lambda(V)$ gate with control on
    the $i$th physical qubit of the encoding of the first qubit and $j$th
    physical qubits of the logical target qubits.

    Let $S$ be the subset of qubits for which we want to compute 
    \begin{align*}
       \tr_{\overline{S}}(O_{(t_1,t_2)}  O_{(t_1-1,t_2)} \cdots O_{2,1}
      O_{1,1} \Enc(\kb{x} \otimes \rho)
            (O_{(t_1,t_2)} O_{(t_1-1,t_2)} \cdots O_{2,1} O_{1,1})^\dagger),
    \end{align*}
    with $|S| \leq s$,
    \begin{equation*}
        \Enc(\ket{x}) = \sum_{y \in \{0,1\}^{(k+1) n}} \alpha_y \ket{y},
    \end{equation*}
    and
    \begin{equation*}
        \Enc(\kb{x} \otimes \rho) = \sum_{y,z\in \{0,1\}^{n}}
           \alpha_y \alpha_z^* \ket{y}\bra{z} \otimes \Enc(\rho).
    \end{equation*}

    We can fix the computational basis states $\ket{y}$ and
    $\ket{z}$ in the support of $\Enc(\ket{x})$ and the result will follow by
    linearity.
    Let $V(j)$ denote the gate $V$ applied to the $j$th physical
    qubits of the encodings of each qubit of $\rho$. We have that
    \begin{align*}
      &   O_{(t_1,t_2)}  O_{(t_1-1,t_2)} \cdots O_{2,1} O_{1,1} \ket{y}\bra{z}
        \otimes \Enc(\rho)
            (O_{(t_1,t_2)} O_{(t_1-1,t_2)} \cdots O_{2,1} O_{1,1})^\dagger \\
            & = \ket{y}\bra{z} \otimes  V(t_2)^{b} V(t_2-1)^{x} V(t_2-2)^{x} \cdots
            V(1)^{x} \Enc(\rho)
      \left(V(t_2)^{c} V(t_2-1)^{x} V(t_2-2)^{x} \cdots V(1)^{x}\right)^\dagger,
    \end{align*}
    where we use the fact that $V(j)^{|y|} = V(j)^x$ for every $y$ in
    $\Enc(\ket{x})$ and we define $b = y_{1} + \ldots + y_{{t_1}}$ and $c = z_{1} +\ldots + z_{{t_1}}$. Let
   $S_1 = S \cap
   \{1,\ldots,n\}$, and let $S_2 = S \setminus S_1$ (relabelled to
    be a subset of $\{1,\ldots,kn\}$). 
    Then
    \begin{align*}
      &   \tr_{\overline{S}}(O_{(t_1,t_2)}  O_{(t_1-1,t_2)} \cdots O_{2,1}
      O_{1,1} \ket{y}\bra{z} \otimes \Enc(\rho)
            (O_{(t_1,t_2)} O_{(t_1-1,t_2)} \cdots O_{2,1} O_{1,1})^\dagger)\\
        & = \tr_{\overline{S_1}}(\ket{y}\bra{z}) \otimes \tr_{\overline{S_2}}
        \left(V(t_2)^{b} V(t_2-1)^{x} \cdots V(1)^{x} \Enc(\rho)
      \left(V(t_2)^{c} V(t_2-1)^{x} V(t_2-2)^{x} \cdots
      V(1)^{x}\right)^\dagger\right).
    \end{align*}
    Because $x$, $y$, $z$ and the code are fixed, we can compute
    $\tr_{\overline{S_1}}(\ket{y}\bra{z})$ explicitly, so the only remaining
    step is to compute
    \begin{align*}
        \tr_{\overline{S_2}}
        \left(V(t_2)^{b} V(t_2-1)^{x} \cdots V(1)^{x} \Enc(\rho)
      \left(V(t_2)^{c} V(t_2-1)^{x} V(t_2-2)^{x} \cdots
      V(1)^{x}\right)^\dagger\right).
    \end{align*}
    By Lemma \ref{L:paulitrace}, it suffices to compute
    \begin{align}\label{Eq:codetarget}
        \tr_{\overline{S_2}}
        \left(V(t_2)^{b} V(t_2-1)^{x} \cdots V(1)^{x} \Enc(\rho)
      \left(V(t_2)^{c} V(t_2-1)^{x} V(t_2-2)^{x} \cdots
      V(1)^{x}\right)^\dagger\Delta^{S_2}_{kn} (w) \right).
    \end{align}
    for all $w \in \mcP_{|S_2|}$. Since $x$ is fixed, and can compute $b$
    and $c$ from $y$ and $z$. %
    Hence if $b=c$, we can compute \eqref{Eq:codetarget} by Lemma \ref{L:bitwisecliffords},
    without knowledge of $\rho$, since $|S_2| \leq |S| \leq d$.
    If $b \neq c$, then we can rewrite $V(t_2)^{b}(V(t_2)^c)^\dagger$ as a linear combination of weight
    $k$ elements of $\mcP_{kn}$, meaning that we can compute \eqref{Eq:codetarget} as
    long as we can compute 
    \begin{equation*}
        \tr \left(V(t_2-1)^{x} \cdots V(1)^{x} \Enc(\rho)
                V(1)^{x} \cdots V(t_2-1)^{x}  w' \right)
    \end{equation*}
    where $w'$ has weight at most $|S_2|+k \leq s+k < d$. So once again,
    we can compute this trace using Lemma \ref{L:bitwisecliffords}. 
\end{proof}

\subsubsection{Coherent encoding}\label{S:coherent}

We now present how to remove the measurements from \Cref{fig:magic}. First,
we notice that measuring one qubit is equivalent of copying ($\Lambda(X)$) it
into a fresh $\ket{0}$ ancilla, and then tracing out this ancilla. 

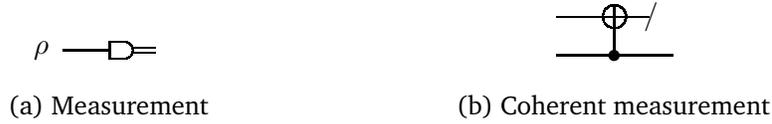
\begin{figure}[H]
\centering
\begin{subfigure}[b]{0.4\textwidth}
  \[
\Qcircuit @C=0.8em @R=0.8em {
  \lstick{\rho}  & \qw & \measureD{}  & \cw  
}
\]
\caption{Measurement}
\end{subfigure}
\begin{subfigure}[b]{0.4\textwidth} \[
\Qcircuit @C=0.8em @R=0.8em {
  & \qw & \targ \qw &  \qw & \lstick{/} \\  
  & \qw & \ctrl{-1} \qw & \qw & \qw 
}
\]
\caption{Coherent measurement}
\end{subfigure}
\label{fig:coherent-measurement}
  \caption{Replacing measurement by coherent measurement}
\end{figure}

Even though this is sufficient to remove the measurements of \Cref{fig:magic},
this will not be sufficient for the simulation, since we cannot guarantee that
the ancilla will be traced out, and it could be hard to keep track of
the entangled values. However, if we are just interested in a
subset $S$ of qubits with $|S| \leq s$ for a fixed $s$, we can repeat the previous
procedure with $s$ ancillas, and in this case, one of such $s+1$ register must
be traced out when simulating on $S$. We depict this in 
\Cref{fig:coherent_measurement}.

\begin{figure}[h]
\centering
$
\Qcircuit @C=0.8em @R=0.8em {
  & \qw   &  \qw & \targ & \qw   & \qw & \qw &  & \qw &  \qw &  \qw &  \\
  & \qw   &  \qw & \qw   & \targ & \qw & \qw &  & \qw & \qw &  \qw &    \\
 & & &  & & &  & \vdots\quad & & & & & & &&&&\\
  & \qw   &   \qw & \qw & \qw & \qw & \qw &  & \qw &\targ & \qw   \\
  & \qw &   \qw & \ctrl{-4} & \ctrl{-3}& \qw & \qw &  & \qw &\ctrl{-1} \qw      & \qw &
  {\inputgroupv{1}{4}{.2em}{2.0em}{\ket{0}^s}}\\
}
$
\caption{Coherent measurement with more ancilla}
\label{fig:coherent_measurement}
\end{figure}
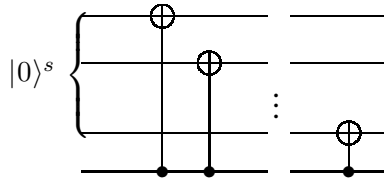

\subsubsection{Simulation of coherent non-transversal encoding}\label{S:simulation-non-transversal}

Starting with \Cref{fig:magic} and making it non-transversal, as defined in
\Cref{S:non-transversal}, and coherent, as defined in \Cref{S:coherent}, we get 
the circuit defined in \Cref{fig:coherent_toffoli}.

\begin{figure}[t]
\centering
$
\Qcircuit @C=0.8em @R=0.8em {
&  & \qw   &  \qw & \targ & \qw   & \qw & \qw &  & \qw & \qw & \qw & \qw & \qw & \qw & \qw & \qw & \qw & & \qw & \qw    \\
&  & \qw   &  \qw & \qw   & \targ & \qw & \qw &  & \qw & \qw & \qw & \qw & \qw & \qw & \qw & \qw & \qw & & \qw & \qw     \\
&  & \qw   &  \qw& \qw    & \qw & \targ & \qw &  & \qw & \qw & \qw & \qw & \qw & \qw & \qw & \qw & \qw & & \qw & \qw    \\
& & &  & & & & & &  & & & & & & & & &&&&&\\
&  & \qw   &   \qw & \qw & \qw & \qw & \qw &  & \qw &\targ & \qw & \qw & \qw &\qw \qw & \qw & \qw & \qw & & \qw &\qw   \\
&  & \qw   &   \qw & \qw & \qw & \qw & \qw &  & \qw & \qw &\targ & \qw & \qw & \qw & \qw &\qw     & \qw & & \qw & \qw     \\
&  & \qw   &   \qw & \qw & \qw & \qw & \qw &  & \qw & \qw & \qw &\targ & \qw & \qw & \qw &\qw     & \qw & & \qw & \qw     \\
  & & & &  & & &  & \vdots\quad& & & & & & & &&& \vdots\quad & & &\\
  & \qw & \multigate{6}{V_0} &   \qw & \ctrl{-8} & \qw & \qw & \qw &  & \qw
  &\ctrl{-4}       & \qw & \qw & \qw  & \ctrl{4} & \qw & \qw &   \qw &  &\qw &  \qw & \\
  & \qw & \ghost{V_0}        &   \qw & \qw & \ctrl{-8} & \qw & \qw &  &  \qw  &
  \qw &\ctrl{-4} & \qw & \qw & \qw  & \ctrl{3} &\qw & \qw &  & \qw &\qw \\
& \qw & \ghost{V_0}&  \qw & \qw & \qw & \ctrl{-8} & \qw &  & \qw  &\qw & \qw  &
  \ctrl{-4} & \qw & \qw & \qw & \ctrl{2} & \qw &  & \ctrl{2} & \qw  \\
& &  &  & & & & & & & & & &&&&&\\
& & \ghost{V_0}&   \qw & \qw & \qw  & \qw & \qw &  & \qw & \qw & \qw & \qw & \qw
  & \multigate{2}{V_1} & \multigate{2}{V_2} & \multigate{2}{V_3} &
  \qw &  &  \multigate{2}{V_m} & \qw \\
&  & \ghost{V_0}&  \qw &\qw & \qw & \qw  & \qw &  & \qw & \qw & \qw & \qw & \qw
  & \ghost{V_1} & \ghost{V_2} & \ghost{V_3} & \qw & &   \ghost{V_m} & \qw \\
&  & \ghost{V_0} & \qw & \qw & \qw & \qw &\qw &  & \qw & \qw & \qw & \qw &\qw &
  \ghost{V_1} & \ghost{V_2} & \ghost{V_3} & \qw &   &  \ghost{V_m}& \qw
\inputgroupv{13}{15}{.2em}{2.0em}{\ket{\Magic_G} \quad } \inputgroupv{1}{7}{.2em}{4.5em}{\ket{0}^{\otimes ks} \quad} 
}
$
\caption{The coherent non-transversal encoding of $G$}
\label{fig:coherent_toffoli}
\end{figure}
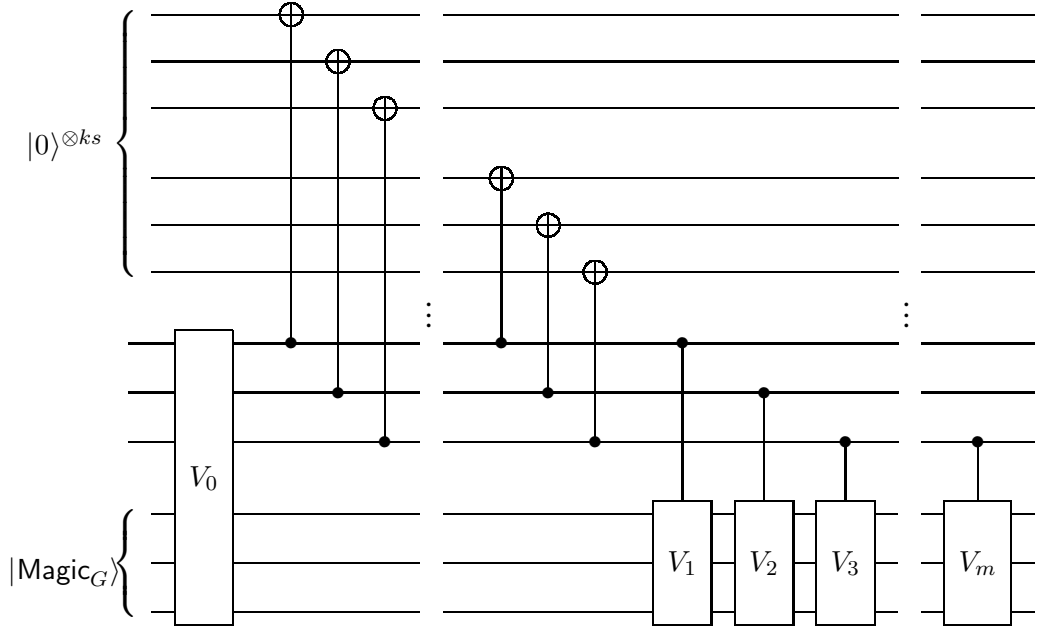

  For simplicity, we denote $U_0, \ldots,U_{m+ks}$ the gates of
  \Cref{fig:coherent_toffoli} such that 
  $U_0 = V_0$, $U_i$ denotes the $\Lambda(X)$ where the control is the $(i
  \pmod{k})$th input qubit and the target is the $i$th ancilla qubit for $1 \leq
  i \leq ks$; and $U_i$ denotes the controlled-$V_{i-ks}$ for $i > ks$.

By the correctness of the $G$-gadget of \Cref{fig:magic}, the following lemma
holds.

\begin{lemma}\label{L:uniform-measurement}
  Let $U_0, \ldots,U_{m+ks}$ be the quantum gates of the coherent non-transversal
  encoding of a $k$-qubit gate $G$. Let us also assume that 
  for any $k$-qubit state $\ket{\psi}$, there are states $\ket{\psi_x}$, $x \in \{0,1\}^k$,
            such that
            \begin{equation*} 
                U_0 \ket{\psi} \ket{\Magic_G}
                    = \sum_{x \in \{0,1\}^k} \frac{1}{\sqrt{2^k}} \ket{x}
                    \ket{\psi_x}.
            \end{equation*}
  It follows that
            \begin{equation*} 
              U_{ks}\cdots U_0\ket{\phi}    = \sum_{x \in \{0,1\}^k} \frac{1}{\sqrt{2^k}} \ket{x}^{\otimes s+1} \ket{\psi_x}
              \quad \text{ and }  \quad 
               U\ket{\phi}
                    = \sum_{x \in \{0,1\}^k} \frac{1}{\sqrt{2^k}}
                    \ket{x}^{\otimes s+1} G\ket{\psi}.
            \end{equation*}
\end{lemma}

We are now ready to prove that the encoding of the gadget depicted in
\Cref{fig:coherent_toffoli} is simulatable.

\begin{proposition}\label{P:toffolisimulation}
  Let $U_0, \ldots,U_{m+ks}$ be the quantum gates of the coherent non-transversal
  encoding of a $k$-qubit gate $G$.
  Let also $\mcS$ be a $[[n,k,d]]$ stabilizer code that is order-consistent with
  $U_i$, $0 \leq i \leq m+ks$.
  If
  for any $k$-qubit state $\ket{\psi}$, there are states $\ket{\psi_x}$, $x \in \{0,1\}^k$,
            such that
            \begin{equation}\label{eq:random-output} 
                U_0 \ket{\psi} \ket{\Magic_G}
                    = \sum_{x \in \{0,1\}^k} \frac{1}{\sqrt{2^k}} \ket{x}
                    \ket{\psi_x},
            \end{equation}
  then the coherent non-transversal encoding of $G$ is $s$-simulatable
    for any $s < d-k$.
\end{proposition}
\begin{proof}
    Let $\rho$ be a the logical $k$-qubit state on whose encoding we want to apply the
    logical gate $G$ and let  $\sigma = 
\ket{0}\bra{0})^{\otimes ks} \otimes \rho \otimes
            \ket{\Magic_G}\bra{\Magic_G}$
    Suppose that each $U_i$ above is encoded by operations $O_{r_i},\ldots,O_{r_{i+1}-1}$,
    where $1 = r_0 < r_1 < \ldots < r_{ks+m}=\ell+1$.  In order to show that the
    code is simulatable, we need to show how to compute
    \begin{align}
      &\tr_S( O_{t} \cdots O_1 \Enc\left(
        \sigma
            \right) (O_{t} \cdots O_1)^\dagger )  \nonumber \\ 
            &=
            \tr_S( (O_{t} \cdots O_{r_i+1}) \Enc\left( U_{i-1} \cdots U_0 
            \sigma
                (U_{i-1} \cdots U_0)^\dagger \right) (O_{t} \cdots
                O_{r_i+1})^\dagger),
\label{Eq:partialtoffoli}
    \end{align}
    for $r_i \leq t \leq
    r_{i+1}$ and $0 \leq i \leq ks + m$.

    By Proposition \ref{P:transversalsim}, if $0 \leq i \leq ks$, we can compute the partial trace in
    \Cref{Eq:partialtoffoli}
    for any $|S| < d$ without knowledge of $\rho$ in polynomial time in
    $|S|$ and $(a_1,...,a_k)$, since $U_i$ can be
    applied transversally  .

    It remains to compute the partial trace in Equation
    \eqref{Eq:partialtoffoli} when $ks < i \leq ks + m$. 
    From \Cref{eq:random-output} and \Cref{L:uniform-measurement}, we have
    \begin{align*}
        \tr_{\overline{S}}& \left(O_t \cdots O_1 \Enc(\sigma) (O_t \cdots O_1)^{\dagger} \right) \\
            & = \sum_{x \in \{0,1\}^k} \frac{1}{2^k} \tr_{\overline{S}} \left( O_t \cdots O_{r_{ks}} 
                    \Enc\left((\ket{x}\bra{x})^{\otimes s+1} \otimes \rho_x \right) 
                    (O_t \cdots O_{r_{ks}})^\dagger \right).
    \end{align*}
    for some collection of states $\rho_x$ and we implicitly used the fact that
    since  $|S| \leq s$ and there are $s+1$ registers 
    containing the encoding of $x$, then there must be a register with an encoding of
    $x$ that does not overlap with $S$.

    For a fixed $x$, 
    we can continue to simplify, writing 
    \begin{align*}
      &\tr_{\overline{S}} \left( O_t \cdots O_{r_{ks}} 
                    \Enc\left((\ket{x}\bra{x})^{\otimes s+1} \otimes \rho_x \right) 
                    (O_t \cdots O_{r_{ks}})^\dagger \right) \\
                & = \tr_{\overline{S}} \left( O_t \cdots O_{r_{i}} 
                    \Enc\left(U_{r_{i-1}} \cdots U_{r_{ks}}(\ket{x}\bra{x})^{\otimes s+1} 
                    \otimes \rho_x (U_{r_{i-1}} \cdots U_{r_{ks}})^\dagger\right) 
                    (O_t \cdots O_{r_{i}})^\dagger \right) \\
                & = \tr_{\overline{S}} \left( O_t \cdots O_{r_{i}} 
                    \Enc((\ket{x}\bra{x})^{\otimes s+1} 
                    \otimes \rho'_x) (O_t \cdots O_{r_{i}})^\dagger \right)
    \end{align*}
    for some state $\rho'_x$, where the equalities follow because $\mcS$ is
    order-consistent with all $V_i$ and  $U_{r_{ks}}, \ldots, U_{r_{i-1}}$ are
    only control gates on the first $k(s+1)$ logical qubits.
    
    Thus we only need to compute
    \begin{equation*}
        \tr_{\overline{S}} \left(O_t \cdots O_{r_i}) \Enc((\ket{x}
        \bra{x})^{\otimes s+1} \otimes \rho'_x)(O_t \cdots O_{r_i})^{\dagger}
            \right)
    \end{equation*}
    for every (known) $x \in\{0,1\}^k$ and (unknown) states $\rho'_x$, and this
    can be done efficiently from \Cref{L:controlled-transversal} if $|S| \leq s
    < d - k$. In this case, 
    \begin{equation*}
        \tr_{\overline{S}} \left(O_t \cdots O_1 \Enc((\ket{0}\bra{0})^{\otimes ks} \otimes \rho
            \otimes \ket{\Magic_G}\bra{\Magic_G}) (O_t \cdots O_1)^{\dagger} \right) 
    \end{equation*}
    is a linear combination of a constant number of matrices which can be computed without knowledge
    of $\rho$, in time polynomial in $S$ and $2^k$. Hence the encoding
    is $s$-simulatable. 
\end{proof}

\subsection{Explicit construction}
\label{sec:explicit-simulatable}

In this section, we show that the family of concatenated
    Steane codes is simulatable. We start by defining the
    (concatenated) Steane code in \Cref{sec:concatenated-steane} and then in
    \Cref{sec:steane_simulatable} we argue that it is simulatable, proving
    \Cref{T:simulatable}.

\subsubsection{Steane code}
\label{sec:concatenated-steane}

The Steane code is the $[[7,1,3]]$ code stabilized by the following generators:

\begin{table}[H]
\centering
\begin{tabular}{|c c c c c c c|}
\hline
I &  I &  I & X &  X &  X &  X \\ \hline
I &  X &  X & I &  I &  X &  X \\ \hline
X &  I &  X & I &  X &  I &  X \\ \hline
I &  I &  I & Z &  Z &  Z &  Z \\ \hline
I &  Z &  Z & I &  I &  Z &  Z \\ \hline
Z &  I &  Z & I &  Z &  I &  Z \\ \hline
\end{tabular}
  \label{tab:stabilizer-seven-encoding}
  \caption{Generators of the Steane code}
\end{table}

In particular, the encoding of the basis states are then
\begin{multline}\label{Eq:codeword0}
  \Enc(\ket{0}) = \frac{1}{2\sqrt{2}}(
     \ket{0000000} + \ket{1010101} + \ket{0110011} + \ket{1100110} \\
    +\ket{0001111}+ \ket{1011010} + \ket{0111100} + \ket{1101001}),
\end{multline}
and
\begin{multline}\label{Eq:codeword1}
  \Enc(\ket{1}) =  \frac{1}{2\sqrt{2}}(
     \ket{1111111} + \ket{0101010} + \ket{1001100} + \ket{0011001} \\
    +\ket{1110000}+ \ket{0100101} + \ket{1000011} + \ket{0010110}),
\end{multline}
and the encoding of an arbitrary one-qubit pure state $\ket{\psi}$ and
one-qubit mixed state $\rho$ are defined by linearity.

We define the concatenated Steane codes by letting $\Steane_1 = \Steane$ denote
the $[[7,1,3]]$ Steane code, and
setting $\Steane_K = \Steane_{K-1} \circ \Steane$, a $[[7^K,1,3^K]]$ code.
We notice that the concatenated Steane code can also be expressed in the stabilizer formalism.
Suppose we have constructed a stabilizer code $\mcS_j$ for $\Steane_{j}$ for all $1
\leq j \leq K-1$. Let $g_1,\ldots,g_6 \in \mcS_1$ be the stabilizers for $\Steane$, and let $h_1,\ldots,h_{m-1}$ be a
minimal generating set for $\mcS_{K-1}$, where $m=7^{K-1}$. To encode the $m$ physical
qubits of $\Steane_{K-1}$ using $\Steane$, we form the stabilizer $\mcS_1^{\otimes m}$.
Define a homomorphism $\overline{\Delta} : \mcP_{m} \arr \mcP_{7m}$ by sending
$\Delta_{1,m}^i(W) \mapsto \Delta_{7,7m}^i(WWWWWWW)$ for each $W \in \mcP_1$
(in other words, we replace $W$ in the $i$th position with $WWWWWWW$, so $Z$
gets replaced with $\overline{Z}$, etc.).
The operator $\overline{\Delta}(w)$ gives an encoding of $w$ in $\Steane$, so
we can express $\Steane_K$ in the stabilizer formalism by taking the
stabilizer code $\mcS_K$ with minimal generating set
\begin{equation*}
    \{\Delta_{7,7m}^i(g_j), 1 \leq i \leq m, 1 \leq i \leq 6\} \cup \{\overline{\Delta}(h_j) :
        1 \leq j \leq m-1\}.
\end{equation*}

$\Steane_K$
contains stabilizers of weight $4$ for all $K \geq 1$, so while $\Steane$ is
nondegenerate, $\Steane_K$ is degenerate for $K \geq 2$.

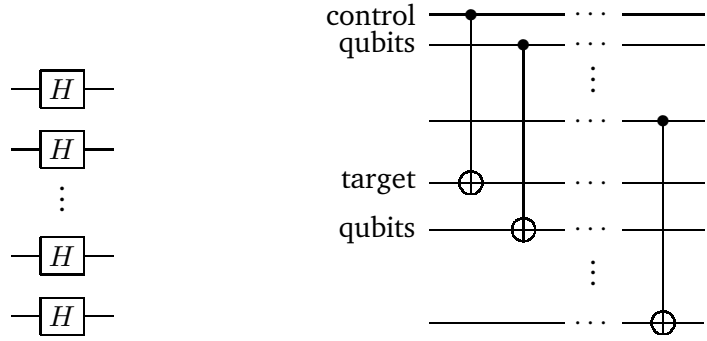
\begin{figure}[t]
    \centering
    \begin{subfigure}[b]{0.4\textwidth}
        \begin{equation*}
        \Qcircuit @C=1em @R=.8em {
            & \gate{H} & \qw & \\
            & \gate{H} & \qw & \\
            & \vdots & & \\
            & & & \\
            & \gate{H} & \qw & \\
            & \gate{H} & \qw & \\
        }
        \end{equation*}
        \caption{$H$ on encoded single-qubit state}
    \end{subfigure}
    \begin{subfigure}[b]{0.4\textwidth}
        \begin{equation*}
        \Qcircuit @C=1em @R=.8em {
            \lstick{\text{control}} & \ctrl{6} & \qw      & \qw & & \lstick{\cdots} & \qw      & \qw & \\ 
            \lstick{\text{qubits}}  & \qw      & \ctrl{6} & \qw & & \lstick{\cdots} & \qw      & \qw & \\
                                    &          &          &     & \vdots &          &          &     & \\
            && &&&&&& \\
                                    & \qw      & \qw      & \qw & & \lstick{\cdots} & \ctrl{6} & \qw & \\ 
            && &&&&&& \\
            \lstick{\text{target}}  &   \targ  & \qw      & \qw & & \lstick{\cdots} & \qw      & \qw & \\
            \lstick{\text{qubits}}  & \qw      & \targ    & \qw & & \lstick{\cdots} & \qw      & \qw & \\
                                    &          &          &     & \vdots &          &          &     & \\
            && &&&&&& \\
                                    &  \qw     & \qw      & \qw & & \lstick{\cdots} & \targ    & \qw & \\ 
        }
        \end{equation*}
        \caption{$\Lambda(X)$ on encoded two-qubit state}
    \end{subfigure}
    \caption{$\Steane_K$ encodings of $H$ and $\Lambda(X)$}
    \label{F:transversalencoding}
\end{figure}
Both $H$ and $\Lambda(X)$ have transversal encodings in $\Steane_K$, that is,
encodings where $H$ (resp. $\Lambda(X)$) is applied to each physical qubit
(resp. each pair of physical qubits). More formally, $H(a)$ (the Hadamard gate
acting on the $a$th logical qubit) is encoded by a sequence of gates
$O_1(a),\ldots,O_\ell(a)$, where $\ell = 7^K$, and $O_i(a)$ is the Hadamard
gate acting on the $i$th physical qubit of the $a$th logical qubit.  The gate
$\Lambda(X)(\underline{a})$ is encoded as a sequence of gates
$O_1(\underline{a}),\ldots,O_{\ell}(\underline{a})$, where $\ell = 7^K$, and
$O_i(a_1,a_2)$ is the $\Lambda(X)$ gate with the control on the $i$th physical
qubit of the $a_1$th logical qubit, and the $X$ gate on the $i$th physical
qubit of the $a_2$th logical qubit. In both cases, no ancilla states are necessary.
These encodings are shown in Figure \ref{F:transversalencoding}.

In order to complete a universal gateset, we show now how to apply Toffoli gates
using Clifford operations and Toffoli magic states.
The Toffoli magic state is the three-qubit
state prepared by the following circuit:

$$
    
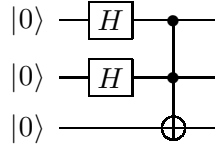
\Qcircuit @C=1em @R=.7em {
      \lstick{\ket{0}}  & \gate{H}   & \ctrl{1} & \qw\\
      \lstick{\ket{0}}  & \gate{H} & \ctrl{1} & \qw \\
      \lstick{\ket{0}} & \qw         & \targ & \qw
      }
$$
Explicitly, the state is \[\ket{\Toffoli} = \frac{1}{2} \sum_{d,e}
\ket{d,e,de}.\]

\begin{figure}[H]
\centering
$
\Qcircuit @C=0.8em @R=0.8em {
& \qw & \targ & \qw & \qw & \qw & \measureD{Z} & \cw & \cw & \cw & \cw & \cw & \cw & \cctrl{4} & \cctrl{5}\\
& \qw & \qw & \targ & \qw & \qw & \measureD{Z} & \cw & \cw & \cw &\cw &  \cctrl{3} &\cctrl{4} \\
& \qw & \qw & \qw & \ctrl{4} & \qw & \measureD{X} & \cw & \cctrl{2} & \cctrl{4} \\
& & & & & & & & & & & & & & & & &\\
& & \ctrl{-4} & \qw & \qw & \qw & \qw & \qw & \ctrl{1} & \qw & \qw & \ctrl{2} & \qw & \gate{X} &\qw & \qw \\
&  & \qw & \ctrl{-4} & \qw & \qw & \qw &\qw & \gate{Z} & \qw &\qw &  \qw & \gate{X} & \qw & \ctrl{1} & \qw \\
&  & \qw & \qw & \targ & \qw & \qw & \qw & \qw & \gate{Z} & \qw & \targ & \qw &\qw  &\targ & \qw
\inputgroupv{5}{7}{.2em}{2.0em}{\ket{\Toffoli} \quad }
}
$
\caption{The Toffoli gadget}
\label{fig:toffoli}
\end{figure}
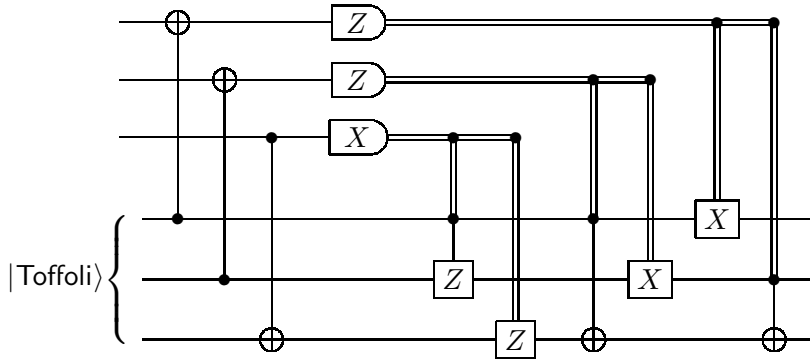

The Toffoli gate can be implemented using the Toffoli magic state, measurements
of $Z$ and $X$ observables, and classically controlled Clifford gates, using
the circuit shown in Figure \ref{fig:toffoli} (the double wires in this circuit
are classical bits).

\subsubsection{Concatenated Steane code is simulatable}
\label{sec:steane_simulatable}

We are now ready to prove \Cref{T:simulatable}.

\begin{proof}[Proof of Theorem \ref{T:simulatable}]
  Let $K > \log(s + 3)$, and let us consider the $\Steane_K$ code. We have that
  it is is a $[[7^K, 1, 3^K]]$ stabilizer code, and since $s$ is constant, so is
  $K$ and $m = 7^K$. 

As discussed previously, the gates $H$ and $\Lambda(X)$ can be implemented transversally in $\Steane_K$, and the Toffoli
gates can be applied using $\ket{\Toffoli}$ magic-state and the circuit in
\Cref{fig:toffoli}. We remark that in such a gadget,  we need to apply
controlled $H$ and $\Lambda(X)$ gates, and $\Steane_K$ is order-consistent with
such gates. Finally, it is not hard to calculate that the measurement result of
\Cref{fig:toffoli} are uniformly random bits for all input states $\rho$.  In
this case, all of the assumptions of
\Cref{P:transversalsim,P:toffolisimulation} are attained, and therefore
  $\Steane_K$ is a $s$-simulatable code.
\end{proof}

 \appendix
\section{Differences between~\cite{FitzsimonsJVY18} and the zero knowledge protocol $V_{ZK}$}
\label{sec:differences}

The structure and format of the protocol $V_{ZK}$ is essentially the same as the protocols that arise from protocol compression in~\cite{FitzsimonsJVY18}. However, we list the few differences and provide explanations for why the soundness of the protocol is unaffected by these changes.

\begin{enumerate}
\item The outer code for $V_{ZK}$ is a $4$-qubit error detecting code, instead of the $7$-qubit Steane code. As mentioned in Section~\ref{sec:prelim}, the soundness analysis of~\cite{FitzsimonsJVY18} only requires two properties from the outer code, and those properties are satisfied by the $4$-qubit error detecting code. This is why we are able to have fewer additional provers than in the protocol compression result of~\cite{FitzsimonsJVY18}.

\item The questions to the verifier players in $V_{ZK}$ are six tuples of commuting two-qubit Pauli observables, whereas in~\cite{FitzsimonsJVY18} they are triples. This is because the verifier players are in charge of measuring the clock qubits as well as the snapshot qubits, whereas in~\cite{FitzsimonsJVY18} the clock measurements were delegated to a different set of players. However the clock measurements are also just Pauli measurements, so we can simply merge the snapshot and clock measurements together.

\item In the protocol $V_{ZK}$, the referee may ask questions $QF_i$ or $AF_i$ to prover player $PP_i$, which does not occur in the compression protocol of~\cite{FitzsimonsJVY18}. The referee will ask these questions when it decides to check the propagation of the gate at time $t_\star(i) - 1$ (in which case it will send question $QF_i$ to $PP_i$), or at time $t_\star(i) + 1$ (in which case it will send question $AF_i$ to $PP_i$). The check performed by the referee is the identical to that when it tests the propagation of the prover gate $\star_i$. 

For completeness, in the honest strategy the prover player $PP_i$ measures a $\sigma_X$ on a designated ``question flag'' register (when asked question $QF_i$), or measures a $\sigma_X$ on a designated ``answer flag'' register (when asked question $AF_i$). 

Soundness is unaffected. If there was no valid history state before the addition of the $QF$ and $AF$ questions, then there is no valid history state with them.

\end{enumerate}

\bibliographystyle{plain}
\bibliography{references}

\end{document}